\documentclass[12pt, a4paper]{article}
\usepackage[margin=1.25in]{geometry}
\usepackage{xcolor}

\pdfoutput=1 

\newcommand*{\myBHat}{\hat{X}_b}

\colorlet{blau}{blue}
\colorlet{gruen}{green}
\colorlet{rot}{red}
\colorlet{orangsch}{orange}
\newcommand\myscale{0.75}

\usepackage{orcidlink}

\usepackage{amsthm}
\usepackage{amssymb}
\usepackage{amsmath}
\usepackage{mathtools}

\usepackage{tabularx}
\usepackage[export]{adjustbox}
\usepackage[shortlabels]{enumitem}

\usepackage{csquotes}

\newcolumntype{C}{>{\centering\arraybackslash}X}

\newcommand*{\mal}{\mathbin{\cdot}} 
 
\newcommand{\isdef}{\coloneqq} 
 
\newcommand*{\Def}[1]{\ensuremath{\operatorname{dom}(#1)}}        
\newcommand*{\Img}[1]{\ensuremath{\operatorname{img}(#1)}}        

\renewcommand*{\hom}[2]{\ensuremath{\operatorname{hom}(#1, #2)}}  
\newcommand*{\Hom}[2]{\ensuremath{\operatorname{Hom}(#1, #2)}}    

\newcommand*{\set}[1]{\{ #1 \}}

\newcommand{\mbracketleft}{\ensuremath{\{\!\!\{ }}
\newcommand{\mbracketright}{\ensuremath{\}\!\!\} }}

\newcommand*{\mset}[1]{\ensuremath{\mbracketleft #1 \mbracketright}}

\newcommand*{\nat}{\mathbb{N}}
\newcommand*{\natpos}{\nat_{\geq 1}}

\newcommand*{\union}{\cup}
\newcommand*{\bigunion}{\bigcup}
\newcommand*{\intersect}{\cap}

\newcommand*{\disunion}{\mathbin{\dot{\cup}}}

\newcommand*\V[1]{\ensuremath{V(#1)}}
\newcommand*\E[1]{\ensuremath{E(#1)}}
\newcommand*{\tupel}[1]{\ensuremath{\overline{#1}}}

\newcommand*{\Nachbarn}[2]{N_{#2}(#1)}

\makeatletter

\newcommand*{\inzidenz}[1]{\ensuremath{\begingroup\def\x{#1}\ifx\x\@empty{I}\else{I}_{\x}\fi\endgroup}}
\makeatother

\newcommand*{\rV}[1]{R(#1)}
\newcommand*{\bV}[1]{B(#1)}

\newcommand*{\reseatR}[3]{#1[#2{\to}#3]}
\newcommand*{\reseatB}[3]{#1\langle#2{\to}#3\rangle}

\newcommand*{\reclaimR}[2]{\reseatR{#1}{#2}{\bullet}}
\newcommand*{\reclaimB}[2]{\reseatB{#1}{#2}{\bullet}}
\newcommand*{\switch}[2]{#1[{\leadsto}#2]}

\newcommand*{\free}[1]{\operatorname{free}(#1)}

\newcommand*{\glue}[2]{(#1 \mathbin{\cdot} #2)}

\newcommand*{\VAR}{\textsf{Var}}

\newcommand*{\bigland}{\bigwedge}
\newcommand*{\biglor}{\bigvee}
\newcommand*{\limplies}{\rightarrow}

\newcommand*{\existsi}[1][i]{\exists^{\scalebox{0.6}{$\geq$}#1}}
\newcommand*{\existsex}[1][i]{\exists^{\scalebox{0.6}{$=$}#1}}

\let\emptyset\varnothing

\let\phi\varphi

\theoremstyle{plain}
\newtheorem{lemma}{Lemma}[section]
\newtheorem*{lemma*}{Lemma}

\newtheorem*{proposition*}{Proposition}
\newtheorem{theorem}[lemma]{Theorem}
\newtheorem*{theorem*}{Theorem}

\theoremstyle{definition}
\newtheorem{definition}[lemma]{Definition}
\newtheorem*{definition*}{Definition}

\newtheorem*{notation*}{Notation}
\newtheorem{example}[lemma]{Example}

\newcounter{HDCounter}
\newcounter{KGuardedCounter} 
\newcounter{LogicCounter} 
\newcounter{InductionConuter} 

\definecolor{cblue}{HTML}{007582}
\definecolor{cred}{HTML}{FF0000}

\usetikzlibrary{calc}
\usetikzlibrary{positioning,decorations.pathmorphing,shapes.geometric}

\tikzstyle{regular}=[circle, inner sep = .5pt]
\tikzstyle{edge}=[dotted]
\tikzstyle{guard}=[edge, very thick]
\tikzstyle{blue}=[regular, text=cblue]
\tikzstyle{red}=[regular, text=cred]

\usepackage{stmaryrd}
\usepackage{xspace}

\newenvironment{mea}{\begin{enumerate}[(a)]}{\end{enumerate}}

\newcommand{\NN}{\mathbb{N}}
\newcommand{\NNpos}{\NN_{\scriptscriptstyle \geq 1}}

\newcommand{\RR}{\mathbb{R}}

\renewcommand{\leq}{\leqslant}
\renewcommand{\geq}{\geqslant}

\renewcommand{\phi}{\varphi}
\renewcommand{\epsilon}{\varepsilon}

\newcommand{\nc}[1]{\newcommand{#1}}
\newcommand{\rnc}[1]{\renewcommand{#1}}

\newcommand{\deff}{\isdef}

\nc{\ov}[1]{\bar{#1}}

\nc{\openinterval}[2]{\ensuremath{\textup{rat}(#1,#2)}}

\renewcommand{\mid}{:}
\newcommand{\setc}[2]{\ensuremath{\set{#1 : #2}}}
\newcommand{\msetc}[2]{\ensuremath{\mset{#1:#2}}}
\nc{\bigset}[1]{\ensuremath{\big\{ #1 \big\}}}
\nc{\bigsetc}[2]{\bigset{#1 : #2}}
\nc{\setsize}[1]{\ensuremath{|#1|}}
\nc{\bigsetsize}[1]{\ensuremath{\big|#1\big|}}
\nc{\Setsize}[1]{\bigsetsize{#1}}

\nc{\und}{\ensuremath{\wedge}}
\nc{\Und}{\ensuremath{\bigwedge}}
\nc{\oder}{\ensuremath{\vee}}
\nc{\Oder}{\ensuremath{\bigvee}}
\nc{\nicht}{\ensuremath{\neg}}
\nc{\impl}{\ensuremath{\to}}
\nc{\lgdw}{\ensuremath{\leftrightarrow}}

\nc{\sem}[1]{\llbracket #1 \rrbracket}

\nc{\pto}{\ensuremath{\rightharpoonup}} 
\nc{\Dom}[1]{\Def{#1}}

\nc{\I}{\ensuremath{I}}
\nc{\J}{\ensuremath{J}}

\nc{\DBag}{\ensuremath{\textit{bag}}}
\nc{\DCover}{\ensuremath{\textit{cover}}}
\nc{\bVI}{\ensuremath{\bV{\I}}}
\nc{\rVI}{\ensuremath{\rV{\I}}}
\nc{\bVJ}{\ensuremath{\bV{\J}}}
\nc{\rVJ}{\ensuremath{\rV{\J}}}
\nc{\bVJn}{\ensuremath{\bV{\J_n}}}
\nc{\rVJn}{\ensuremath{\rV{\J_n}}}

\nc{\HOM}[2]{\ensuremath{\textsf{\textup{HOM}}_{#1}(#2)}}
\nc{\ClassC}{\mathcal{C}}

\nc{\Class}[1]{\ensuremath{\textsl{#1}}}

\newcommand*{\w}[1]{\ensuremath{\operatorname{w}(#1)}}

\newcommand*{\hw}[1]{\ensuremath{\operatorname{ghw}(#1)}}
\nc{\ghw}[1]{\hw{#1}}
\newcommand*{\ehw}[1]{\ensuremath{\operatorname{ehw}(#1)}}
\newcommand*{\hypertreewidth}[1]{\ensuremath{\operatorname{hw}(#1)}}

\newcommand*{\HD}[1]{\ensuremath{\textit{ghds}(#1)}}
\newcommand*{\Ehd}[1]{\ensuremath{\textit{ehds}(#1)}}

\nc{\TW}[1]{\ensuremath{\Class{TW}_{#1}}}

\nc{\GHW}[1]{\ensuremath{\Class{GHW}_{#1}}}
\nc{\SGHW}[1]{\ensuremath{\Class{sGHW}_{#1}}}
\nc{\GHWk}{\GHW{k}}
\nc{\BA}[1]{\ensuremath{\Class{BA}}}

\newcommand*{\IHW}[1][k]{\ensuremath{\Class{IGHW}_{#1}}}
\nc{\IGHW}[1]{\ensuremath{\Class{IGHW}_{#1}}}

\nc{\IEHW}[1]{\ensuremath{\Class{IEHW}_{#1}}}
\nc{\IEHWk}{\IEHW{k}}
\nc{\IBA}[1]{\ensuremath{\Class{IBA}}}

\nc{\AbbR}{\ensuremath{\mathcal{R}}}
\nc{\AbbB}{\ensuremath{\mathcal{B}}}
\nc{\AbbRB}{\ensuremath{\mathcal{RB}}}
\nc{\upH}{\ensuremath{\textup{H}}}

\nc{\VarB}{\ensuremath{\VAR_B}}
\nc{\VarR}{\ensuremath{\VAR_R}}
\nc{\Var}{\ensuremath{\VAR}}

\nc{\IntI}{\ensuremath{\mathcal{I}}}
\nc{\IntJ}{\ensuremath{\mathcal{J}}}

\nc{\LogicFont}[1]{\ensuremath{\textsf{\upshape{#1}}}}
\nc{\AGC}[1]{\ensuremath{\LogicFont{AGC}^{{#1}}}}
\nc{\GC}[1]{\ensuremath{\LogicFont{GC}^{{#1}}}}
\nc{\RGC}[1]{\ensuremath{\LogicFont{RGC}^{{#1}}}}
\nc{\NGC}[1]{\RGC{#1}}
\nc{\GCk}{\GC{k}}
\nc{\RGCk}{\RGC{k}}
\nc{\NGCk}{\NGC{k}}
\nc{\AGCk}{\AGC{k}}

\rnc{\L}{\GCk}
\nc{\AL}{\AGCk}

\nc{\LogGuard}[1]{\ensuremath{\Delta_{#1}}}

\nc{\VarFont}[1]{\ensuremath{\texttt{\textup{#1}}}}
\nc{\varv}{\VarFont{v}}
\nc{\vare}{\VarFont{e}}
\nc{\varx}{\VarFont{x}}

\nc{\equivGCk}{\ensuremath{\equiv_{\GCk}}}
\nc{\equivRGCk}{\ensuremath{\equiv_{\RGCk}}}

\nc{\GLI}[1]{\ensuremath{\Class{GLI}_{#1}}}
\nc{\GLIk}{\GLI{k}}

\nc{\QGLI}[1]{\ensuremath{\Class{QGLI}_{#1}}}
\nc{\QGLIk}{\QGLI{k}}

\nc{\kLI}{\ensuremath{L}}
\nc{\kLIfct}[1]{\ensuremath{M_{#1}}}
\nc{\kLIg}{\kLIfct{g}}
\nc{\kLIf}{\kLIfct{f}}

\nc{\myR}{\ensuremath{{X_r}}}
\nc{\myRStrich}{\ensuremath{{X'_r}}}
\nc{\myB}{\ensuremath{{X_b}}}
\nc{\myBStrich}{\ensuremath{{X'_b}}}
\nc{\myBTilde}{\ensuremath{{\tilde{X}_b}}}

\nc{\isom}{\ensuremath{\cong}}

\nc{\undefined}{\ensuremath{\bot}}

\nc{\assig}[1]{\ensuremath{{\hat{\I}_{#1}}}}

\nc{\domb}[1]{\ensuremath{\textup{db}_{#1}}}
\nc{\dombQ}{\domb{Q}}
\nc{\domr}[1]{\ensuremath{\textup{dr}_{#1}}}
\nc{\domrQ}{\domr{Q}}

\nc{\form}[2]{\ensuremath{\phi_{#1,#2}}}
\nc{\myInt}[1]{\ensuremath{\mathcal{I}_{#1}}}
\nc{\mySeg}[1]{\ensuremath{\textit{Seg}(#1)}}

\nc{\newtransition}{\text{transition}\xspace}

\nc{\indfreeB}[1]{\ensuremath{\textit{ifree}_{B}(#1)}}
\nc{\indfreeR}[1]{\ensuremath{\textit{ifree}_{R}(#1)}}

\nc{\quanth}[2]{\ensuremath{{Q}_{#1,#2}}}

\nc{\mysum}[1]{\ensuremath{\textit{num}(#1)}} 

\rnc{\enquote}[1]{``#1''}

\nc{\Deins}{\ensuremath{D^{1}}}
\nc{\Jeins}{\ensuremath{J^{1}}}
\nc{\Teins}{\ensuremath{T^{1}}}
\nc{\bageins}{\ensuremath{\DBag^{1}}}
\nc{\covereins}{\ensuremath{\DCover^{1}}}

\nc{\Dzwei}{\ensuremath{D^{2}}}
\nc{\Jzwei}{\ensuremath{J^{2}}}
\nc{\Tzwei}{\ensuremath{T^{2}}}
\nc{\bagzwei}{\ensuremath{\DBag^{2}}}
\nc{\coverzwei}{\ensuremath{\DCover^{2}}}

\tikzset{
	pics/hyperedge/.style ={
		/tikz/shift={(-35,-33.3)},
		/tikz/scale={.333},
		/tikz/x={.75pt},
		/tikz/y={-.75pt},
		/tikz/line width={.75pt},
		code = {
		\draw   
		(81.13,90.44) .. controls (110.13,61.69) and (140.63,-19.31) .. 
		(181.13,20.44) .. controls (221.63,60.19) and (150.13,91.19) .. 
		(121.13,120.44) .. controls (92.13,149.69) and (61.63,230.19) .. 
		(21.13,190.44) .. controls (-19.37,150.69) and (52.13,119.19) .. 
		(81.13,90.44) -- cycle ;
		}
	},
	pics/hypergraph/.style = {
		/tikz/x={.75pt},
		/tikz/y={.75pt},
		/tikz/scale={1.25},
		/tikz/line width={.75pt},
		code = {
		\draw  [dashed] 
		(30,10) .. controls (48.2,8.6) and (115.8,-2.2) .. 
		(140,20) .. controls (164.2,42.2) and (119,48.2) .. 
		(130,90) .. controls (141,131.8) and (31.8,116.2) .. 
		(30,90) .. controls (28.2,63.8) and (13.8,55.4) .. 
		(10,40) .. controls (6.2,24.6) and (11.8,11.4) .. 
		(30,10) -- cycle ;
		}
	}
}

\newcommand{\convexpath}[2]{
[
    create hullnodes/.code={
        \global\edef\namelist{#1}
        \foreach [count=\counter] \nodename in \namelist {
            \global\edef\numberofnodes{\counter}
            \node at (\nodename) [draw=none,name=hullnode\counter] {};
        }
        \node at (hullnode\numberofnodes) [name=hullnode0,draw=none] {};
        \pgfmathtruncatemacro\lastnumber{\numberofnodes+1}
        \node at (hullnode1) [name=hullnode\lastnumber,draw=none] {};
    },
    create hullnodes
]
($(hullnode1)!#2!-90:(hullnode0)$)
\foreach [
    evaluate=\currentnode as \previousnode using \currentnode-1,
    evaluate=\currentnode as \nextnode using \currentnode+1
    ] \currentnode in {1,...,\numberofnodes} {
-- ($(hullnode\currentnode)!#2!-90:(hullnode\previousnode)$)
  let \p1 = ($(hullnode\currentnode)!#2!-90:(hullnode\previousnode) - (hullnode\currentnode)$),
    \n1 = {atan2(\y1,\x1)},
    \p2 = ($(hullnode\currentnode)!#2!90:(hullnode\nextnode) - (hullnode\currentnode)$),
    \n2 = {atan2(\y2,\x2)},
    \n{delta} = {-Mod(\n1-\n2,360)}
  in
    {arc [start angle=\n1, delta angle=\n{delta}, radius=#2]}
}
-- cycle
}

\author{
	Benjamin Scheidt\,\orcidlink{0000-0003-2379-3675} \\
	benjamin.scheidt@hu-berlin.de \\
	Humboldt-Universität zu Berlin
	\and Nicole Schweikardt\,\orcidlink{0000-0001-5705-1675} \\
	schweikn@hu-berlin.de \\
	Humboldt-Universität zu Berlin
}

\title{Counting Homomorphisms from Hypergraphs of Bounded Generalised
  Hypertree Width:\newline A Logical Characterisation\footnote{This is
  the extended version of the conference contribution \cite{ScheidtSchweikardtMFCS23}.}}

\begin{document}
\maketitle

\begin{abstract}
We introduce the 2-sorted counting logic $\GCk$ that expresses properties of hypergraphs. This
logic has available $k$ variables to address hyperedges, an unbounded number of
variables to address vertices of a hypergraph, and atomic formulas $E(e,v)$ to express that
vertex $v$ is contained in hyperedge~$e$.

We show that two hypergraphs $H,H'$ satisfy the same sentences of the logic $\GCk$ if, and only if,
they are homomorphism indistinguishable over the class
of hypergraphs of generalised
hypertree width at most $k$. Here, $H,H'$ are called homomorphism indistinguishable over a class
$\ClassC$ if for every hypergraph $G\in\ClassC$ the number of homomorphisms from $G$ to $H$
equals the number of homomorphisms from $G$ to $H'$.

This result can be viewed as a lifting (from graphs to hypergraphs) of a result
by Dvořák  (2010)
stating that any two (undirected, simple, finite) graphs $H,H'$ are
indistinguishable by the $k{+}1$-variable counting logic $C^{k+1}$ if, and only
if, they are homomorphism indistinguishable over the class of 
graphs of tree-width
at most $k$.
\end{abstract}

\clearpage
\tableofcontents

\clearpage

\section{Introduction}
\label{sec:introduction}

Counting homomorphisms from a given class $\ClassC$ of graphs induces a similarity measure between graphs:
Consider an arbitrary graph $H$.
The results of the homomorphism counts for all $G\in\ClassC$ in $H$ can be represented by
a mapping (or, \enquote{vector}) $\HOM{\ClassC}{H}$ that associates
with every $G\in\ClassC$ the number $\hom{G}{H}$ of homomorphisms from
$G$ to $H$. 
A similarity measure for the mappings $\HOM{\ClassC}{H}$ and $\HOM{\ClassC}{H'}$ can then
be viewed as a similarity measure of two given graphs $H$ and $H'$.
An overview of this approach, its relations to \emph{graph neural
  networks}, and its usability as a similarity measure of graphs can
be found in Grohe's survey \cite{Grohe2020a}. 

Two graphs $H,H'$ are viewed as \enquote{equivalent} (or,
\emph{indistinguishable}) \emph{over $\ClassC$} if
$\HOM{\ClassC}{H}=\HOM{\ClassC}{H'}$, i.e., 
for every graph $G$ in $\ClassC$ the number of homomorphisms from $G$ to $H$ equals the number of homomorphisms from $G$ to $H'$. 

A classical result by Lovász  \cite{Lovasz1967} shows that two graphs
$H$ and $H'$ are indistinguishable over the class of \emph{all} graphs
if, and only if, they are isomorphic. 
This inspired a lot of research in recent years, examining the notion
of \emph{homomorphism indistinguishability over a class $\ClassC$} for
various classes $\ClassC$ 
\cite{Dvorak2010,DBLP:conf/icalp/DellGR18,Boeker-MasterThesis,DBLP:conf/mfcs/BokerCGR19,Boeker2019,Grohe2020,DBLP:conf/focs/MancinskaR20}.

In particular, Grohe \cite{Grohe2020} proved that two graphs are
homomorphism indistinguishable over the class of graphs of
\emph{tree-depth at most $k$} if, and only  if, they are
indistinguishable by sentences of first-order counting logic $C$ of
quantifier-rank at most $k$ ($C$ is the extension of first-order logic
with counting quantifiers of the form $\exists^{\geq n} x$  meaning
\enquote{there exist at least $n$ elements $x$}). 

A decade earlier,
Dvořák \cite{Dvorak2010} proved that two graphs are homomorphism
indistinguishable over the class of graphs of \emph{tree-width at most
  $k$} if, and only if, they are indistinguishable by sentences of the
$k{+}1$-variable fragment $C^{k+1}$ of $C$. 
From Cai, Fürer and Immerman \cite{Cai1992} we know that this precisely coincides with
indistinguishability by the $k$-dimensional Weisfeiler-Leman algorithm.

An obvious question is if and how these kinds of results can be lifted from graphs to
hypergraphs.

A first answer to this question was given by Böker in \cite{Boeker2019}: He introduces a
new version of a \emph{color refinement} algorithm on hypergraphs and proves that two
hypergraphs $H$ and $H'$ cannot be distinguished by this algorithm if, and only if, they are
homomorphism indistinguishable over the class
of \emph{Berge-acyclic} hypergraphs.
This can be viewed as a lifting ---  from graphs to hypergraphs --- of the result of
\cite{Dvorak2010,Cai1992} for the case $k=1$ (i.e., trees) to
\enquote{tree-like} hypergraphs. Note that there are different
concepts of ``tree-likeness'' for hypergraphs. Berge-acyclicity is a
rather restricted one; it is subsumed by the more general concept of
$\alpha$-acyclic hypergraphs, which coincides with the hypergraphs of
\emph{generalised hypertree width~$1$} (cf.,
\cite{Gottlob2002,Gottlob2003,Gottlob2016}). 

This paper gives a further answer to the above question: For arbitrary
$k\geq 1$ we consider the class $\GHWk$ of hypergraphs of generalised
hypertree width at most $k$. 
Our main result provides a logical characterisation of
homomorphism indistinguishability over the class $\GHWk$.
We introduce a new logic called $\GCk$ and show that two hypergraphs
are homomorphism indistinguishable over $\GHWk$ if, and only if, they
are indistinguishable by sentences of the logic $\GCk$. 

$\GCk$ is a 2-sorted counting logic for expressing properties of hypergraphs. It
has available $k$ ``blue'' variables to address edges, and an unbounded number of
``red'' variables to address vertices of a hypergraph, and atomic formulas
$E(e,v)$ to express that vertex $v$ is contained in edge $e$, as well
as atomic formulas $e=e'$ and $v=v'$ for expressing equality of edge
or vertex variables. 
Counting quantifiers are of the form $\exists^{\geq n}\tupel{z}$ where
$\tupel{z}=(z_1,\ldots,z_\ell)$ is either a tuple of edge variables or
a tuple of vertex variables; and their meaning is ``there exist at
least $n$ tuples $\tupel{z}$''. 
In the logic $\GCk$, each vertex variable $v$ has to be \emph{guarded}
by an edge variable $e$ and an atomic statement $E(e,v)$ (meaning that
vertex $v$ is included in edge $e$); 
the use of quantifiers is restricted in a way to ensure that guards are always present.
Our design of the logic $\GCk$ is heavily inspired by the \emph{guarded fragment of
first-order logic} (cf., \cite{DBLP:journals/jphil/AndrekaNB98,DBLP:journals/jsyml/Gradel99,Gottlob2003}).

Our main result can be viewed as a lifting --- from graphs to
hypergraphs --- of Dvořák's \cite{Dvorak2010} result: Dvořák proves
that two graphs are homomorphism indistinguishable over the class
$\TW{k}$ of 
graphs of tree-width $\leq k$ iff they are indistinguishable by the logic $C^{k+1}$.
We prove that two hypergraphs are homomorphism indistinguishable over the class $\GHWk$ of
hypergraphs of generalised hypertree width $\leq k$ iff they are indistinguishable by
the logic $\GCk$.

This is analogous (although not tightly related) to the following classical results:
Kolaitis and Vardi \cite{Kolaitis2000} proved that the conjunctive queries of tree-width $\leq k$ are precisely the queries expressible in
the $k{+}1$-variable fragment of a certain subclass $L$ of first-order
logic. Gottlob et al.\ \cite{Gottlob2003} proved that the conjunctive
queries of hypertree width $\leq k$ are precisely the ones expressible
in the $k$-guarded fragment of $L$. 
This is somehow parallel to our result generalising Dvořák's
characterisation; it is what initially gave us the confidence to work
on our hypothesis. 

The proof of our theorem is at its core very similar to Dvořák's proof --- but it is far from straightforward.
Before being able to follow along the lines of Dvořák's proof, we
first have to perform a number of reduction steps and build the
necessary machinery. 

The first step is to move over from 
homomorphisms on hypergraphs to homomorphisms on incidence graphs.
Fortunately, Böker \cite{Boeker2019} already implicitly achieved what
is needed in our setting. The result is: Two hypergraphs $H,H'$ are
homomorphism indistinguishable over the class $\GHWk$ iff their
incidence graphs $I,I'$ are homomorphism indistinguishable over the
class $\IGHW{k}$ of incidence graphs of generalised hypertree width
$\leq k$; 
see Section~\ref{sec:generalization_to_regular_homomorphisms}.

Next, for an inductive proof in the spirit of Dvořák, we would need an inductive characterisation of the class $\IGHW{k}$
in the spirit of \cite{Courcelle1993}. Unfortunately, generalised
hypertree decompositions seem to be unsuitable for such a
characterisation. That is why we work with severely restricted
decompositions that we call \emph{entangled hypertree decompositions}
(ehds, for short).  
In Section~\ref{sec:generalization_to_hw} we prove that homomorphism
indistinguishability over the class $\IGHW{k}$ coincides with
homomorphism indistinguishability over the class $\IEHWk$ of incidence
graphs of \emph{entangled hypertree width} $\leq k$. 
In our opinion this is interesting on its own, since the requirements
of ehds are quite harsh and $\IEHWk\varsubsetneq\IGHW{k}$ for arbitrarily large $k$.

In Section~\ref{sec:logic} we introduce the logic $\GCk$ and provide some
example formulas.
In Section~\ref{sec:normal_form} we establish a \emph{normal form}
for $\GCk$ called $\RGCk$ that
imposes certain restrictions on the way guards of red variables can
change between quantifications in a formula. These restrictions are
crucial for the proof of our main theorem.
The inductive characterisation of $\IEHWk$ follows
in Section~\ref{sec:recursive_def}, where we also provide the
machinery of \emph{quantum} incidence graphs 
as an analogue of the quantum
graphs used in Dvořák's proof, tailored towards our setting.
In Section~\ref{sec:main_theorem} we prove that
two incidence graphs $I,I'$ are indistinguishable by the logic $\RGCk$
(or, equivalently, $\GCk$)
if, and only if, they are homomorphism indistinguishable over the
class $\IEHWk$. This is achieved by two inductive proofs: We use our
inductive characterisation of $\IEHWk$ to show that for every
incidence graph $J$ in $\IEHWk$ and every $m \geq 0$
there exists an $\NGCk$-sentence that is satisfied by an incidence graph $I$ iff there are precisely $m$ homomorphisms from $J$ to $I$. For the opposite direction, we proceed by induction on the definition
of $\NGCk$ and  construct for every sentence $\chi$ in $\NGCk$ and
certain size parameters $m,d \geq 0$
a quantum incidence graph $Q$ in $\IEHWk$ satisfying the following:
for all  incidence graphs $I$ that match the size parameters $m,d$,
the number $\hom{Q}{I}$ of homomorphisms from $Q$ to $I$ is either 0
or 1, and it is 1 if and only if $I$ satisfies the sentence 
$\chi$.
Both proofs are quite intricate, and the details of the syntax
definition of $\NGCk$ had to be tweaked right in order to enable
proving \emph{both} directions. 

We wrap up in Section~\ref{sec:conclusion}: Plugging together the
results achieved in the previous sections yields our main theorem: Two
hypergraphs are homomorphism indistinguishable over the class $\GHWk$
of 
hypergraphs of generalised hypertree width $\leq k$ iff they are indistinguishable by
the logic $\GCk$.

While the main body of this paper contains concise proof sketches, the
formal proof details can be found in a comprehensive appendix.

\paragraph*{Acknowledgement:} We thank Isolde Adler for pointing us to
the results in \cite{AdlerDiss,Adler2004,Adler2007}, which led to Theorem~\ref{thm:IEHWsubsetIGHW}.

\section{Preliminaries}
\label{sec:preliminaries}

This section provides basic notions concerning hypergraphs, incidence graphs, hypertree decompositions, and homomorphisms.
We write $\RR$ for the set of
reals, $\NN$ for the set of
non-negative integers, and we let $\NNpos\deff\NN\setminus\set{0}$
and $[n]\deff\set{1,\ldots,n}$ for all $n\in\NNpos$.

\subsection{Hypergraphs}

The hypergraphs considered in this paper are generalisations of
ordinary undirected graphs, where each edge can consist of an
arbitrary number of vertices. For our proofs it will be necessary to
deal with hypergraphs in which the same edge can have multiple
occurrences. Furthermore, it will be convenient to assume that every
vertex belongs to at least one edge. This is provided by the following
definition that is basically taken from \cite{Boeker2019}. 

\begin{definition}
A \emph{hypergraph} $H \isdef (\V{H},\allowbreak \E{H},f_H)$ 
consists of disjoint finite sets $\V{H}$ of \emph{vertices} and
$\E{H}$ of edges, and an \emph{incidence function} $f_H$ associating
with every $e\in \E{H}$ the set $f_H(e) \subseteq \V{H}$ of vertices
incident with edge $e$, such that $\V{H} = \bigunion_{e \in \E{H}}
f_H(e)$. \\
A \emph{simple hypergraph} is a hypergraph $H$ where the function $f_H$ is injective.
\end{definition}

We can identify the edges of a hypergraph $H$ with the multiset $M_H\deff\msetc{f_H(e)}{e\in E(H)}$; the
number of occurrences
of a set $s\subseteq V(H)$ in this multiset then is the number of
occurrences of ``edge $s$'' in $H$. The \emph{simple} hypergraphs are
the hypergraphs in which every ``edge $s$'' has only one occurrence. 

Every hypergraph $H=(V(H),E(H),f_H)$ can be represented by an ordinary, bipartite graph $\inzidenz{H}$ in the following way:
The vertices $v\in V(H)$ occur as \emph{red} nodes of $\inzidenz{H}$, i.e., $\rV{\inzidenz{H}}\deff V(H)$.
The edges $e\in E(H)$ occur as \emph{blue} nodes of $\inzidenz{H}$, i.e.,
$\bV{\inzidenz{H}}\deff E(H)$.
And there is an edge from each blue node $e$ to all red nodes $v\in f_H(e)$. I.e.,
$\E{\inzidenz{H}}\isdef\{ (e, v) \in \bV{\inzidenz{H}} \times \rV{\inzidenz{H}} \mid v \in f_H(e) \}$.
The condition $\V{H} = \bigunion_{e \in \E{H}} f_H(e)$ implies that every red node is adjacent to at least one blue node.
It is straighforward to see that the mapping $H\mapsto \inzidenz{H}$
provides a bijection between the class of all hypergraphs and the
class of all \emph{incidence graphs}, where the notion of incidence
graphs is as follows.  

\begin{definition}\label{def:IncidenceGraph}
An \emph{incidence graph} $\I=(\rVI,\allowbreak \bVI,\allowbreak
\E{\I})$ consists of disjoint finite sets $\rVI$ and $\bVI$ of
\emph{red} nodes and \emph{blue} nodes, respectively, and a set of
edges $\E{\I} \subseteq \bVI \times\rVI$, such that each red node is
adjacent to at least one blue node. 
\end{definition}

As usual for graphs, the \emph{neighbourhood}
of a node $v$ is the set $N_\I(v)$ of all nodes adjacent to $v$.
Thus, if $\I$ is the incidence graph $\inzidenz{H}$ of a hypergraph $H$, then the neighbourhood
of every blue node $e$ is $N_\I(e)= f_H(e)$, i.e., the set of all vertices of $H$ that are incident with edge $e$. 
The neighbourhood
of every red node $v$ is $N_\I(v)=\setc{e\in E(H)}{v\in f_H(e)}$, i.e., the set of all edges of $H$ that are incident with
vertex $v$.

Two incidence graphs $\I,\I'$ are \emph{isomorphic} ($\I\isom\I'$, for short) if there exists
an \emph{isomorphism $\pi=(\pi_R,\pi_B)$ from $\I$ to $\I'$}, i.e, bijections $\pi_R:\rVI\to \rV{\I'}$ and $\pi_B:\bVI\to\bV{\I'}$ such that
for all $(e,v)\in\bVI\times \rVI$ we have: $(e,v)\in\E{\I}$ $\iff$ $(\pi_B(e),\pi_R(v))\in\E{\I'}$.
We sometimes drop the subscript and write $\pi(e)$ and $\pi(v)$ instead of $\pi_B(e)$ and $\pi_R(v)$. 

\subsection{Generalised Hypertree Decompositions}

We use the same notation as \cite{Gottlob2016} for decompositions of
hypergraphs, but we write $\DBag(t)$ and $\DCover(t)$ instead of
$\chi(t)$ and $\lambda(t)$, respectively, and we formalise them with
respect to incidence graphs rather than hypergraphs. 

\begin{definition}\label{def:hypertree_decomposition}
A \emph{complete generalised hypertree decomposition} (\emph{ghd}, for
short) of an incidence graph $\I$ is a tuple $D \isdef (T, \DBag,
\DCover)$, where $T \isdef (V(T), \E{T})$ is a 
finite undirected tree,       
and
$\DBag$ and $\DCover$ are mappings that associate with every tree-node
$t\in V(T)$ a set $\DBag(t)\subseteq \rVI$ of red nodes of $\I$ and a
set $\DCover(t)\subseteq\bVI$ of blue nodes of $\I$, having the
following properties: 
\begin{enumerate}
\item \label{def:hypertree_decomposition:completeness}
\emph{Completeness}: For every $e \in \bV{\inzidenz{}}$ there is a
tree-node $t \in V(T)$ with $\Nachbarn{e}{\I} \subseteq \DBag(t)$ and
$e \in \DCover(t)$. 

\item \label{def:hypertree_decomposition:connectedness_of_vertices}
\emph{Connectedness for red nodes}: For every $v \in \rV{\inzidenz{}}$
the subgraph $T_v$ of $T$ induced on $V_v \deff \setc{ t \in V(T)}{v
 \in \DBag(t) }$ is a tree. 

\item \label{def:hypertree_decomposition:covering}
\emph{Covering of Bags}: For every $t \in V(T)$ we have $\DBag(t)
\subseteq \bigunion_{e \in \DCover(t)} \Nachbarn{e}{\inzidenz{}}$. 
\setcounter{HDCounter}{\value{enumi}}

\end{enumerate}
The \emph{width $\w{D}$ of a ghd} $D$ is defined as the maximum number
of blue nodes in the cover of a tree-node, i.e., $\w{D} \isdef
\max\setc{ | \DCover(t) |}{t \in V(T)}$. 
We write $\HD{\inzidenz{}}$ to denote the class of
all ghds of an incidence graph $\I$.
\end{definition}

It is straightforward to see that this notion of a ghd of an incidence
graph $\I$  coincides the classical notion  (cf.,
\cite{Gottlob2002,Gottlob2016}) of a complete generalised hypertree
decomposition of a hypergraph $H$ where $\inzidenz{H}=\I$. 

\begin{definition}
The \emph{generalised hypertree width} of an incidence graph $\I$ is
$\hw{\I} \isdef \min\setc{\w{D}}{D \in \HD{\I}}$.
By $\IGHW{k}$ we denote the class of all incidence graphs of generalised hypertree width $\leq k$.
\end{definition}

It is straightforward to see that
$\hw{\inzidenz{H}}$ coincides with the classical notion (cf., \cite{Gottlob2016}) of generalised hypertree width of 
a hypergraph $H$, and $\IHW{}$ is the class of incidence graphs $\inzidenz{H}$ of all hypergraphs $H$ of generalised hypertree width $\leq k$.

For our proofs we need ghds with specific further properties, defined as follows;
we are not aware of any related work that studies this particular kind
of decompositions.\footnote{Hypertree decompositions satisfying
  condition~\ref{def:entangled_hypertree_decomposition:precision} (but
  not necessarily
  condition~\ref{def:entangled_hypertree_decomposition:connectedness_of_edges})
  of Definition~\ref{def:entangled_hypertree_decomposition} are known
  as \emph{strong}  decompositions \cite{Gottlob2003}.} 

\begin{definition}\label{def:entangled_hypertree_decomposition}
An \emph{entangled hypertree decomposition} (\emph{ehd}, for short) of
an incidence graph $\I$ is a ghd $D$ of $\I$ that additionally
satisfies the following requirements: 
\begin{enumerate}
\setcounter{enumi}{\value{HDCounter}}
\item \label{def:entangled_hypertree_decomposition:precision} 
\emph{Precise coverage of bags}: For all tree-nodes $t \in V(T)$ we
have \[\bigunion_{e \in \DCover(t)} \Nachbarn{e}{\I} = \DBag(t).\] 

\item \label{def:entangled_hypertree_decomposition:connectedness_of_edges}
\emph{Connectedness for blue nodes}: For every $e \in \bVI$ the
subgraph  $T_e$ of $T$ induced on  $V_e \deff \setc{ t \in V(T)}{e \in
  \DCover(t) }$ is a tree. 
\end{enumerate}
We write $\Ehd{\I}$ to denote the class of all ehds of an incidence graph $\I$.
\end{definition}

\begin{definition}\label{def:ehw}
The \emph{entangled hypertree width} of an incidence graph $\I$ is
$\ehw{\I} \isdef \min\setc {\w{D}}{D \in \Ehd{\I}}$.
For a hypergraph $H$ we let $\ehw{H}\deff \ehw{I_H}$.

By $\IEHW{k}$ we denote the class of all incidence graphs of entangled hypertree width $\leq k$.
\end{definition}

Clearly, $\hw{\I}\leq\ehw{\I}$ for every incidence graph $\I$. Hence, $\IEHW{k}\subseteq\IGHW{k}$ for all $k\in\NNpos$.
Applying results from \cite{AdlerDiss,Adler2004,Adler2007} shows that
there exist arbitrarily large $k$ such that
$\IEHW{k}$ is a strict subclass of $\IGHW{k}$.
More precisely:

\begin{theorem}\label{thm:IEHWsubsetIGHW}
$\IEHW{k}\subseteq \IGHW{k}$, for every $k\in\NNpos$.
Furthermore,   
$\IEHW{1}=\IGHW{1}$, but $\IEHW{k}\varsubsetneq\IGHW{k}$ for each
$k\in\set{2,3}$.
Moreover, for every $n\in\NN$ there exists a $k\in\NNpos$ such 
that $\IGHW{k}\not\subseteq\IEHW{k+n}$ (and hence, $\IEHW{k+n}\varsubsetneq\IGHW{k+n}$).
\end{theorem}
\begin{proof}
$\IEHW{k}\subseteq \IGHW{k}$ holds because every ehd also is a ghd.
$\IEHW{1}=\IGHW{1}$ holds because ghds of width 1 are known to be
equivalent to so-called \emph{join trees}, and these can easily be
translated into ehds
of width 1.
For the remaining statements, we use elaborate results from 
\cite{AdlerDiss,Adler2004,Adler2007} that relate the
\emph{hypertree width} $\hypertreewidth{H}$ (cf.,
\cite{Gottlob2002,Gottlob2003}) of a hypergraph to its 
\emph{generalised hypertree width} $\ghw{H}$: 

From \cite[Proposition~3.3.2]{AdlerDiss} (cf.\ also
\cite[Example~3]{Adler2007}) and
\cite[Claim~6.1]{Adler2004} we obtain for each $k\in\set{2,3}$ a
simple hypergraph $H_k$ such that $\ghw{H_k}=k$ and $\hypertreewidth{H_k}=k{+}1$.
Furthermore,
\cite[Fact~3.3.1]{AdlerDiss}
and \cite[Theorem~4.1]{Adler2004} 
provide for every $n\in\NNpos$ a simple
hypergraph $H^n$ such that $\hypertreewidth{H^n}=\ghw{H^n}{+}n$.\footnote{Note that the notions $c_H{\text -}\hypertreewidth{H}$ and $c_H{\text -}\ghw{H}$ in \cite{AdlerDiss} correspond to $\hypertreewidth{H}$ and $\ghw{H}$ for all hypergraphs $H$ according to \cite[Example 2.1.10]{AdlerDiss}.}

It is straightforward to verify that every ehd also is a
complete \emph{hypertree decomposition} in the sense of
\cite{Gottlob2002,Gottlob2003}. Consequently, for every hypergraph $H$
we have $\hypertreewidth{H}\leq \ehw{H}$.
Therefore, for each $k\in\set{2,3}$, the incidence graph of $H_k$
witnesses that $\IEHW{k}\varsubsetneq \IGHW{k}$.

To address the theorem's next statement, consider an arbitrary
$n\in\NN$. Let $H\deff H^{n+1}$ and let $k\deff \ghw{H}$. Then,
$\ehw{H}\geq \hypertreewidth{H}=k{+}n{+}1$. Thus, the incidence graph
of $H$ belongs to $\IGHW{k}$ but not to $\IEHW{k+n}$.
\end{proof}

\subsection{Homomorphisms}
We use the classical notions for hypergraphs and incidence graphs:

\begin{definition}\label{def:HypergraphHomom}
Let $F$ and $H$ be hypergraphs. A \emph{homomorphism} from $F$ to $H$ is a
pair $(h_V,h_E)$ of mappings $h_V: \V{F} \to \V{H}$ and $h_E:\E{F}\to \E{H}$
such that for every $e\in \E{F}$ we have
$f_H(h_E(e)) = \setc{h_V(v)}{v\in f_F(e)}$.
\end{definition}

By $\Hom{F}{H}$ we denote the set of all homomorphisms from $F$ to
$H$, and we let $\hom{F}{H} \isdef | \Hom{F}{H} |$ be the number of
homomorphisms from $F$ to $H$.

\begin{definition}\label{def:IncidenceGraphHomom}
Let $\J$ and $\I$ be incidence graphs. A \emph{homomorphism} from $\J$ to $\I$ is a pair
$h=(h_R,h_B)$ of mappings  
$h_R: \rVJ \to \rVI$ and $h_B: \bVJ \to \bVI$ such that for all $(e,v)\in \E{\J}$ we have 
$(h_B(e), h_R(v)) \in \E{\I}$.
\\
We sometimes drop the subscript and write $h(e)$ and $h(v)$ instead of $h_B(e)$ and $h_R(v)$. 
\end{definition}

By $\Hom{\J}{\I}$ we denote the set of all homomorphisms from $\J$ to
$\I$, and we let $\hom{\J}{\I} \isdef | \Hom{\J}{\I} |$ be the number
of homomorphisms from $\J$ to $\I$. 

As pointed out in \cite{Boeker2019}, every homomorphism from a
hypergraph $F$ to a hypergraph $H$ also is a homomorphism from the
incidence graph $\inzidenz{F}$ to the incidence graph $\inzidenz{H}$;
but there exist homomorphisms from $\inzidenz{F}$ to $\inzidenz{H}$
that do not correspond to any homomorphism from $F$ to $H$. 
In fact, every homomorphism $(h_R,h_B)$ from $\inzidenz{F}$ to
$\inzidenz{H}$ is a pair of mappings $(h_V,h_E)\deff (h_R,h_B)$ with
$h_V:\V{F}\to\V{H}$ and $h_E:\E{F}\to\E{H}$ such that for every $e\in
\E{F}$ we have 
$f_H(h_E(e)) \supseteq \setc{h_V(v)}{v\in f_F(e)}$ --- i.e., the
condition ``$=$'' of Definition~\ref{def:HypergraphHomom} is relaxed
into the condition ``$\supseteq$''. 

\section{Homomorphism Indistinguishability}
\label{sec:generalization_to_regular_homomorphisms}

Let $(B,B',\ClassC)$ be either two incidence graphs and a class of incidence graphs or two hypergraphs and a class of hypergraphs.
By $\HOM{\ClassC}{B}$ we denote the function $\alpha:\ClassC\to\NN$ that associates with every $A\in\ClassC$ the number $\hom{A}{B}$ of homomorphisms from $A$ to $B$.
We say that $B$ and $B'$ are \emph{homomorphism indistinguishable over $\ClassC$} if $\HOM{\ClassC}{B}=\HOM{\ClassC}{B'}$.
Note that $\HOM{\ClassC}{B}\neq\HOM{\ClassC}{B'}$ means that there
exists an $A\in\ClassC$ that \emph{distinguishes} between $B$ and $B'$
in the sense that $\hom{A}{B}\neq\hom{A}{B'}$, i.e., the number of
homomorphisms from $A$ to $B$ is different from the number of
homomorphisms from $A$ to $B'$. 

Recall from Section~\ref{sec:preliminaries} that $\IGHW{k}$ is the class of all incidence graphs of generalised hypertree width $\leq k$.
We write $\GHW{k}$ for the class of all hypergraphs of generalised
hypertree width $\leq k$ (i.e., all hypergraphs $H$ for which
$\inzidenz{H}\in\IGHW{k}$), and $\SGHW{k}$ for the subclass consisting
of all \emph{simple} hypergraphs (i.e., hypergraphs where each edge
has multiplicity 1) in 
$\GHW{k}$.
Proof details for the following theorem can be found in Appendix~\ref{appendix:generalization_to_regular_homomorphisms}.

\begin{theorem}[implicit in \cite{Boeker2019}]
\label{thm:Boeker}
Let $H,H'$ be hypergraphs.
\begin{mea}
\item\label{item:a:thm:Boeker}
  If $H$ and $H'$ are simple hypergraphs, then\\
  $\HOM{\GHW{k}}{H}=\HOM{\GHW{k}}{H'}$ \ \ $\iff$ \ \
  $\HOM{\SGHW{k}}{H}=\HOM{\SGHW{k}}{H'}$.

\item\label{item:b:thm:Boeker}
  $\HOM{\GHW{k}}{H}=\HOM{\GHW{k}}{H'}$ \ \ $\iff$ \ \
  $\HOM{\IGHW{k}}{\inzidenz{H}}=\HOM{\IGHW{k}}{\inzidenz{H'}}$.
\end{mea}  
\end{theorem}

\noindent
Böker \cite{Boeker2019} proved the analogous statement for
$\BA{k},\IBA{k}$ instead of $\GHW{k},\IGHW{k}$, where $\BA{k}$ is the
class of all Berge-acyclic hypergraphs and $\IBA{k}$ is the class of
all incidence graphs of hypergraphs in $\BA{k}$. 
Böker's proof, however, works for all classes $\Class{C}$ of
hypergraphs and the associated class $\Class{IC}$ of all incidence
graphs of hypergraphs in $\Class{C}$, provided that $\Class{C}$ is
closed under \emph{local merging} and \emph{pumping} (Böker calls this
\enquote{leaf adding}). 
A class $\Class{C}$ is closed under \emph{local merging}, if every
$F\in\Class{C}$ remains in $\Class{C}$ when merging two vertices
$u_1,u_2 \in \V{F}$ that are adjacent via a common edge $e \in \E{F}$
(i.e., $u_1,u_2 \in f_F(e)$), into a single new vertex $u$.  
A class $\Class{C}$ is closed under \emph{pumping}, if every
$F\in\Class{C}$ remains in $\Class{C}$ when inserting a newly created
vertex into an edge $e \in \E{F}$ (but note that we must not insert
vertices in the intersection of multiple edges).
It is easy to see that the class $\GHW{k}$ is closed under local merging and pumping.

\section{Relating $\IGHW{k}$ to $\IEHW{k}$}
\label{sec:generalization_to_hw}

Recall from Section~\ref{sec:preliminaries} that 
$\IEHW{k}\subseteq\IGHW{k}$, for the class $\IEHW{k}$ of
incidence graphs of entangled hypertree width $\leq k$ and the class $\IGHW{k}$
of incidence graphs of generalised hypertree width $\leq k$.
By Theorem~\ref{thm:IEHWsubsetIGHW} there exist arbitrarily large $k$ such that $\IEHW{k}$ is a strict
subclass of $\IGHW{k}$.
This section's main result is that, nevertheless:

\begin{theorem}\label{thm:EHWandIGHW}
For all incidence graphs $\I$ and $\I'$ we have\\
$\HOM{\IGHW{k}}{\I}=\HOM{\IGHW{k}}{\I'}$ \ \ $\iff$ \ \
$\HOM{\IEHW{k}}{\I}=\HOM{\IEHW{k}}{\I'}$.
\end{theorem}

The proof heavily relies on the following technical main lemma, which
uses the following notation:
For an arbitrary incidence graph $\J$, 
for $s\subseteq\rVJ$, and for $n\in\NN$ we write $\J+n{\cdot}s$ to
denote the incidence graph $\J'$ obtained from $\J$ by inserting $n$
new blue nodes $\hat{e}_1,\ldots,\hat{e}_n$ and edges $(\hat{e}_i,v)$ for all
$i\in[n]$ and all $v\in s$ --- i.e., $N_{\J'}(\hat{e}_i)=s$.
\begin{lemma}\label{lemma:adding_edges_keeps_distinguishing}
  Let $\J,\I,\I'$ be incidence graphs with
  $\hom{\J}{\I}\neq\hom{\J}{\I'}$, let $e\in\bVJ$, and let $s\subseteq
  N_{\J}(e)$.
  For every $m\in \NN$ there exists an $n\in\NN$ with
  $n\geq m$ such that $\J_n\deff \J+n{\cdot}s$ satisfies
  $\hom{\J_n}{\I}\neq\hom{\J_n}{\I'}$.
\end{lemma}
\begin{proof} 
Let $\J,\I,\I',e,s,m$ be given as in the lemma's assumption.
Let $[e]\deff \setc{e'\in\bVJ}{N_{\J}(e')=N_{\J}(e)}$.
Let $z\deff |[e]|$.

We partition the set $\Hom{\J}{\I}$ of all homomorphisms from
$\J$ to $\I$ as follows.
Let $\AbbR$ be the set of all mappings $r:\rVJ\to\rVI$, let
$\AbbB$ be the set of all mappings $b:\bVJ\setminus [e]\to\bVI$.
For all $(r,b)\in\AbbR\times\AbbB$ let
$\upH_{r,b}\deff \setc{(h_R,h_B)\in\Hom{\J}{\I}}{h_R=r\text{ and
}h_B(e')=b(e')\text{ for all }e'\in\bVJ\setminus [e]}$.
Let
$\AbbRB=\setc{(r,b)\in\AbbR\times\AbbB}{H_{r,b}\neq\emptyset}$.
Obviously,
the sets $\upH_{r,b}$ are pairwise disjoint, and
$\Hom{\J}{\I}=\bigcup_{(r,b)\in\AbbRB}\upH_{r,b}$.

Let $X_{r,b}$ be the set of all $\tilde{e}\in\bVI$ for which
there exist $(h_R,h_B)\in\upH_{r,b}$ and $e'\in[e]$ such that $h_B(e')=\tilde{e}$.
Let $x_{r,b}\deff |X_{r,b}|$. Clearly, $x_{r,b}\in\NNpos$.
\medskip

\noindent
\emph{Claim~1:} \ $|\upH_{r,b}|= (x_{r,b})^z$ \ for all
$(r,b)\in\AbbRB$.
\smallskip

 \noindent
 \emph{Proof of Claim~1:} \
Let $e'_1,\ldots,e'_z$ be a list of all elements in $[e]$.
For every tuple $t=(\tilde{e}_1,\ldots,\tilde{e}_z)\in (X_{r,b})^z$ let
$h_{b,t}:\bVJ\to\bVI$ be defined by $h_{b,t}(e'_i)=\tilde{e}_i$ for
all $i\in[z]$ and $h_{b,t}(e')=b(e')$ for all $e'\in\bVJ\setminus [e]$.
It is straightforward to verify that the mapping $\pi$
defined by $\pi(t)\deff (r,h_{b,t})$ for all $t\in(X_{r,b})^z$ is a
bijection from $(X_{r,b})^z$ to the set $\upH_{r,b}$.
\qed$_{\textit{Claim 1}}$

\medskip
	
\noindent
From Claim~1 we obtain: 
\begin{equation}\label{eq:hw:HomJI}
   \hom{\J}{\I} \ \ = \ \ \sum_{(r,b)\in\AbbRB} (x_{r,b})^z\;.
\end{equation}
Consider, for each $n\in\NN$, the incidence graph $\J_n\deff 
\J+n{\cdot}s$ obtained from $\J$ by inserting $n$ new blue nodes
$\hat{e}_1,\allowbreak\ldots,\allowbreak\hat{e}_n$ and according edges such that $N_{\J_n}(\hat{e}_i)=s$ for all $i\in[n]$.

Let $Y_{r,b}$ be the set of all $\tilde{e}\in\bVI$ for which
there exist $(h_R,h_B)\in\Hom{\J_1}{\I}$ such that $h_R=r$,
$h_B(\hat{e}_1)=\tilde{e}$ and $h_B(e')=b(e')$ for all
$e'\in\bVJ\setminus[e]$.
Note that $Y_{r,b}\supseteq X_{r,b}$. 
Let $y_{r,b}\deff |Y_{r,b}|$. Clearly, $y_{r,b}\geq x_{r,b}$.
\medskip

\noindent
\emph{Claim~2:} For all $n\in\NN$ we have:
\begin{equation}\label{eq:hw:HomJnI}
   \hom{\J_n}{\I} \ = \ \sum_{(r,b)\in\AbbRB}\ (x_{r,b})^z \cdot
   (y_{r,b})^n\;.
\end{equation}  
\smallskip

\noindent
\emph{Proof of Claim~2:} \
\noindent
Let $e'_1,\ldots,e'_z$ be a list of all elements in $[e]$.
For every tuple $t=(\tilde{e}_1,\ldots,\tilde{e}_z,e_1,\ldots,e_n)\in
(X_{r,b})^z\times (Y_{r,b})^n$ let
$h_{b,t}:\bVJn\to\bVI$ be defined by $h_{b,t}(e'_i)=\tilde{e}_i$ for
all $i\in[z]$, $h_{b,t}(\hat{e}_i)=e_i$ for all $i\in[n]$, and $h_{b,t}(e')=b(e')$ for all $e'\in\bVJ\setminus [e]$.
It is straightforward to verify that the mapping $\pi$
defined by $\pi(t)\deff (r,h_{b,t})$ for all $t\in (X_{r,b})^z\times
(Y_{r,b})^n$ is a bijection from $(X_{r,b})^z\times
(Y_{r,b})^n$ to the set $\upH_{r,b}^{(n)}$ of all
$(h_R,h_B)\in\Hom{\J_n}{\I}$ with $h_R=r$ and $h_B(e'')=b(e'')$ for
all $e''\in \bVJ\setminus[e]$. 

Thus, $|\upH_{r,b}^{(n)}|=(x_{r,b})^z\cdot (y_{r,b})^n$.
The claim follows by noting that the sets $\upH_{r,b}^{(n)}$ are
pairwise disjoint and
$\Hom{\J_n}{\I}=\bigcup_{(r,b)\in\AbbRB} \upH_{r,b}^{(n)}$.
 \qed$_{\textit{Claim 2}}$

\bigskip

\noindent	
By performing the analogous reasoning with $\I'$ instead of $\I$ we
obtain a set $\AbbRB'$ and positive integers $x'_{r',b'}$ and $y'_{r',b'}$ with
$y'_{r',b'}\geq x'_{r',b'}$ for all $(r',b')\in\AbbRB'$ such that
\begin{equation}\label{eq:hw:HomJIStrich}
   \hom{\J}{\I'} \ \ = \ \ \sum_{(r',b')\in\AbbRB'} (x'_{r',b'})^{z}
\end{equation}
and, for all $n\in\NN$, 
\begin{equation}\label{eq:hw:HomJnIStrich}
  \hom{\J_n}{\I'} \ = \ \sum_{(r',b')\in\AbbRB'}\ (x'_{r',b'})^z \cdot
  (y'_{r',b'})^n\;.
\end{equation}  
If there exists an $n\in\NN$ with $n\geq m$ such that
$\Hom{\J_n}{\I}\neq\Hom{\J_n}{\I'}$, we are done.
Assume for contradiction that no such $n$ exists. Then, for all
$n\in\NN$ with $n\geq m$ we have $\Hom{\J_n}{\I}=\Hom{\J_n}{\I'}$. By
\eqref{eq:hw:HomJnI} and \eqref{eq:hw:HomJnIStrich} this yields
\[
\sum_{(r,b)\in\AbbRB} (x_{r,b})^z \cdot (y_{r,b})^n
\ = \
\sum_{(r',b')\in\AbbRB'} (x'_{r',b'})^z \cdot (y'_{r',b'})^n
\]  
and hence
\[
\sum_{(r,b)\in\AbbRB}\!\!\! (x_{r,b})^z \cdot (y_{r,b})^n
\ -
\!\!\!\sum_{(r',b')\in\AbbRB'}\!\!\! (x'_{r',b'})^z \cdot (y'_{r',b'})^n
\ = \
0
\]  
for all $n\in\NN$ with $n\geq m$.

We view this as a system of equations, and we group the summands with
respect to the distinct values in
$Z\deff \setc{y_{r,b}}{(r,b)\in\AbbRB}\cup\setc{y'_{r',b'}}{(r',b')\in\AbbRB'}$.
To this end let $k\deff |Z|$ and let $z_1,\ldots,z_k$ be a list of all
elements in $Z$.
For every $i\in[k]$ let $M_i\deff \setc{(r,b)\in\AbbRB}{y_{r,b}=z_i}$ and
$M'_i\deff \setc{(r',b')\in\AbbRB'}{y'_{r',b'}=z_i}$, and let
\[
a_i \ \deff \
\sum_{(r,b)\in M_i}\!\!(x_{r,b})^z
\ -
\sum_{(r',b')\in M'_i}\!\!(x'_{r',b'})^z.
\]
The above system of equations yields that 
\begin{equation*}
  \sum_{i=1}^k \, a_i\cdot (z_i)^n \ \ = \ 0
  \qquad
  \text{for all $n\geq m$.}
\end{equation*}  
By assumption we know that $\hom{\J}{\I}\neq\hom{\J}{\I'}$, and hence
\eqref{eq:hw:HomJI} and \eqref{eq:hw:HomJIStrich} yield that
\[
  \sum_{(r,b)\in\AbbRB} (x_{r,b})^z
  \ \ \neq \ \
  \sum_{(r',b')\in\AbbRB'} (x'_{r',b'})^z\;.
\]  
Therefore,
\begin{equation}\label{eq:hw:aiNotNull}
  a_i\neq 0 \quad \text{for at least one $i\in[k]$}.
\end{equation}  
The system of equations
$\Big(\,{\displaystyle\sum_{i=1}^k} a_i{\cdot}\, (z_i)^n  = 0  \Big)_{n\in\set{m,\ldots,m+k-1}}$
can be expressed as
\[
\underbrace{
\left(    
\begin{array}{lll}
(z_1)^{m} & \dots & (z_k)^{m} \\
(z_1)^{m+1} & \dots & (z_k)^{m+1} \\
\ \ \vdots & \ddots & \ \ \vdots \\
(z_1)^{m+k-1} & \dots & (z_k)^{m+k-1}
\end{array}
\right)
}_{\textstyle =: V}
\cdot
\underbrace{
\begin{pmatrix}
a_1 \\ a_2 \\ \vdots \\ a_k
\end{pmatrix}
}_{\textstyle =:\ov{a}}
=
\underbrace{
\begin{pmatrix}
0 \\ 0 \\ \vdots \\ 0
\end{pmatrix}
}_{\textstyle =:\ov{0}}
\]
But $V$ is a \emph{generalised Vandermonde matrix} on pairwise
distinct non-zero values $z_1,\ldots,z_k$, and is known to be an
invertible matrix \cite{DeAlba2007}\footnote{to see this, combine
``application 2'' with ``facts 8'' and ``18'' of \cite{DeAlba2007}}.
Hence, the equation $V\cdot\ov{a}=\ov{0}$ implies that
$\ov{a}=\ov{0}$.
This contradicts \eqref{eq:hw:aiNotNull} and completes the proof of
Lemma~\ref{lemma:adding_edges_keeps_distinguishing}.
\end{proof}  

\medskip

\begin{proof}[\upshape\textbf{Proof sketch for Theorem~\ref{thm:EHWandIGHW}:}] \ \\
The direction ``$\Longrightarrow$'' of Theorem~\ref{thm:EHWandIGHW} is trivial.
For the direction ``$\Longleftarrow$'' it suffices to prove the following:
If there is a $\J\in\IGHW{k}$ with  $\hom{\J}{\I} \neq \hom{\J}{\I'}$,
then there also exists a $\J'\in\IEHW{k}$ with $\hom{\J'}{\I} \neq
\hom{\J'}{\I'}$. 
We construct such a $\J'$ in a 2-step process.
We start with a ghd $D=(T,\DBag,\DCover)$ of $J$ with $\w{D}\leq k$.
First, we transform $D$ into a ghd $\Deins$ of an incidence graph $\Jeins$ such that $\w{\Deins}\leq \w{D}$ and 
$\hom{\Jeins}{\I} \neq \hom{\Jeins}{\I'}$ and $\Deins$ satisfies
condition~\ref{def:entangled_hypertree_decomposition:precision} of
Definition~\ref{def:entangled_hypertree_decomposition} (but
condition~\ref{def:entangled_hypertree_decomposition:connectedness_of_edges}
might still be violated).
Afterwards, we transform $\Deins$ into a ghd $\Dzwei$ of an
incidence graph $\Jzwei$ such that $\w{\Dzwei}= \w{\Deins}$ and 
$\hom{\Jzwei}{\I} \neq \hom{\Jzwei}{\I'}$ and $\Dzwei$ satisfies
conditions~\ref{def:entangled_hypertree_decomposition:precision}
and \ref{def:entangled_hypertree_decomposition:connectedness_of_edges}
of Definition~\ref{def:entangled_hypertree_decomposition} and hence 
is an ehd. Letting $\J'\deff\Jzwei$ then completes the proof.

For the construction of $\Deins, \Jeins$ we consider all those
$t\in \V{T}$ and $e\in \DCover(t)$ where $\Nachbarn{e}{J}
\not\subseteq \DBag(t)$ and let $s\deff \Nachbarn{e}{J}
\cap\DBag(t)$.
We use Lemma~\ref{lemma:adding_edges_keeps_distinguishing} to choose
a suitable number $n_s\geq 1$ and replace $\J$ by
$\J+n_s{\cdot}s$ (let us write $e'_{1},\ldots e'_{n_s}$ for
the $n_s$ newly inserted blue nodes).
In  $D$ we replace $e$ with $e'_{1}$ in $\DCover(t)$, and
we add new leaves $t_j$ for $j\in\set{2,\ldots,n_s}$ adjacent to $t$
with $\DCover(t_j)=\set{e'_j}$ and $\DBag(t_j)=s$.
After having done this for all combinations of $t$ and $e$, we end up
with the desired incidence graph $\Jeins$ and ghd $\Deins=(\Teins,\bageins,\covereins)$.

For the construction of $\Dzwei,\Jzwei$, for
each $e\in \bV{\Jeins}$ we let $m_e$ be the
number of connected components of the subgraph
$\Teins_{e}$, i.e., the subgraph of $\Teins$ induced on
$V_e\deff\setc{t\in \V{\Teins}}{e\in\covereins(t)}$.
Let $V_{e,0},\ldots,V_{e,m_e-1}$
be the sets of tree-nodes (i.e., nodes in $\V{\Teins}$) of these
connected components.
We consider all those $e\in \bV{\Jeins}$ where $m_e\geq 2$
and let $s\deff N_{\Jeins}(e)$. We use Lemma~\ref{lemma:adding_edges_keeps_distinguishing} to choose
a suitable number $n_e\geq m_e{-}1$ and replace ${\J}$ with
${\J}+n_e{\cdot}s$ (let us write
$e'_1,\ldots,e'_{n_e}$ for the $n_e$ newly inserted blue nodes).
In $\Deins$ we consider for every $i\in\set{1,\ldots, m_e{-}1}$
all $t\in V_{e,i}$ and replace $e$ with $e'_{i}$ in
$\covereins(t)$.
Furthermore, we pick an arbitrary $t\in V_{e,0}$, and for each
$i\in [n_e]$ with $i\geq m_e$, we insert into $\Teins$ a new leaf $t_{e,i}$
adjacent to $t$ and let $\bageins(t_{e,i}) \deff s$ and
$\covereins(t_{e,i})\deff \set{e'_{e,i}}$.
After having done this for all $e\in \bV{\Jeins}$ with $m_e\geq 2$, we
end up with the desired incidence graph $\Jzwei$ and ehd $\Dzwei$.
This completes the proof sketch of Theorem~\ref{thm:EHWandIGHW} (see
Appendix~\ref{appendix:Proof_of_thm:EHWandIGHW} and
\ref{appendix:illustrations_for_thm:EHWandIGHW} for more details and
illustrations). 
\end{proof}

\section{Notation for Partial Functions}
\label{section:FurtherNotation}

Let us introduce some further notation that will be convenient for the remaining parts of the paper.

We write $f:A\pto B$ to indicate that $f$ is a partial function from
$A$ to $B$. By $\Def{f}$ we denote the domain of $f$, i.e., the set of
all $a\in A$ on which $f(a)$ is defined. 
By $\Img{f}$ we denote the image of $f$, i.e., 
$\Img{f} = \setc{f(a)}{a\in \Dom{f}}$.
Two partial functions $f: A \pto B$ and $g: A \pto B$ are called \emph{compatible} if $f(a)=g(a)$ holds for all
$a\in\Dom{f}\cap\Dom{g}$.

We identify a partial function $f$ with the set $\setc{(a,f(a))}{a\in
\Def{f}}$. This allows us to compare and combine partial functions
via standard notation from set theory. 
E.g., $f \subseteq g$ indicates that $\Def{f} \subseteq \Def{g}$ and $f(a) = g(a)$ for all $a \in \Def{f}$.
And $f \union g$ denotes the partial function $h$ with
$\Def{h}=\Def{f}\cup\Def{g}$ and $h(a) = f(a)$ for all $a \in \Def{f}$
and $h(a) = g(a)$ for all $a \in \Def{g} \setminus \Def{f}$ (note that
$f$ has precedence over $g$ in case that $f$ and $g$ are 
not~compatible).
For a set $S$ we write $f-S$ to denote the partial function $g$ with $g\subseteq f$ and $\Def{g}=\Def{f}\setminus S$.

\section{A 2-Sorted Counting Logic with Guards: $\GCk$}
\label{sec:logic}

This section provides the syntax and semantics of our 2-sorted logic.
Formulas of this logic are evaluated on incidence graphs
(cf.~Definition~\ref{def:IncidenceGraph}). We fix a $k\in\NNpos$.

To address \emph{blue} nodes (i.e., \emph{edges} of a hypergraph), we have available $k$ \emph{blue
variables} $\vare_1,\ldots,\vare_k$. To address \emph{red} nodes
(i.e., \emph{vertices} of a hypergraph), we have
available countably many \emph{red variables}
$\varv_1,\varv_2,\varv_3,\ldots{}$.
An atomic formula $E(\vare_j,\varv_i)$ states that a
hypergraph's vertex $\varv_i$ is included in the hypergraph's edge $\vare_j$.

Let $\VarB\deff \set{\vare_1,\ldots,\vare_k}$, $\VarR\deff
\setc{\varv_i}{i\in\NNpos}$, and $\Var\deff \VarB\cup\VarR$.
An \emph{interpretation} $\IntI=(\I,\beta)$ consists of an incidence
graph $\I=(\rVI,\bVI,\E{I})$ and an \emph{assignment $\beta$ in $\I$}, i.e., a mapping
$\beta:\Var\to\rVI\cup\bVI$ with $\beta(\vare_j)\in\bVI$ for all
$\vare_j\in\VarB$ and $\beta(\varv_i)\in\rVI$ for all $\varv_i\in\VarR$.

In the formulas of our logic, red variables
$\varv_i$ have to be \emph{guarded} by a blue variable $\vare_j$ in the
sense that $E(\vare_j,\varv_i)$ holds. This is formalised  by a
\emph{guard function}, i.e., a \emph{partial function} $g:\NNpos\pto[k]$ with
\emph{finite} domain $\Dom{g}$.
Every guard function $g$ corresponds to the formula
\[
  \LogGuard{g} \ \isdef \ \ \Und_{i\in\Dom{g}} \; E(\vare_{g(i)},\varv_i)\;,
\]
and for the special case where $\Dom{g}=\emptyset$ we let
$\LogGuard{g}\deff \top$ where $\top$ is a special atomic formula
satisfied by \emph{every} interpretation $\IntI$.
We let $\free{\LogGuard{g}}$ be the set of all (red or blue) variables that occur in
$\LogGuard{g}$.

An interpretation $\IntI=(\I,\beta)$ \emph{satisfies} a guard function
$g$ (in symbols: $\IntI\models \LogGuard{g}$) if
for all $i\in\Dom{g}$ we have: $(\beta(\vare_{g(i)}),\beta(\varv_i))\in\E{\I}$.
I.e., for every $i\in\Dom{g}$, the red variable $\varv_i$ is guarded by
the blue variable $\vare_{g(i)}$ in
the sense that it is connected to it by an edge of the incidence graph.

For any formula $\chi$ we write
$\indfreeB{\chi}$ for
the set of all indices $j\in[k]$ such that the blue variable $\vare_j$
belongs to $\free{\chi}$. Accordingly,
$\indfreeR{\chi}\deff\setc{i\in\NNpos}{\varv_i\in\free{\chi}}$. 
The definition of the syntax of 
$\GCk$ is inductively given as follows.
\medskip

\noindent
\textbf{Base cases:} The atomic formulas in $\GCk$ are of the form $\top$, $E(\vare_j, \varv_i)$, $\vare_{j} {=} \vare_{j'}$, and $\varv_i {=} \varv_{i'}$ for $j, j' \in [k]$ and $i,i' \in \natpos$.
\\
We let $\free{\top}=\emptyset$,
$\free{E(\vare_j, \varv_i)}=\set{\vare_j,\varv_i}$, $\free{\vare_{j}
  {=} \vare_{j'}}=\set{\vare_j,\vare_{j'}}$, and
$\free{\varv_i {=} \varv_{i'}}=\set{\varv_i,\varv_{i'}}$.
\medskip

\noindent \textbf{Inductive cases:}
\begin{enumerate}
\setcounter{enumi}{\value{LogicCounter}}
\item
  If $ \psi \in \GCk$, then
  $\lnot\psi \in \GCk$; \ let $\free{\lnot\psi}=\free{\psi}$;
  
\item
  If $\psi_1, \psi_2 \in \GCk$, then
  $(\psi_1 \land \psi_2) \in\GCk$; \ let $\free{(\psi_1\land\psi_2)}=
   \free{\psi_1}\cup\free{\psi_2}$;
\item
  If $\psi \in \GCk$ and $g$ is a guard function with $\Dom{g} =
  \indfreeR{\psi}$ and $n,\ell \in \natpos$ and, for $\chi \isdef
  (\LogGuard{g} \land \psi)$ and  $i_1 < \dots < i_\ell$ with
  \begin{enumerate}[(a)]
	\item $i_1, \dots, i_\ell \in \indfreeR{\chi}$, then $\phi \in
          \GCk$ for $\phi \isdef \existsi[n] (\varv_{i_1}, \dots,
          \varv_{i_\ell}).(\LogGuard{g} \land \psi)$;
          \ let $\free{\phi}=\free{(\LogGuard{g}\land\psi)}\setminus\set{\varv_{i_1},\ldots,\varv_{i_\ell}}$;
	\item $i_1, \dots, i_\ell \in \indfreeB{\chi}$, then $\phi \in
          \GCk$ for $\phi \isdef \existsi[n] (\vare_{i_1}, \dots,
          \vare_{i_\ell}).(\LogGuard{g} \land \psi)$;
          \ let $\free{\phi}=\free{(\LogGuard{g}\land\psi)}\setminus\set{\vare_{i_1},\ldots,\vare_{i_\ell}}$.
  \end{enumerate}
\end{enumerate}
\medskip

The semantics are defined as expected. 
In particular, an interpretation $\IntI=(\I,\beta)$ satisfies the
formula $\phi\deff
\existsi[n](\varv_{i_1},\ldots,\varv_{i_\ell}).(\LogGuard{g}\und\psi)$
iff there are
at least $n$ tuples $(v_{i_1},\ldots,v_{i_\ell})\in\rVI^\ell$ such that 
$\IntI'=(\I,\beta')$ satisfies $(\LogGuard{g}\und\psi)$, where
$\beta'(\varv_{i_j})=v_{i_j}$ for all $j\in[\ell]$ and
$\beta'(\varx)=\beta(\varx)$ for all
$\varx\in\VAR\setminus\set{\varv_{i_1},\ldots,\varv_{i_\ell}}$. 
Similarly, $\IntI=(\I,\beta)$ satisfies
$\phi\deff\existsi[n](\vare_{i_1},\ldots,\vare_{i_\ell}).(\LogGuard{g}\und\psi)$ iff
there are at least $n$ tuples
$(e_{i_1},\ldots,e_{i_\ell})\in \bVI^\ell$ such that
$\IntI'=(\I,\beta')$ satisfies $(\LogGuard{g}\und\psi)$, where
$\beta'(\vare_{i_j})=e_{i_j}$ for all $j\in[\ell]$ and
$\beta'(\varx)=\beta(\varx)$ for all
$\varx\in\VAR\setminus\set{\vare_{i_1},\ldots,\vare_{i_\ell}}$.
Obviously we can emulate the $\forall$-quantifier (and disjunction) using $\existsi[1]$ and $\lnot$ (and $\land$ and $\lnot$, respectively).

We write $\IntI\models \chi$  to indicate that $\IntI$ satisfies
the formula $\chi$; and
$\IntI\not\models\chi$ indicates that $\IntI$ does not satisfy $\chi$.

\emph{Sentences} of $\GCk$ are formulas $\chi \in \GCk$ with $\free{\chi} = \emptyset$.

For an incidence graph $\I$ and a sentence $\chi\in\GCk$ we write
$\I\models\chi$ to indicate that $\IntI\models\chi$
where $\IntI=(\I,\beta)$ for any assignment $\beta$ in $\I$ (since
$\chi$ has no free variable, the assignment does not matter).
For a hypergraph $H$ and a sentence $\chi\in\GCk$ we write $H\models\chi$ to
indicate that $\inzidenz{H}\models\chi$.

For two incidence graphs $\I$ and $\I'$ we write $\I\equivGCk\I'$ and
say that $\I$ and $\I'$ are \emph{indistinguishable by the logic $\GCk$} if
for all sentences $\chi\in\GCk$ we have: $\I\models\chi$ $\iff$ $\I'\models\chi$.

The subsequent sections of this paper are devoted to proving the following
theorem, stating that indistinguishability by the logic $\GCk$
coincides with homomorphism indistinguishability over the class
$\IEHW{k}$ of incidence graphs of entangled hypertree width $\leq k$.

\begin{theorem}\label{thm:GCkVsHomomIndist}
For all incidence graphs $\I,\I'$ and all $k\in\NNpos$ we have: \\
$\I\equivGCk\I'\iff
\HOM{\IEHW{k}}{\I}=\HOM{\IEHW{k}}{\I'}$.
\end{theorem}

This result can be viewed as a lifting of Dvořák's theorem
\cite{Dvorak2010} stating that any two graphs $G,G'$ are
indistinguishable by the $k{+}1$-variable logic $C^{k+1}$ if, and only
if, they are homomorphism indistinguishable over the class $\TW{k}$ of
graphs of tree-width $\leq k$.
Our proof of Theorem~\ref{thm:GCkVsHomomIndist} is heavily inspired by
Dvořák's proof.
But in order to proceed along a similar construction, we first have to provide 
a \enquote{normal form} for $\GCk$, which we present in Section~\ref{sec:normal_form}, and 
a suitable inductive characterisation of the class $\IEHWk$. The latter is
presented in Section~\ref{sec:recursive_def}, where we also provide the machinery
of \emph{quantum} incidence graphs as an analogue of the quantum
graphs used in Dvořák's proof. 
Section~\ref{sec:main_theorem} is devoted to the proof of
Theorem~\ref{thm:GCkVsHomomIndist}. 

\medskip

Before closing this section, let us
provide examples of formulas in $\GCk$. 

\begin{example}\label{example:GCk:eins}
Let $k=2$.
Consider the following formula $\psi_1 \in \GCk$:
\begin{equation*}
	\psi_1 \;\;\isdef\;\; \existsi[1](\varv_1).\bigl(E(\vare_1, \varv_1) \land E(\vare_2, \varv_1) \bigr).
\end{equation*}
$\psi_1$ expresses that the hyperedges $\vare_1$ and $\vare_2$ share at least one vertex $\varv_1$, i.e.\ they intersect. Since we quantify over $\varv_1$, the definition of $\GCk$ requires us to insert a guard ranging over the set of free red variables, i.e. over $\set{ \varv_1 }$. We chose $E(\vare_1, \varv_1)$ as the guard but note that $E(\vare_2, \varv_1)$ would have been a valid choice as well.
Next, consider the formula $\psi_2 \in \GCk$:
\begin{equation*}
	\psi_2 \;\;\isdef\;\; \!\!\!\bigland_{j \in \set{1,2}}\!\!\! \existsi[3](\varv_1).\bigl(E(\vare_j, \varv_1) \land E(\vare_j, \varv_1)\bigr).
\end{equation*}
$\psi_2$ expresses that each of the hyperedges $\vare_1$, $\vare_2$ contains at least three vertices. Again, we have to insert a guard after the quantifier, which is why $E(\vare_j, \varv_1)$ appears twice in $\psi_2$ --- as a guard \emph{and} as our \enquote{actual} subformula.

Finally, we use the formulas $\psi_1$, $\psi_2$ to construct a sentence $\phi \in \GCk$:
\begin{equation*}
	\phi \;\;\isdef\;\; \lnot\,\existsi[1](\vare_1, \vare_2).\bigl( 
		\top \land (\, 
			(\psi_1 \land 
				\psi_2)\, \land \lnot\;
                                \vare_1{=}\vare_2 \,
			)\,
	\bigr).
\end{equation*}
$\phi$ expresses that there is no pair of non-equal hyperedges
$(\vare_1, \vare_2)$ that intersect and that both contain at least 3
vertices. I.e., $\phi$ expresses that all hyperedges that contain at
least 3 vertices are pairwise disjoint. Once again, the quantification
requires us to insert a guard; since there are no free red variables,
we insert $\top$ as the guard.
\end{example}

For a tuple $\ov{\varx}$ of either blue or red variables, we write
$\existsex[n]\,\ov{\varx}.\phi$\, as a shorthand for
\,$(\existsi[n]\,\ov{\varx}.\phi
\,\land\,\lnot\,\existsi[n+1]\,\ov{\varx}.\phi)$.

\begin{example}\label{example:formula}
Let $H$ be
the hypergraph with vertices $1,2,3,4,5,6$, edges $e_S$ for
each $S\in\set{\, \set{1,2,3}, \allowbreak \set{4,5,6}, \allowbreak
  \set{1,2}, \allowbreak \set{2,3}, \allowbreak \set{3,1}, \allowbreak
  \set{4,5}, \allowbreak \set{5,6}, \allowbreak \set{6,4} \,}$ and
incidence function $f_H$ with 
$f_H(e_S)=S$ for each edge $e_S$ of $H$.
We construct a sentence $\chi\in\GC{2}$ that describes $H$ up to isomorphism.
$\chi$ is the conjunction of
the following statements $\phi_1,\ldots,\phi_5$.
It can easily be checked that for
each $i\in[5]$ the formula $\phi_i$ is a sentence in
$\GC{2}$. Hence, $\chi$ is also a
sentence in $\GC{2}$.

\begin{enumerate}
\item
There exist exactly 8 blue nodes (i.e., edges of $H$):
\[
\phi_1 \deff  \ \existsex[8](\vare_1).(\top\und\vare_1{=}\vare_1)
\]
\item
There exist exactly 2 blue nodes each of which has exactly 3 red neighbours
(i.e., $H$ has exactly 2 edges that are each
incident with exactly 3 vertices of $H$):
\[
  \phi_2 \deff  \ \existsex[2](\vare_1).\big(\top\und \existsex[3] (\varv_1). \big(
    E(\vare_1,\varv_1) \und E(\vare_1,\varv_1)
  \big)\big)
\]
The first and 2nd occurrence of $E(\vare_1,\varv_1)$ are used as the
guard and the formula, respectively.

\item
A similar formula $\phi_3$ expresses that there are exactly 6 blue
nodes each of which has exactly 2 red neighbours

\item
All blue nodes that have at least 3 red neighbours have
disjoint neighbourhoods:
\[
  \phi_4\deff \ \nicht\,\existsi[1] (\vare_1,\vare_2).\big(\top \,\und\, (\alpha\,\und\,
     \existsi[1](\varv_1).(\Und_{j\in[2]}E(\vare_j,\varv_1)))\big)
\]
where
\[
  \alpha\deff  \ \big(\nicht\,\vare_1{=}\vare_2 \;\und
  \Und_{j\in[2]} \!\existsi[3](\varv_1).(E(\vare_j,\varv_1)\,\und\,
  E(\vare_j,\varv_1))\big)
\]
The formula $\alpha$ expresses that $\vare_1$ and $\vare_2$ are
distinct blue nodes that each have at least $3$ red neighbours.

\item
The final statement is expressed by the formula
\[
  \phi_5\deff  \
  \existsex[2](\vare_1).(\top\,\und\;\existsi[1](\varv_1,\varv_2,\varv_3).\big(
    \Delta \und (\chi \und \psi)
  \big)
\]    
where
\[
  \Delta\deff \Und_{i=1}^3 E(\vare_1,\varv_i)
  \quad\text{and}\quad
  \chi\deff  \!\!\!\!\Und_{1\leq i<j\leq 3}\!\!\!\!\nicht\;\varv_{i}{=}\varv_{j}
\]
and
$\psi\deff  \big( (\, \psi_{1,2}\und\psi_{2,3}\,)\und\psi_{3,1}\big)$
where
\[
  \psi_{i,j} \ \deff  \   \existsex[1](\vare_2).\big( \Delta_{i,j} \und
  \vartheta_{i,j}\big)
  \quad\text{with}
\]  
{\setlength{\arraycolsep}{0mm}
\[
 \begin{array}{ll}
 \Delta_{i,j}\deff  \big( & E(\vare_1,\varv_i)\,\und\, E(\vare_1,\varv_j))\qquad\text{and}
 \\[2ex]
 \vartheta_{i,j} \deff 
   \big(\hspace*{1mm}
   &
    E(\vare_2,\varv_i)\und E(\vare_2,\varv_j)\,\und \,
    \existsex[2](\varv_4). (E(\vare_2,\varv_4)\und E(\vare_2,\varv_4))
  \big)\,.
 \end{array}
\]  
}%
The formula $\phi_5$ states that there exist exactly 2 blue nodes
$e_1$ (i.e., edges of $H$) 
each of which is adjacent to at least 3 distinct red nodes
$v_1,v_2,v_3$ (i.e., vertices of $H$) such that $H$ contains an edge
incident with exactly $v_1$ and $v_2$, an edge incident with exactly
$v_2$ and $v_3$, and an edge incident with exactly $v_3$ and $v_1$.
\end{enumerate}
Finally, $\chi\deff\Und_{i=1}^5\varphi_i$ is a sentence in $\GC{2}$.
It is not difficult to verify that $\chi$ describes $H$ up to isomorphism. 
\end{example}

\section{A Normal Form for $\GCk$: $\NGCk$}
\label{sec:normal_form}

In this section we provide a normal form for sentences of $\GCk$ that
will be crucial for our proof of
Theorem~\ref{thm:GCkVsHomomIndist}. In the restriction $\RGCk$ of
$\GCk$,
every formula is of the form $(\LogGuard{g}\und \psi)$,
where $g$ is a guard function whose domain $\Dom{g}$
consists of
all indices $i\in\NNpos$ such that the red
variable $\varv_i$ is a free variable of $\psi$.
Note that 
$\free{(\LogGuard{g}\und\psi)}\deff\free{\LogGuard{g}}\cup\free{\psi}$
is the set of free variables of the formula $(\LogGuard{g}\und\psi)$.
The definition of the syntax of 
$\NGCk$ is inductively given as follows.
\medskip

\noindent\textbf{Base cases:} 
$(\LogGuard{g} \und \psi) \in \NGCk$ for all $\psi$ and all $g: \NNpos
\pto [k]$ matching one of the following:
\begin{enumerate}
\item\label{item:syntaxdef:1}
  $\psi$ is $E(\vare_j, \varv_i)$ and $\Dom{g}=\set{i}$ and $j\in[k]$

  (note that $g(i)$ can be an arbitrary element in $[k]$);

\item\label{item:syntaxdef:2}
  $\psi$ is $\vare_{j} {=} \vare_{j'}$ with $\Dom{g}=\emptyset$ and $j,j'\in[k]$;

\item\label{item:syntaxdef:3}
  $\psi$ is $\varv_i {=} \varv_{i'}$ with $\Dom{g}=\set{i,i'}$.

\setcounter{LogicCounter}{\value{enumi}}
\end{enumerate}

\noindent \textbf{Inductive cases}:
\begin{enumerate}
	\setcounter{enumi}{\value{LogicCounter}}
	\item\label{item:syntaxdef:4}
	  If $(\LogGuard{g}\und \psi)\in\NGCk$, then
	  $(\LogGuard{g}\und\nicht{\psi})\in \NGCk$;

	\item\label{item:syntaxdef:5}
	  If $(\LogGuard{g_i}\und\psi_i)\in\NGCk$ for $i\in[2]$ and
	  $g_1$ and $g_2$ are compatible (i.e., they agree on $\Dom{g_1}\cap\Dom{g_2}$),
	  then
	  $(\LogGuard{g}\und \phi) \in\NGCk$ for
	  $g \deff  g_1\union g_2$ and $\phi\deff (\psi_1 \land \psi_2)$;
	
	\item\label{item:syntaxdef:6}
	  If $(\LogGuard{g}\und\psi)\in \NGCk$ and $n,\ell\in\NNpos$, and
	  $i_1,\ldots,i_\ell\in\Dom{g}$ with $i_1<\cdots<i_\ell$,
	  then $(\LogGuard{\tilde{g}}\und\phi)\in \NGCk$ for
	  \begin{eqnarray*}
		\phi & \deff  &
		\existsi[n] (\varv_{i_1},\ldots,\varv_{i_\ell}).(\LogGuard{g}\und \psi)
					\quad\text{and}
		\\
		\tilde{g} & \deff  & g-\set{i_1,\ldots,i_\ell}
	  \end{eqnarray*}
	  (note that
	  $\free{\phi}=
	  \free{(\LogGuard{g}\und\psi)}\setminus\set{\varv_{i_1},\ldots,
		\varv_{i_\ell}}$);

	\item\label{item:syntaxdef:10}\label{item:syntaxdef:7}\label{item:syntaxdef:8}
	  If $(\LogGuard{g}\und\psi)\in \NGCk$ and 
	  $n,\ell\in\NNpos$, and 
	  $S\deff\set{i_1,\ldots,i_\ell}\subseteq
	  \indfreeB{\chi}$ for $\chi\deff(\LogGuard{g}\und\psi)$
	  with
	  $i_1<\cdots <i_\ell$, and 
	  if $\tilde{g}:\NNpos\pto [k]$ with
	  $\Dom{\tilde{g}}=\Dom{g}$
	  such that
	  all $i\in\Dom{g}$ satisfy
	\begin{equation}\label{eq:syntax:guard}
	  \tilde{g}(i)=g(i)
	  \quad\text{or}\quad
	  \tilde{g}(i)\in S
	  \quad \text{or} \quad
	  \tilde{g}(i)\not\in
	  \Img{g}
	\end{equation}
	  then $(\LogGuard{\tilde{g}}\und \phi)\in\NGCk$ for
	\begin{eqnarray*}
		\phi & \deff  &
		\existsi[n] (\vare_{i_1},\ldots,\vare_{i_\ell}).(\LogGuard{g}\und
						\psi)\,.
	  \end{eqnarray*}
	  (note that
	  $\free{\phi}=
	  \free{\chi}\setminus\set{\vare_{i_1},\ldots,\vare_{i_\ell}}$).
\end{enumerate}

\noindent
Let us have
a closer look at rule~\ref{item:syntaxdef:10}): The
formula $\phi$ has exactly the same free red variables as the formula
$\chi$. 
But the guard of red variable $\varv_i$ in
$\tilde{\chi}\deff (\LogGuard{\tilde{g}}\und\phi)$ 
is $j'\deff\tilde{g}(i)$, whereas in $\chi$ it is $j\deff
g(i)$. Condition \eqref{eq:syntax:guard} is equivalent to the
following: the guard remains unchanged
(i.e., $j'{=}j$), or the new guard $j'$ has become \enquote{available}
by the quantification (i.e., $j'\in S$) or it has not been used as a
guard by $g$ (i.e., $j'\not\in\Img{g}$).

\begin{example}
Let $\phi$ be the $\GCk$-sentence from Example~\ref{example:GCk:eins}.
It is straightforward to see that $(\top\und\phi)$ is an equivalent
sentence in $\RGCk$.

Let $\chi$ be the $\GCk$-sentence from Example~\ref{example:formula}.
It can be shown that $(\top\und\chi)$ is an equivalent sentence in
$\RGCk$ (cf.\ Appendix~\ref{appendix:logic}).
\end{example}  

Inductively one obtains for all $\chi\deff(\LogGuard{g}\und\psi)\in
\NGCk$ that
\begin{equation}\label{eq:FreeRedVarsAndDomg}
  \Dom{g}
  \ = \
  \setc{i\in\NNpos}{\varv_i\in\free{\chi}}
  \ = \
  \setc{i\in\NNpos}{\varv_i\in\free{\psi}}
  \,.
\end{equation}

Note that $\NGCk \subseteq \GCk$. We use the same notions and notation
as for $\GCk$.
\emph{Sentences} of $\NGCk$ are formulas $\chi\deff (\LogGuard{g}\und
\psi)$ in $\NGCk$ with $\free{\chi}=\emptyset$. 
From \eqref{eq:FreeRedVarsAndDomg} we know that this implies that
$\Dom{g}=\emptyset$, i.e., $g=g_\emptyset$ where $g_\emptyset$ is
the uniquely defined partial mapping with empty domain; recall that
$\LogGuard{g_\emptyset}=\top$. 
For two incidence graphs $\I$ and $\I'$ we write $\I\equivRGCk\I'$ and
say that $\I$ and $\I'$ are \emph{indistinguishable by the logic $\RGCk$} if
for all sentences $\chi\in\RGCk$ we have: $\I\models\chi$ $\iff$ $\I'\models\chi$.

Let us introduce some shorthands, that will be useful in subsequent sections.
If $(\LogGuard{g_i}\und\psi_i)\in\NGCk$ for
$i\in[2]$, $g_1,g_2$ are compatible, and $g=g_1\cup g_2$, we write
\begin{itemize}
\item $\big(\LogGuard{g} \land (\psi_1 \lor \psi_2)\big)$ for $\big(\LogGuard{g} \land \nicht(\nicht \psi_1
  \und \nicht \psi_2)\big)$, and
\item $\big( \LogGuard{g} \land (\psi_1 \limplies \psi_2) \big)$ for $\big(\LogGuard{g} \land (\nicht \psi_1 \oder \psi_2) \big)$.
\end{itemize}  
If $(\LogGuard{g} \land \psi) \in \NGCk$ and $i_1,\ldots,i_\ell \in \Dom{g}$ with
$i_1 < \cdots < i_\ell$ and $\tilde{g} = g - \set{i_1,\ldots,i_\ell}$ and
$\ov{\varv} = (\varv_{i_1},\ldots,\varv_{i_\ell})$, then we write
\begin{itemize}
\item
  $\big(\LogGuard{\tilde{g}} \land \forall\,
  \tupel{\varv}.(\LogGuard{g}\impl\psi)\big)$ for
  $\big( \LogGuard{\tilde{g}} \land
  \nicht\,\existsi[1]\,\tupel{\varv}.(\LogGuard{g}\und \nicht\psi)\big)$ and
\item
  $\big(\LogGuard{\tilde{g}} \land \existsex[n]\,\tupel{\varv}.(\LogGuard{g}\und
  \psi) \big)$ for\\
  $\big(\LogGuard{\tilde{g}} \land \big(\,\existsi[n]\,\tupel{\varv}.(\LogGuard{g}\und
  \psi) \;\und\; \nicht\, \existsi[n+1]\,\tupel{\varv}.(\LogGuard{g}\und
  \psi)\,\big) \big)$.
\end{itemize}
And we use the analogous shorthands for formulas quantifying over blue
variables $\vare_{i_1},\ldots,\vare_{i_\ell}$. 

Two $\GCk$-formulas $\phi$ and $\psi$ are said to be \emph{equivalent}
(for short: $\phi\equiv\psi$) if for every interpretation $\IntI$ we
have: $\IntI\models\phi$ $\iff$ $\IntI\models\psi$.
This section's main result shows that every $\GCk$-sentence is
equivalent to an $\RGCk$-sentence and, in general, $\RGCk$ can be be
viewed as a \enquote{normal form} for $\GCk$ in the following sense:
\begin{theorem}\label{thm:ngck_is_equivalent}\label{thm:normalform}
	For all formulas $\phi \in \GCk$ and all guard functions $f$
        with $\Dom{f} = \indfreeR{\phi}$, there exists a formula
        $\phi_f$ such that $(\LogGuard{f} \land \phi_f) \in \NGCk$ and
	\begin{equation*}
		(\LogGuard{f} \land \phi) \ \equiv \ (\LogGuard{f}
                \land \phi_f)
                \qquad\text{and}\qquad
\free{(\LogGuard{f} \land \phi)} \; = \; \free{(\LogGuard{f}
                \land \phi_f)}\,.
        \end{equation*}
In particular, if $\phi$ is a \emph{sentence} in $\GCk$, then for
the empty guard function $f$,
$(\top\und\phi_f)$ is a sentence in $\RGCk$ that is equivalent to $\phi$.              
\end{theorem}
\begin{proof} We proceed by induction on the construction of $\GCk$-formulas.
  \medskip
  
  \noindent\textbf{Base cases}:
  For the atomic formula $\phi\deff\top \in\GCk$ note that the only
  applicable guard function is the function $f$ with empty domain and
  $\LogGuard{f}=\top$. 
  We can choose $\phi_f\deff \nicht\,\existsi[1](\vare_1).(\top\,\und\,\nicht\,\vare_1{=}\vare_1)$.
  It is straightforward to verify that $(\top\land\phi_f)\in\NGCk$ is a
  sentence equivalent to $(\top\land\top)$.
  \smallskip
  
  For all the remaining atomic formulas $\phi \in \GCk$ we can set $\phi_f \isdef \phi$ since, in the base case, $\NGCk$ does not impose any restrictions on $f$ other than $\Dom{f} = \indfreeR{\phi_f}$.

	\medskip
	\noindent\textbf{Induction hypothesis}: The theorem's statement holds for formulas $\chi, \xi \in \GCk$.

	\medskip
	\noindent\textbf{Inductive step}:\medskip
        
	\noindent\textbf{Case 1}: $\phi$ is of the form $\lnot
        \chi$. By induction hypothesis we have: $(\LogGuard{f} \land
        \chi) \equiv (\LogGuard{f} \land \chi_f)$ and
$\free{(\LogGuard{f} \land
        \chi)}= \free{ (\LogGuard{f} \land \chi_f)}$
        and $(\LogGuard{f} \land \chi_f) \in \NGCk$. Then, also
        $(\LogGuard{f} \land \lnot \chi_f) \in \NGCk$ and
        $\free{(\LogGuard{f} \land \lnot \chi_f)}=\free{(\LogGuard{f} \land \phi)}$.
        We let $\phi_f \isdef \lnot \chi_f$ and claim that
        $(\LogGuard{f} \land \phi) \equiv (\LogGuard{f} \land
        \phi_f)$: \
	Consider an arbitrary interpretation $\IntI$. If $\IntI \not\models \LogGuard{f}$, then $\IntI \not\models (\LogGuard{f} \land \phi)$ and $\IntI \not\models (\LogGuard{f} \land \phi_f)$. If $\IntI \models \LogGuard{f}$ then 
	\begin{align*}
		\IntI \models (\LogGuard{f} \land \phi) &\iff \IntI \models \phi \\
		&\iff \IntI \not\models \chi \\
		&\iff \IntI \not\models (\LogGuard{f} \land \chi) \\
		&\iff \IntI \not\models (\LogGuard{f} \land \chi_f) \\
		&\iff \IntI \not\models \chi_f \\
		&\iff \IntI \models \phi_f \\
		&\iff \IntI \models (\LogGuard{f} \land \phi_f)\,. \\
	\end{align*}

	\noindent\textbf{Case 2}: $\phi$ is of the form $(\chi \land
        \xi)$. Let $f_\chi$ and $f_\xi$ be the restrictions of $f$ to
        $\indfreeR{\chi}$ and $\indfreeR{\xi}$, respectively. It
        easily follows from $f = f_\chi \union f_\xi$ and the
        induction hypothesis, that $(\LogGuard{f} \land (\chi \land
        \xi)) \equiv (\LogGuard{f} \land (\chi_{f_\chi} \land
        \xi_{f_\xi}))$ and
        $\free{(\LogGuard{f} \land (\chi \land
        \xi))}
        =
        \free{(\LogGuard{f} \land (\chi_{f_\chi} \land
        \xi_{f_\xi}))}$.
        Thus, we use $\phi_f \isdef (\chi_{f_{\chi}} \land
        \xi_{f_{\xi}})$.
        It is straightforward to verify that $(\LogGuard{f}\und\phi_f)\in\NGCk$.
	\medskip

\noindent\textbf{Case 3}: $\phi$ is of the form
$\existsi[n](\varv_{i_1}, \dots, \varv_{i_\ell}).(\LogGuard{g} \land
\chi)$, where $\Dom{g}=\indfreeR{\chi}$ and $i_1<\cdots<i_\ell$ and
$i_1,\ldots,i_\ell\in\indfreeR{(\LogGuard{g}\und\chi)}$. 
Let $f$ be an arbitrary guard function with $\Dom{f}=\indfreeR{\phi}$.

Let $S \isdef \set{ i_1, \dots, i_\ell }$.
Note that $\Dom{f}=\Dom{g}\setminus S$.
For all $i \in \Dom{f}$ choose $g'(i) \isdef f(i)$, and for all $i \in
S$ choose $g'(i) \isdef g(i)$ (thus, $g'=f\cup g$ according to the
notation introduced in Section~\ref{section:FurtherNotation}). Let
	\begin{equation*}
		\phi_f \ \ \isdef \ \ \existsi[n](\varv_{i_1}, \dots, \varv_{i_\ell}).(\LogGuard{g'} \land (\LogGuard{g} \land \chi_{g'})).
	\end{equation*} 
	It is easy to see that $(\LogGuard{f} \land \phi_f) \in \NGCk$: $(\LogGuard{g'} \land \chi_{g'}) \in \NGCk$ holds by induction hypothesis, and $(\LogGuard{g'} \land \LogGuard{g}) \in \NGCk$ can easily be verified using $\Dom{g'} = \Dom{g}$. Therefore, $(\LogGuard{g'} \land (\LogGuard{g} \land \chi_{g'})) \in \NGCk$. Since $\Dom{f} = \Dom{g'} \setminus S$ it follows from rule \ref{item:syntaxdef:6} that $(\LogGuard{f} \land \phi_f) \in \NGCk$.

        It is straightforward to verify that $\free{(\LogGuard{f}
          \land \phi_f)}=\free{(\LogGuard{f} \land \phi)}$: \
The red variables in $\free{(\LogGuard{f} \land \phi_f)}$ are exactly the red variables in
        $\free{\LogGuard{f}}$, and these are exactly the red variables
        in $\free{(\LogGuard{f} \land \phi)}$.
The blue variables in $\free{(\LogGuard{f} \land \phi_f)}$ are exactly
the blue variables in $\free{\LogGuard{f}}$ or in
$\free{\LogGuard{g'}}$ or in $\free{\LogGuard{g}}$ or in
$\free{\chi_{g'}}$. These are exactly the blue variables in $\free{\LogGuard{f}}$ or in
$\free{\LogGuard{g}}$ or in $\free{(\LogGuard{g'}\und \chi_{g'})}$. By
induction hypothesis, these are exactly the blue variables in  $\free{\LogGuard{f}}$ or in
$\free{\LogGuard{g}}$ or in $\free{(\LogGuard{g'}\und \chi)}$. These,
in turn, are exactly the blue variables in $\free{\LogGuard{f}}$ or in
$\free{(\LogGuard{g}\und \chi)}$ or in
$\free{\LogGuard{g'}}$. By definition of $g'$ we have
$\free{\LogGuard{g'}}\subseteq\free{\LogGuard{f}}\cup\free{\LogGuard{g}}$. Thus,
the blue variables in  $\free{\LogGuard{f}}$ or in
$\free{(\LogGuard{g}\und \chi)}$ or in
$\free{\LogGuard{g'}}$ are exactly the blue variables in $\free{\LogGuard{f}}$ or in
$\free{(\LogGuard{g}\und \chi)}$; and these are exactly the blue
variables in $\free{(\LogGuard{f}\und \phi)}$.
        
All that remains to be done is to verify that $(\LogGuard{f} \land
\phi_f)\equiv(\LogGuard{f} \land \phi)$: \
Consider an arbitrary interpretation $\IntI \isdef (I,\beta)$.
If  $\IntI \not\models \LogGuard{f}$, then $\IntI
\not\models (\LogGuard{f} \land \phi)$ and $\IntI \not\models
(\LogGuard{f} \land \phi_f)$. If $\IntI \models \LogGuard{f}$,
then $\IntI \models (\LogGuard{f} \land \phi)$ $\iff$ $\IntI
\models \phi$, and $\IntI \models (\LogGuard{f} \land \phi_f)$
$\iff$ $\IntI \models \phi_f$.
We show that $\IntI \models \phi$ $\iff$ $\IntI \models \phi_f$.
For this it suffices to show that
for all $\tupel{a}=(a_1,\ldots,a_\ell) \in \rV{I}^\ell$ we have $\IntI^{\tupel{a}}
\models (\LogGuard{g} \land \chi)$ $\iff$ $\IntI^{\tupel{a}}
\models (\LogGuard{g'} \land (\LogGuard{g} \land \chi_{g'}))$,
where $\IntI^{\tupel{a}}\deff (I, \beta')$ with
$\beta'(\varv_{i_j}) = a_j$ for all $j \in [\ell]$ and
$\beta'(\varx) = \beta(\varx)$ for all variables $\varx\in\Var\setminus\set{\varv_{i_1},\ldots,\varv_{i_\ell}}$.
Recall that $\IntI\models\LogGuard{f}$ and $\Dom{f}=\Dom{g}\setminus
S=\Dom{g'}\setminus S$. Thus, also $\IntI^{\tupel{a}}\models\LogGuard{f}$.

For \enquote{$\Longrightarrow$} assume that $\IntI^{\tupel{a}} \models
(\LogGuard{g} \land \chi)$. Since for all $i \in S$ we have $g'(i) =
g(i)$ and, furthermore, $g'=f\cup g$ and
$\IntI^{\tupel{a}}\models\LogGuard{f}$, it also holds that
$\IntI^{\tupel{a}} \models (\LogGuard{g'} 
\land \chi)$. Thus, by induction hypothesis it holds that
$\IntI^{\tupel{a}} \models (\LogGuard{g'} \land \chi_{g'})$. Hence, we have
$\IntI^{\tupel{a}} \models (\LogGuard{g'} \land (\LogGuard{g}
\land \chi_{g'}))$. 

	For \enquote{$\Longleftarrow$} assume that $\IntI^{\tupel{a}}
        \models (\LogGuard{g'} \land (\LogGuard{g} \land
        \chi_{g'}))$. By induction hypothesis we get
        $\IntI^{\tupel{a}} \models (\LogGuard{g'} \land \chi)$. Hence,
        we have $\IntI^{\tupel{a}} \models (\LogGuard{g} \land \chi)$.

	In summary, we have shown that $(\LogGuard{f} \land \phi)
        \equiv (\LogGuard{f} \land \phi_f)$; this completes Case~3.\bigskip

\noindent\textbf{Case 4}: $\phi$ is of the form
$\existsi[n](\vare_{i_1}, \dots, \vare_{i_\ell}).(\LogGuard{g} \land
\chi)$, where $\Dom{g}=\indfreeR{\chi}$ and $i_1<\cdots<i_\ell$ and
$i_1,\ldots,i_\ell\in\indfreeB{(\LogGuard{g}\und\chi)}$. 
Let $f$ be an arbitrary guard function with $\Dom{f}=\indfreeR{\phi}$.
Note that $\Dom{f}=\Dom{g}$.

Let $S \isdef \set{ i_1, \dots, i_\ell }$.
Let $g'$ be the
guard function with $\Dom{g'} = \Dom{f} = \Dom{g}$ and 
	\begin{equation*}
		g'(i)\ \isdef \ \begin{cases}
			g(i), &\text{if } f(i) \in S \text{ or } f(i) \not\in \Img{g} \\
			f(i), &\text{otherwise, i.e. } f(i) \not\in S \text{ and } f(i) \in \Img{g}.
		\end{cases}
	\end{equation*}
Let
\[        
  \phi_f \quad \isdef \quad
  \existsi[n](\vare_{i_1}, \dots, \vare_{i_\ell}).(\LogGuard{g'} \land
  (\LogGuard{g} \land \chi_{g'}))\,.
\]  
\smallskip

\noindent\emph{Claim 1:} \
  $(\LogGuard{f} \land \phi) \equiv (\LogGuard{f} \land \phi_f)$.
\smallskip

\noindent\emph{Proof of Claim 1}: Consider an arbitrary interpretation $\IntI = (\I, \beta)$. If $\IntI \not\models \LogGuard{f}$, then $\IntI \not\models (\LogGuard{f} \land \phi)$ and $\IntI \not\models (\LogGuard{f} \land \phi_f)$.
If $\IntI \models \LogGuard{f}$, then $\IntI \models (\LogGuard{f}
\land \phi)$ $\iff$ $\IntI \models \phi$, and $\IntI \models (\LogGuard{f}
\land \phi_f)$ $\iff$ $\IntI \models \phi_f$. We show that $\IntI \models
\phi$ iff $\IntI \models \phi_f$. 

For this it suffices to show that for all $\tupel{u}=(u_1,\ldots,u_\ell) \in \bV{I}^\ell$
we have $\IntI^{\tupel{u}} \models (\LogGuard{g} \land \chi)$ $\iff$
$\IntI^{\tupel{u}} \models (\LogGuard{g'} \land (\LogGuard{g} \land
\chi_{g'}))$, where $\IntI^{\tupel{u}} = (I, \beta')$ with
$\beta'(\vare_{i_j}) = u_j$ for all $j \in [\ell]$ and $\beta'(\varx) =
\beta(\varx)$ for all variables $\varx\in\Var\setminus\set{\vare_{i_1},\ldots,\vare_{i_\ell}}$.

For \enquote{$\Longleftarrow$} assume that $\IntI^{\tupel{u}} \models
(\LogGuard{g'} \land (\LogGuard{g} \land \chi_{g'}))$. In particular,
$\IntI^{\tupel{u}} \models (\LogGuard{g'} \land \chi_{g'})$. By
induction hypothesis we obtain: $\IntI^{\tupel{u}} \models
(\LogGuard{g'} \land \chi)$. Hence, we have $\IntI^{\tupel{u}} \models
(\LogGuard{g} \land \chi)$. 
 
For \enquote{$\Longrightarrow$} assume that $\IntI^{\tupel{u}} \models
(\LogGuard{g} \land \chi)$. It suffices to show that
$\IntI^{\tupel{u}} \models \LogGuard{g'}$, because then we can proceed
as follows: We then have $\IntI^{\tupel{u}} \models
(\LogGuard{g'} \land \chi)$, and hence the 
induction hypothesis yields that $\IntI^{\tupel{u}} \models
(\LogGuard{g'} \land \chi_{g'})$, and thus $\IntI^{\tupel{u}} \models
(\LogGuard{g'} \land (\LogGuard{g} \land \chi_{g'}))$. 

All that remains to be done is to prove that $\IntI^{\tupel{u}} \models \LogGuard{g'}$.
We know that $\IntI \models \LogGuard{f}$ and $\IntI^{\tupel{u}}
\models \LogGuard{g}$. We have to show that $\IntI^{\tupel{u}} \models
E(\vare_{g'(i)}, \varv_i)$ for every $i \in \Dom{g'}$. Consider an
arbitrary $i\in\Dom{g'}$. If $g'(i) =
g(i)$ we are done since $\IntI^{\tupel{u}} \models \LogGuard{g}$. If
$g'(i) \neq g(i)$ then, by definition of $g'$ we have: $g'(i) =
f(i)\not\in S$.
We are done
since $\beta'(\vare_{g'(i)}) = \beta'(\vare_{f(i)}) =
\beta(\vare_{f(i)})$ and $\IntI \models \LogGuard{f}$. 
\qed\textsubscript{Claim 1}

\bigskip

\noindent\emph{Claim 2}: $(\LogGuard{f} \land \phi_f) \in \NGCk$.
\smallskip        

\noindent\emph{Proof of Claim 2}: By induction hypothesis we have
$(\LogGuard{g'} \land \chi_{g'}) \in \NGCk$. It is easy to see that
$(\LogGuard{g'} \land \LogGuard{g}) \in \NGCk$, and hence also
$(\LogGuard{g'} \land (\LogGuard{g} \land \chi_{g'})) \in \NGCk$. We
want to apply rule \ref{item:syntaxdef:7} of the inductive definition
of $\NGCk$ to show that $(\LogGuard{f} \land \phi_f) \in \NGCk$. For
this we have to show that all $i \in \Dom{g'}$ satisfy: 
	\begin{equation*}
		f(i) = g'(i) \quad\text{or}\quad f(i) \in S \quad\text{or}\quad f(i) \not\in \Img{g'}\,.
	\end{equation*}
	Assume for contradiction, that there exists an $i \in \Dom{g'}$ such that:
	\begin{align}
		\qquad \qquad f(i) \neq g'(i) \label{thm:ngck_is_equivalent:eq:1} \\
		 \text{and}\quad\qquad\qquad f(i) \not\in S \label{thm:ngck_is_equivalent:eq:2} \\
		 \text{and} \qquad\; f(i) \in \Img{g'}. \label{thm:ngck_is_equivalent:eq:3}
	\end{align}
	Then, by definition of $g'$ we must have $g'(i) = g(i)$ and
\begin{equation}\label{thm:ngck_is_equivalent:eq:4}
        f(i) \not\in \Img{g}.
\end{equation}
        By \eqref{thm:ngck_is_equivalent:eq:3}
        there exists a $j \in \Dom{g'}$ such that $f(i) = g'(j)$; and
        since $f(i) \not\in \Img{g}$ it must hold that $g'(j) \neq
        g(j)$ (because otherwise $f(i) = g'(j) = g(j) \in \Img{g}$,
        which is obviously contradictory). 

        From $g'(j)\neq g(j)$ and 
	the definition of $g'$ it follows that $g'(j) = f(j)$. Thus,
        $f(i) = g'(j) = f(j)$, i.e. $f(i) = f(j)$. Thus, from
        \eqref{thm:ngck_is_equivalent:eq:4} we obtain
        that $f(j) \not\in \Img{g}$. This means that by
        definition of $g'$ we have $g'(j) = g(j)$. This is a
        contradiction, since we already know that $g'(j)\neq g(j)$.

	In summary, we obtain that such an $i$ cannot exist, i.e. all $i \in \Dom{g'}$ satisfy the conditions imposed by rule \ref{item:syntaxdef:7}. Hence, $(\LogGuard{f} \land \phi_f) \in \NGCk$.
	\qed\textsubscript{Claim 2}

        \bigskip
        
\noindent\emph{Claim 3:} \
  $\free{(\LogGuard{f} \land \phi)} \ = \ \free{(\LogGuard{f} \land \phi_f)}$.
\smallskip

\noindent\emph{Proof of Claim 3}:
The red variables in $\free{(\LogGuard{f} \land \phi)}$ are exactly
the red variables in $\free{\LogGuard{f}}$, and these are exactly the red
variables in $\free{(\LogGuard{f} \land \phi_f)}$.

The blue variables in $\free{(\LogGuard{f} \land \phi_f)}$ are exactly
the blue variables in $\free{\LogGuard{f}}$ or in $\free{\phi_f}$.
Let $V\isdef\set{\vare_{i_1},\ldots,\vare_{i_\ell}}$.
Note that $\free{\phi_f} = \big( \free{\LogGuard{g}}\cup
\free{(\LogGuard{g'}\und \chi_{g'})}\big) \setminus V$.
By induction hypothesis, $\free{(\LogGuard{g'}\und
  \chi_{g'})}=\free{(\LogGuard{g'}\und \chi)}$. Hence,
\begin{eqnarray*}
  \free{\phi_f}
  &  =
  & \big( \free{\LogGuard{g}}\cup \free{(\LogGuard{g'}\und \chi)}\big)
    \setminus V
\\
  &  =
 & \big(
 \free{(\LogGuard{g}\und\chi)}\cup\free{\LogGuard{g'}}
   \big)\setminus V
\\
  & =
 & \big(\free{(\LogGuard{g}\und\chi)}\setminus V \big) \cup \big(
 \free{\LogGuard{g'}} \setminus V \big)
  \\
  & =
 & \free{\phi} \cup \big( \free{\LogGuard{g'}} \setminus V \big).
\end{eqnarray*}
To complete the proof it suffices to show that every blue variable in
$\free{\LogGuard{g'}} \setminus V$
belongs to $\free{\LogGuard{f}}\cup\free{\phi}$.
Blue variables in $\free{\LogGuard{g'}} \setminus V$ are of the form
$\vare_{g'(i)}$ with $i\in\Dom{g'}=\Dom{g}=\Dom{f}$ and $g'(i)\not\in S$. By the
definition of $g'$ we know that $g'(i)=f(i)$ or $g'(i)=g(i)$.
\\
If $g'(i)=f(i)$ then
$\vare_{g'(i)}=\vare_{f(i)}\in\free{\LogGuard{f}}$, and we are done.
\\
If $g'(i)=g(i)$ then $g(i)\not\in S$, and hence
$\vare_{g'(i)}=\vare_{g(i)}\in\free{\phi}$, and we are done. 
\qed\textsubscript{Claim 3}

\medskip
        
	This completes the proof of Theorem~\ref{thm:ngck_is_equivalent}.\qedhere
\end{proof}

\section{An Inductive Characterisation of $\IEHW{k}$}
\label{sec:recursive_def}

In this section we give an inductive definition of what we call
\emph{guarded $k$-labeled incidence graphs}, and  we prove that these are
equivalent to the incidence graphs of entangled hypertree width $\leq k$.  
Throughout this section, we fix an arbitrary number $k\in\NNpos$.

\subsection{$k$-Labeled Incidence Graphs and the Class $\GLIk$}\label{sec:defining_k_guarded}

We enrich an incidence graph $\I$ by labeling some of its blue nodes with labels
in $[k]$, by labeling some of its red nodes with labels in $\NNpos$, and by providing,
for every $i\in\NNpos$ that is used as a label for a red node,
a ``blue label'' $g(i)\in[k]$ that should be regarded as ``the guard'' of $i$.
Each label can only be
used once,
not all labels have to be used, not all vertices have to be labeled,
one vertex may have multiple labels, and ``guards'' can be chosen
arbitrarily. 
This is formalised as follows.

\begin{definition}\label{def:klabeledincidencegraph}
A \emph{$k$-labeled incidence graph} $\kLI = (\I, r, b, g)$ consists of an
incidence graph $\I$ and partial mappings
$r: \natpos \pto \rVI$, 
$b: [k] \pto \bVI$,  and
$g: \natpos \pto [k]$ such that
$\Def{g} = \Def{r}$ is finite.
\\
We write $\I_\kLI,r_\kLI, b_\kLI, g_\kLI$ to address $\kLI$'s components $\I,r,b,g$.
\end{definition}

Let $\kLI=(\I,r,b,g)$ be a $k$-labeled incidence graph.
If $j\in\Dom{b}$, then the blue node $b(j)$ of $\I$ is labeled with the number $j$.
If $i\in\Dom{r}$, then the red node $r(i)$ of $\I$ is labeled with the
number $i$, and $g(i)=j$ indicates that the blue node labeled with the
number $j$ (if it exists) should be regarded as ``the guard'' of the
red node labeled with the number $i$. 

We say that $\kLI$ \emph{has real guards} if for every $i\in\Dom{r}$
the red node $v$ labeled $i$ is ``guarded'' by the blue node $e$
labeled $j\deff g(i)$ in the sense that $\I$ contains an edge from $e$ to $v$.
This is formalised in the following definition.

\begin{definition}
A $k$-labeled incidence graph $\kLI=(\I,r,b,g)$ is said to
\emph{have real guards w.r.t.\ $f$} for a partial function $f:\NNpos\pto[k]$
if
$\Dom{f}\subseteq\Dom{r}$ and
for all $i \in \Def{f}$ we have $f(i) \in \Def{b}$ and $(b(f(i)), r(i)) \in \E{\I}$.
\\
We say that $\kLI$ \emph{has real guards} if it has real guards w.r.t.\ $g$.
\end{definition}

Particularly simple examples of $k$-labeled incidence graphs with real guards are provided by the following definition.

\begin{definition}\label{def:kLIf}
Let $f:\NNpos\pto [k]$ with finite, non-empty $\Dom{f}$. The $k$-labeled incidence graph $\kLIf$ \emph{defined by $f$} is
the $k$-labeled incidence graph
$\kLI=(\I,r,b,g)$ with $g\deff f$, 
where $\I$ consists of a red node $v_i$ for every $i\in\Dom{f}$, a
blue node $e_j$ for every $j\in\Img{f}$, and an edge $(e_{f(i)},v_i)$
for every $i\in\Dom{f}$, and where $\Dom{r}=\Dom{f}$ and $r(i)=v_i$
for all $i\in\Dom{r}$, and $\Dom{b}=\Img{f}$ and $b(j)=e_j$ for all
$j\in \Dom{b}$. 
Note that $\kLIf$ has real guards. 
\end{definition}

We introduce a number of operations on $k$-labeled incidence graphs. The
first kind of operations provides ways to modify the labels (the latter
two of these do not necessarily preserve real guards).

\begin{definition}[Changing labels]
\label{def:changinglabels}
\ \\
Let $\kLI = (\I, r, b, g)$ be a $k$-labeled incidence graph.\\
Let $\myR \subseteq \natpos$ be finite, and let $\myB \subseteq [k]$.

\begin{enumerate}[$-$]
\item
  Removing from the red nodes all the labels in $\myR$ is achieved by
  the operation\\
  $\reclaimR{\kLI}{\myR}\deff(\I,r',b,g')$ with $r'\deff r-\myR$ and $g'\deff g-\myR$.
  
\item
  Removing from the blue nodes all the labels in $\myB$ is
  achieved by the operation
  $\reclaimB{\kLI}{\myB}\deff(\I,r,b',g)$ with $b'\deff b-\myB$.

\item
  Let $\myR=\set{i_1,\ldots,i_\ell}$ for $\ell\deff |\myR|$ and $i_1<\cdots<i_\ell$.
  
  For every $\tupel{v}=(v_1,\ldots,v_\ell)\in \rVI^\ell$ we let
  $\reseatR{\kLI}{\myR}{\tupel{v}}\deff (\I,r',b,g)$ with
  $\Dom{r'}=\Dom{r} \union \myR$
  and $r'(i_j)=v_j$ for all $j\in[\ell]$ and $r'(i)=r(i)$ for all $i\in\Dom{r}\setminus \myR$
  (i.e., for each $j\in[\ell]$, the red label $i_j$ is moved onto the
  red node $v_j$, and all other labels remain unchanged).
  
\item
  Let $\myB=\set{i_1,\ldots,i_\ell}$ for $\ell\deff |\myB|$ and $i_1<\cdots<i_\ell$.
  
  For every $\tupel{e}=(e_1,\ldots,e_\ell)\in \bVI^\ell$ we let
  $\reseatB{\kLI}{\myB}{\tupel{e}}\deff (\I,r,b',g)$ with
  $\Dom{b'}=\Dom{b} \union \myB$
  and $b'(i_j)=e_j$ for all $j\in[\ell]$ and $b'(i)=b(i)$ for all $i\in\Dom{b}\setminus \myB$
  (i.e., for each $j\in[\ell]$, the blue label $i_j$ is moved onto the blue node $e_j$, and all other labels remain unchanged).  
\end{enumerate}
\end{definition}

The next operation enables us to \emph{glue} two $k$-labeled
incidence graphs $\kLI_1$ and $\kLI_2$. This is achieved by first taking the 
disjoint union of $\kLI_1$ and $\kLI_2$ and then merging all red
(blue) nodes that carry the same label into a single red (blue) node
that inherits all neighbours of the merged nodes; we write
$\glue{\kLI_1}{\kLI_2}$ to denote the resulting $k$-labeled incidence graph.
This is achieved by the following formal definition.

\begin{definition}[Glueing $k$-labeled incidence graphs]
\label{def:glueing}
\ \\
Let $\kLI_i \isdef (\inzidenz{i}, r_i, b_i, g_i)$
be a $k$-labeled incidence graph for $i\in[2]$. The \emph{glueing
  operation} produces the $k$-labeled incidence graph
$\glue{\kLI_1}{\kLI_2} \isdef (\inzidenz{}, r, b, g)$ in the following
way: 
\smallskip

\noindent
Let $\inzidenz{}' \isdef \inzidenz{1} \disunion \inzidenz{2}$ be
the disjoint union of $\I_1$ and $\I_2$. 
I.e., $\rV{\I'}=\bigcup_{i\in[2]}(\rV{\I_i}\times\set{i})$,
$\bV{\I'}=\bigcup_{i\in[2]}(\bV{\I_i}\times\set{i})$,
$\E{\I'}=\bigcup_{i\in[2]}\setc{\big((e,i),(v,i)\big)}{(e,v)\in\E{\I_i}}$.
\smallskip

\noindent
Let $\sim_R$ be the equivalence relation on $\rV{\I'}$ obtained as the reflexive, symmetric, and transitive closure of the relation
$\bigsetc{\big((r_1(j),1),(r_2(j),2)\big)}{j\in\Dom{r_1}\cap\Dom{r_2}}$.
\\
Let $[v]_{\sim_R}$ denote the equivalence class of each $v\in\rV{\I'}$.
\\
We let $\rVI\deff\setc{[v]_{\sim_R}}{v\in\rV{\I'}}$.
\smallskip
  
\noindent
Let $\sim_B$ be the equivalence relation on $\bV{\I'}$ obtained as the reflexive, symmetric, and transitive closure of the relation
$\bigsetc{\big((b_1(j),1),(b_2(j),2)\big)}{j\in\Dom{b_1}\cap\Dom{b_2}}$.
\\
Let $[e]_{\sim_B}$ denote the equivalence class of each $e\in\bV{\I'}$.
\\
We let $\bVI\deff\setc{[e]_{\sim_B}}{e\in\bV{\I'}}$.
\smallskip
  
\noindent
$\E{\I}\deff\bigsetc{\big([(e,i)]_{\sim_B}\,,\,[(v,i)]_{\sim_R}\big)}{i\in [2],\ (e,v)\in\E{\I_i}}$.
\smallskip

\noindent
Let $\I=(\rVI,\bVI,\E{\I})$.
\smallskip

\noindent
Let $r:\NNpos\pto \rVI$ with $\Dom{r}=\Dom{r_1}\cup\Dom{r_2}$ be defined by
$r(j)\deff [(r_1(j),1)]_{\sim_R}$ for all $j\in\Dom{r_1}$ and
$r(j)\deff [(r_2(j),2)]_{\sim_R}$ for all $j\in\Dom{r_2}\setminus\Dom{r_1}$.
\smallskip

\noindent
Let $b:\NNpos\pto \bVI$ with $\Dom{b}=\Dom{b_1}\cup\Dom{b_2}$ be defined by
$b(j)\deff [(b_1(j),1)]_{\sim_B}$ for all $j\in\Dom{b_1}$ and
$b(j)\deff [(b_2(j),2)]_{\sim_B}$ for all $j\in\Dom{b_2}\setminus\Dom{b_1}$.
\smallskip

\noindent
Let $g\deff g_1\cup g_2$ (recall from
Section~\ref{section:FurtherNotation} that
this means that $g(j)=g_1(j)$ for all $j\in\Dom{g_1}$ and $g(j)=g_2(j)$ for all
$j\in\Dom{g_2}\setminus\Dom{g_1}$).
\smallskip

\noindent
Finally, we let $\glue{\kLI_1}{\kLI_2} \isdef (\inzidenz{}, r, b, g)$.
\smallskip

\noindent
For $i\in[2]$ we define mappings
$\pi_{i,R}:\rV{\I_i}\to\rV{\I}$ and
$\pi_{i,B}:\bV{\I_i}\to\bV{\I}$ by
$\pi_{i,R}(v)\deff [(v,i)]_{\sim_R}$ for all $v\in\rV{\I_i}$ and
$\pi_{i,B}(e)\deff [(e,i)]_{\sim B}$ for all $e\in\bV{\I_i}$.
Note that $\pi_{i,R}(v)$ is the red node of
$\I_{\glue{\kLI_1}{\kLI_2}}$ that corresponds to $v\in\rV{\I_i}$, and
$\pi_{i,B}(e)$ is the blue node of
$\I_{\glue{\kLI_1}{\kLI_2}}$ that corresponds to $e\in\bV{\I_i}$.
\end{definition}

The following lemma's proof is straightforward (see Appendix~\ref{appendix:glueing_preserves_real_guards}).

\begin{lemma}\label{lemma:glueing_preserves_real_guards}
If $\kLI_1$ and $\kLI_2$ are $k$-labeled incidence graphs with real guards, then also
$\glue{\kLI_1}{\kLI_2}$ has real guards.  
\end{lemma}

We need one further operation on $k$-labeled incidence graphs, namely,
one that admits us to change its guard function.
This is provided by the following definitions.

\begin{definition}[Transition]\label{def:newtransition}
Consider a partial function $g:\NNpos\pto [k]$.
A \emph{transition for $g$} is a partial function $f:\NNpos\pto [k]$
with $\emptyset \neq \Dom{f} \subseteq \Dom{g}$
satisfying the following:
for every $i\in \Dom{g}$ with $g(i)\in\Img{f}$ we have $i\in\Dom{f}$.
\end{definition}

\begin{definition}[Applying a transition]\label{def:transition}
\ \\
Let $\kLI = (\I, r, b, g)$ be a $k$-labeled incidence graph
and let $f$ be a transition for $g$. 
We can \emph{apply} the transition $f$
to $\kLI$ and obtain the
$k$-labeled incidence graph
$\switch{\kLI}{f} \isdef \glue{\,\kLIf\,}{\; \reclaimB{\kLI}{\myB}\,}$\,,
where  $\myB \isdef \Img{g} \intersect \Img{f}\intersect \Dom{b}$
and $\kLIf$ is provided by Definition~\ref{def:kLIf}.
\end{definition}

The idea of applying a transition $f$ to a $k$-labeled incidence graph
$\kLI=(\I,r,b,g)$ is to assign new guards to a set of 
labeled red vertices (i.e.\ the domain of $f$). These new guards
should by newly inserted nodes, and they should 
be \emph{real} guards. To this end, for every $j\in\Img{f}$ we add a new blue
node $e'_j$ labeled $j$; and in case that the label $j$ had already
been used by a blue node $e$ of $\kLI$ (i.e., $j\in\Dom{b}$)
\emph{and} served as a guard according to $g$ (i.e., $j\in\Img{g}$), we
remove this label from $e$. 
For each $i\in\Dom{f}$ with $f(i)=j$ we add an edge from the red
node of $\kLI$ labeled $i$ to the new blue node $e'_j$.

The formal definition $\switch{\kLI}{f} \isdef \glue{\kLIf}{\reclaimB{\kLI}{\myB}}$
achieves this as follows:
By $\reclaimB{\kLI}{\myB}$ we release from $\kLI$ all blue labels $j$
that are present in $\kLI$ and that we want to assign to newly created
nodes. This is achieved by letting
$\myB=\Img{g}\intersect\Img{f}\intersect\Dom{b}$.
Afterwards, adding the edges from the nodes of $\kLI$ that carry a
red label $i\in\Dom{f}$ to the new blue node $e'_{f(i)}$ is achieved
by glueing $\kLIf$ to $\reclaimB{\kLI}{\myB}$.

Note that releasing from $\kLI$ all blue labels in $\myB$ might be
problematic: 
Consider a red node $v$ labeled $i$ that 
was originally guarded by the blue node $e$ of $\kLI$ that carried the label
$j\deff g(i)$. Releasing the label $j$ from node $e$ means that $v$
loses its guard in case that $i\not\in\Dom{f}$. 
Therefore, for $f$ to be a transition for
$g$, we require in Definition~\ref{def:newtransition} that it assigns a new guard to all the affected
labeled red vertices, i.e.\ we require $i \in \Dom{f}$, if $g(i) \in
\Img{f}$. 

The next lemma's proof is straightforward.

\begin{lemma}\label{lemma:transition_preserves_real_guards}
If $\kLI$ is a $k$-labeled incidence graph with real guards and $f$ is a transition for
$g_\kLI$, then $\switch{\kLI}{f}$ has real guards.
\end{lemma}  

Now we have available all the tools needed to define the class $\GLIk$
of \emph{guarded $k$-labeled incidence graphs}.

\begin{definition}[$\GLIk$]\label{def:k_guarded_incidence_graphs}
The class $\GLIk$ of 
\emph{guarded $k$-labeled incidence graphs} is inductively defined as follows.
\smallskip

\noindent \textbf{Base case}:
\begin{enumerate}
\item
Any $k$-labeled incidence graph $\kLI = (\I, r, b, g)$
with  $\rV{\inzidenz{}} = \Img{r}$, $\bV{\inzidenz{}} = \Img{b}$, and
with real guards belongs to $\GLIk$.
\setcounter{KGuardedCounter}{\value{enumi}}
\end{enumerate}

\noindent \textbf{Inductive cases}:
Let $\kLI = (\I, r, b, g)\in \GLIk$.
\begin{enumerate}
\setcounter{enumi}{\value{KGuardedCounter}}
\item\label{item:GLI-reclaimR} $\reclaimR{\kLI}{\myR} \in \GLIk$ for every 
$\myR \subseteq \Def{r}$.
\item\label{item:GLI-reclaimB} $\reclaimB{\kLI}{\myB} \in \GLIk$ for every 
$\myB \subseteq \Def{b} \setminus \Img{g}$.
\item\label{item:GLI-switch} $\switch{\kLI}{f} \in \GLIk$ for every transition $f$ for $g$.
\item\label{item:GLI-glue} $\glue{\kLI}{\kLI'} \in \GLIk$ for every
  $\kLI'=(\I',r',b',g')\in\GLIk$ such that $g$ and $g'$ are
  compatible. 
\end{enumerate}
\end{definition} 
\smallskip

\noindent
By induction, and using the
Lemma~\ref{lemma:glueing_preserves_real_guards} and
\ref{lemma:transition_preserves_real_guards}, one easily obtains: 

\begin{lemma}\label{lemma:GLIk_has_real_guards}
Every $\kLI\in\GLIk$ has real guards.
\end{lemma}

\subsection{$\GLIk$ Corresponds to $\IEHWk$}
\label{sec:k_guarded_characterizes_ehw}

Recall from Section~\ref{sec:preliminaries} that
$\IEHWk$ is the class of all incidence graphs $\I$ with $\ehw{\I}\leq k$, where
$\ehw{\I}$ is $\I$'s entangled hypertree width (cf.\ Definition~\ref{def:ehw}).
This subsection is devoted to proving that $\IEHWk$ is characterised
by $\GLIk$ (Theorem~\ref{thm:IEHWkIsGLIk}). A $k$-labeled incidence
graph $\kLI$ is called \emph{label-free} if $\Dom{r_\kLI}=\Dom{b_\kLI}=\Dom{g_{\kLI}}=\emptyset$.

\begin{theorem}\label{thm:IEHWkIsGLIk}\label{lemma:k_guarded_subset_ehw}\label{lemma:ehw_subset_k_guarded} \mbox{}
\begin{mea}
\item\label{item:GLIkInIEHWk}
The incidence graph $\I_\kLI$ of every $\kLI\in\GLIk$ is in $\IEHWk$.
\item\label{item:IEHWkInGLIk}
For every incidence graph $\I\in\IEHWk$ there exists a label-free
$\kLI\in\GLIk$ such that $\I\isom\I_{\kLI}$.
\end{mea}    
\end{theorem}

\begin{proof}[Proof idea]
\ref{item:GLIkInIEHWk}:\, 
By induction on the definition of $\GLIk$ we show for every
$\kLI\in\GLIk$ that there exists an 
ehd $D=(T,\DBag,\DCover)$ of $\I_\kLI$ of width $\leq k$.
To be able to carry out the induction step, we additionally ensure
that there is a tree-node $\omega\in\V{T}$ with
$\DCover(\omega)=\Img{b_\kLI}$. 
The idea is fairly simple: From Lemma~\ref{lemma:GLIk_has_real_guards}
we know that $\kLI$ has real guards. Thus, the (at most $k$) blue
labels given by $b$ guard all the red labels given by $r$, i.e.,
$(b({g(i)}),r(i))\in\E{\I}$ for every $i\in\Dom{r}$. Hence, we
can use the set of labeled blue vertices to cover the bag containing
all the labeled red vertices.
See Appendix~\ref{appendix:GLIk-in-IEHWk} for a detailed proof.
\smallskip

\ref{item:IEHWkInGLIk}:\,
Given an $\I\in\IEHWk$, the proof proceeds as follows.
We pick an ehd $D=(T,\DBag,\DCover)$ of $\I$ of width $\leq k$ that
has a particularly suitable shape, namely, there exists 
a tree-node $\omega$ such that the rooted tree $(T,\omega)$
is binary (i.e., every node has at most 2 children) and monotone
(i.e., for all parent-child pairs $(t_p,t_c)$ we have
$|\DCover(t_p)|\geq |\DCover(t_c)|$).

For every $t\in\V{T}$, the information provided by $\DCover(t)$ is viewed
as a description of a guarded $k$-labeled incidence graph $\kLI_t$
corresponding to the base case of
Definition~\ref{def:k_guarded_incidence_graphs}. 
We then perform a bottom-up traversal of the rooted tree $(T,\omega)$ and glue 
together all the $\kLI_t$'s.
But this has to be done with care: before glueing them, we have to
ensure that the guards are compatible; we achieve this by adequately
changing labels and applying transitions, so that finally we end up
with an $\kLI\in\GLIk$ whose incidence graph $\I_\kLI$ is isomorphic
to $\I$.
Appendix~\ref{appendix:IEHWk-in-GLIk} provides a detailed description
of the construction; 
a specific example of the construction can be found
in Appendix~\ref{appendix:subsec:example}.
\end{proof}  

\subsection{Homomorphisms on $k$-Labeled Incidence Graphs}\label{sec:hom_on_k_guarded}

We define the notion of \emph{homomorphisms} of $k$-labeled incidence
graphs in such a way that it respects labels, but ignores the guard
function. 

\begin{definition}\label{def:k_labeled_hom}
Let $\kLI=(\I,r,b,g)$ and $\kLI'=(\I',r',b',g')$ be $k$-labeled incidence graphs.
If $\Dom{r}\not\subseteq\Dom{r'}$ or $\Dom{b}\not\subseteq \Dom{b'}$,
then there exists no homomorphism from $\kLI$ to $\kLI'$.
Otherwise, a homomorphism from $\kLI$ to $\kLI'$ is a homomorphism $h=(h_R,h_B)$ from $\I$ to
$\I'$ satisfying the following condition:
$h(r(i)) = r'(i)$ for all $i\in\Dom{r}$ and
$h(b(j)) = b'(j)$ for all $j\in\Dom{b}$.
\end{definition}

By $\Hom{\kLI}{\kLI'}$ we denote the set of all homomorphisms from $\kLI$ to $\kLI'$,
and we let $\hom{\kLI}{\kLI'}:=|\Hom{\kLI}{\kLI'}|$ be the number of
homomorphisms from $\kLI$ to $\kLI'$. In particular, if $\kLI$ is
\emph{label-free}, then $\hom{\kLI}{\kLI'}=\hom{\I_{\kLI}}{\I_{\kLI'}}$.

In order to enable us to \enquote{aggregate} homomorphism counts, we
proceed in a similar way as Dvořák \cite{Dvorak2010}: we use a variant
of the 
\emph{quantum graphs} of Lovász and Szegedy \cite{Lovasz2008}, tailored towards our setting.
We say that $k$-labeled incidence graphs $\kLI_1,\ldots,\kLI_d$ are \emph{compatible} if
their labeling functions all have the same domain and they all have the same guard function, i.e.,
$\Dom{r_{\kLI_1}}=\Dom{r_{\kLI_i}}$, $\Dom{b_{\kLI_1}}=\Dom{b_{\kLI_i}}$, and $g_{\kLI_1}=g_{\kLI_i}$ for all $i\in[d]$. 

\begin{definition}
A \emph{$k$-labeled quantum incidence graph} $Q$ is a formal finite
non-empty linear combination
with real coefficients of compatible $k$-labeled incidence graphs. We represent a $k$-labeled quantum incidence graph $Q$ as
$ \sum_{i=1}^d \alpha_i\, \kLI_i$,
where $d\in\NNpos$, $\alpha_i\in\RR$, and $\kLI_i$ is a $k$-labeled incidence graph for $i\in[d]$.
We let 
$\domrQ\deff\Dom{r_{\kLI_1}}=\cdots=\Dom{r_{\kLI_d}}$,
$\dombQ\deff\Dom{b_{\kLI_1}}=\cdots=\Dom{b_{\kLI_d}}$, and
$g_Q\deff g_{\kLI_1}=\cdots=g_{\kLI_d}$.
The $\alpha_i$'s and $\kLI_i$'s are called the \emph{coefficients} and
\emph{components}, respectively, and $d$ is called the \emph{degree}
of $Q$. 
\end{definition}

Note that a $k$-labeled incidence graph is a $k$-labeled quantum incidence graph with
degree 1 and coefficient 1.

For a $k$-labeled quantum incidence graph $Q=\sum_{i=1}^d\alpha_i\kLI_i$ and an arbitrary $k$-labeled incidence graph $\kLI'$ we let
\[
  \hom{Q}{\kLI'} \ \ \deff \ \ \sum_{i=1}^d \alpha_i {\cdot} \hom{\kLI_i}{\kLI'} \ \ \ \in \ \RR.
\]  

We adapt the operations for $k$-labeled incidence graphs to their
quantum equivalent in the expected way:
\begin{eqnarray*}
\reclaimR{Q}{\myR} & \isdef & \displaystyle\sum_{i=1}^d
                              \reclaimR{\kLI_i}{\myR}\,,
\\                              
\reclaimB{Q}{\myB} & \isdef & \displaystyle\sum_{i=1}^d
                              \reclaimB{\kLI_i}{\myB}\,,
\\                              
\switch{Q}{f} &\isdef & \displaystyle \sum_{i=1}^d \alpha_i\switch{\kLI_i}{f}\,.
\end{eqnarray*}
                        
\noindent
Glueing two $k$-labeled quantum incidence graphs $Q = \sum_{i=1}^d \alpha_{i} \, \kLI_i$ and $Q'= \sum_{j =1}^{d'} \alpha'_{j} \, \kLI'_j$ is achieved by pairwise glueing of their components and multiplication of their respective coefficients, i.e.
\begin{displaymath}
\glue{Q}{Q'} \ \isdef \ \ \sum_{i\in[d] \atop j\in [d']} (\alpha_i {\cdot}\alpha'_j)\; \glue{\kLI_i}{\kLI'_j}\,.
\end{displaymath}

The following can easily be proved for the case where $Q,Q'$ have
degree 1 and coefficient 1 (i.e., $Q,Q'$ are $k$-labeled
incidence graphs), and 
then  be generalised to quantum incidence graphs by simple linear
arguments (for proof details see
Appendix~\ref{appendix:subsec:lemma:hom_of_glued_is_product}).

\begin{lemma}
\label{lem:HomomCountQ}\label{lem:hom_of_glued_is_product}\label{lem:hom_if_red_reclaimed}\label{lem:hom_if_blue_reclaimed}
For all $k$-labeled quantum incidence graphs $Q,Q'$ and all $k$-labeled incidence graphs  $\kLI$ we have: 
\begin{enumerate}
\item\label{item:homQuantumGlue}
  $\hom{\glue{Q}{Q'}}{\kLI} = \hom{Q}{\kLI} \mal \hom{Q'}{\kLI}$.
  \smallskip
\item
  $\hom{\reclaimR{Q}{\myR}}{\kLI} =
  \sum_{\tupel{v} \in \rV{\inzidenz{\kLI}}^{\ell}}
    \hom{Q}{\reseatR{\kLI}{\myR}{\tupel{v}}}$
  
    for all $\myR\subseteq \domrQ$ and $\ell\deff |\myR|$.
    \smallskip
\item
  $\hom{\reclaimB{Q}{\myB}}{\kLI} =
  \sum_{\tupel{e} \in \bV{\inzidenz{\kLI}}^{\ell}}
    \hom{Q}{\reseatB{\kLI}{\myB}{\tupel{e}}}$
  
    for all $\myB\subseteq \dombQ$ and $\ell\deff |\myB|$.
    \smallskip
\item
  $\hom{\switch{Q}{f}}{\kLI} =
   \hom{\kLIf}{\kLI}\cdot \sum_{\tupel{e} \in \bV{\I_{\kLI}}^\ell} \hom{Q}{\reseatB{\kLI}{\myB}{\tupel{e}}}$
  
  for all transitions $f$ for $g_Q$, for
$\myB\deff \dombQ\cap \Img{f}\cap\Img{g}$
and $\ell\deff |\myB|$.
Note that $\hom{\kLIf}{\kLI}\in\set{0,1}$.
\end{enumerate}
\end{lemma}

The class $\QGLIk$ of \emph{guarded $k$-labeled quantum incidence graphs} consists of those
$k$-labeled quantum incidence graphs where all components belong to $\GLIk$. 

The following lemma was provided for series-parallel quantum graphs by
Lovász and Szegedy \cite{Lovasz2008} and for labeled quantum graphs of
tree-width
$\leq k$ by Dvořák \cite{Dvorak2010}; their proof also works for
labeled quantum incidence graphs in $\QGLIk$.

\begin{lemma}\label{lemma:normalizing_quantum_incidence_graphs}
Let $X, Y \subseteq \nat$ be disjoint and finite, and let $Q\in\QGLIk$. There exists a
$Q[X,Y]\in\QGLIk$
with the
same parameters $\domrQ,\dombQ,g_Q$ as $Q$,
such that for all $k$-labeled incidence graphs $\kLI$
with
real guards w.r.t.\ $g_Q$
we have:
\begin{enumerate}
\item \ $\hom{Q}{\kLI} \in X \implies \hom{Q[X,Y]}{\kLI} = 0$\,,
\item \ $\hom{Q}{\kLI} \in Y \implies \hom{Q[X,Y]}{\kLI} = 1$\,.
\end{enumerate}
\end{lemma}
\begin{proof}[Proof idea]
Pick a polynomial
$p(x)$ such that $p(z)=0$ for all $z\in X$ and $p(z)=1$ for all $z\in Y$. If
$p(x)=\sum_{i=0}^d a_i x^i$, one chooses
$Q[X,Y]\deff \sum_{i=0}^d a_i Q^i$,
where $Q^i$ is the result of glueing $i$ copies of $Q$ together, and
$Q^0\deff \kLIfct{g_Q}$.
\end{proof}

\section{Proof of Theorem~\ref{thm:GCkVsHomomIndist}}
\label{sec:main_theorem}

Finally, we have available all the machinery so that, from a
high-level point of view, our 
proof of Theorem~\ref{thm:GCkVsHomomIndist} can follow a similar
approach as Dvořák's proof in \cite{Dvorak2010}. The technical details,
however, are quite intricate because the concept of generalised
hypertree width (as well as the classes $\IEHWk$ and $\GLIk$) is much
more complicated than the concept of tree-width.
Analogously to two main lemmas in \cite{Dvorak2010}, we provide a key
lemma for each of the directions ``$\Longleftarrow$'' and
``$\Longrightarrow$'' of Theorem~\ref{thm:GCkVsHomomIndist} formulated
for $\RGCk$ instead of $\GCk$.
These lemmas use the following notion:
The \emph{interpretation} $\myInt{\kLI'}$ associated with a
$k$-labeled incidence graph $\kLI'$ is an interpretation $(\I,\beta)$
with $\I\deff\I_{\kLI'}$ and 
$\beta(\varv_i)\deff r_{\kLI'}(i)$ for all $i\in\Dom{r_{\kLI'}}$ and
$\beta(\vare_j)\deff b_{\kLI'}(j)$ for all
$j\in\Dom{b_{\kLI'}}$.

\begin{lemma}\label{lemma:HomToFormula}
Let $\kLI = (\I, b, r, g)\in\GLIk$. For every $m \in \NN$ there is a formula
$\form{\kLI}{m}$ with $(\LogGuard{g}\und\form{\kLI}{m})\in\NGCk$ and
$\free{(\LogGuard{g}\und\form{\kLI}{m})}=
\setc{\varv_i}{i\in\Dom{r}}
\cup
\setc{\vare_j}{j\in\Dom{b}}$
such that
for every $k$-labeled incidence graph $\kLI'$
with
$\Def{b_{\kLI'}} \supseteq \Def{b}$,
$\Def{r_{\kLI'}} \supseteq \Def{r}$,
and
with real guards w.r.t.\ $g$ we have: \ $\myInt{\kLI'}\models\LogGuard{g}$ \ and
\begin{displaymath}
  \hom{\kLI}{\kLI'} = m \ \ \iff \ \
  \myInt{\kLI'} \models \form{\kLI}{m}\,.
\end{displaymath}
\end{lemma}
\begin{proof}[Proof sketch]
Throughout this proof we use the following notation:
A \emph{segmentation} for a number $m \in \natpos$ is a pair of $d$-tuples
\[
  (\tupel{c}, \tupel{m}) \ \ = \ \  \big((c_1, \dots, c_d), (m_1,
  \dots, m_d)\big)
\]
of numbers in $\NNpos$,
such that $d\in\NNpos$ and $m_1<\cdots <m_d$ and
$\sum_{i \in [d]} c_i m_i = m$.
We call $d$ the \emph{degree} of $(\tupel{c}, \tupel{m})$
and $c \deff \sum_{i \in [d]} c_i$ its \emph{size}.

By $\mySeg{m}$ we denote the set of all segmentations of $m$.
Note that $\mySeg{m}$ is finite and non-empty for every $m\in\NNpos$.

The proof of Lemma~\ref{lemma:HomToFormula} proceeds by induction based on Definition~\ref{def:k_guarded_incidence_graphs}.
The base case can be handled straightforwardly (see Appendix~\ref{appendix:subsec:HomToFormula}).
For the inductive step
we assume that the lemma's statement is already shown for $\kLI=(\I,r,b,g)\in\GLIk$ and
for $\kLI_i=(\I_i,r_i,b_i,g_i)\in\GLIk$ for $i\in[2]$ where $g_1,g_2$ are compatible.
The goal is to prove the lemma's statement for the particular
$\tilde{\kLI}=(\tilde{\I},\tilde{r},\tilde{b},\tilde{g})\in\GLIk$ 
considered in the following case distinction.
The case where $\tilde{\kLI}=\glue{\kLI_1}{\kLI_2}$ can be handled
rather easily by utilising part \ref{item:homQuantumGlue}) of Lemma~\ref{lem:HomomCountQ}.
The case where $\tilde{\kLI}=\reclaimR{\kLI}{\myR}$ for some
$\myR\subseteq\Dom{r}$ can be handled similarly as the following case.\smallskip

The case where $\tilde{\kLI}=\reclaimB{\kLI}{\myB}$ for some
$\myB\subseteq\Dom{b}\setminus\Img{g}$, is handled as follows.
If $\myB=\emptyset$ we are done since $\tilde{\kLI}=\kLI$.
If $\myB\neq \emptyset$ let $\ell\deff |\myB|$ and let
$\myB=\set{i_1,\ldots,i_\ell}$ with $i_1<\cdots<i_\ell$. Note that
$\tilde{r}=r$, $\tilde{b}=b-\myB$, and $\tilde{g}=g$. 
We let \,$
  \form{\tilde{\kLI}}{0} \deff 
  \forall (\vare_{i_1}, \dots, \vare_{i_\ell}).(\LogGuard{g} \impl
  \form{\kLI}{0})\,.$
For every $m\geq 1$ we let
\,$
  \form{\tilde{\kLI}}{m} \isdef
  \biglor_{(\tupel{c}, \tupel{m}) \in \mySeg{m}} \phi'_{\tupel{c}, \tupel{m}}
$\,
where for every
$(\tupel{c},\tupel{m})=((c_1,\ldots,c_d),(m_1,\ldots,m_d))\in\mySeg{m}$
and $c\deff\sum_{i\in [d]}c_i$,
\[
 \begin{array}{ll}
  \phi'_{\tupel{c},\tupel{m}} \ \ \isdef 
&
  \existsex[c] (\vare_{i_1}, \dots, \vare_{i_\ell}).
  (\LogGuard{g} \land \lnot \form{\kLI}{0}) \ \ \land
\\[1ex]
  & \bigland_{j\in[d]} \left(
    \existsex[c_j] (\vare_{i_1}, \dots, \vare_{i_\ell}).
    (\LogGuard{g} \land \form{\kLI}{m_j} \right)\,.
 \end{array}
\] 
Let $\kLI'=(\I',r',b',g')$ satisfy the lemma's assumptions for $\tilde{\kLI}$ instead of $\kLI$. I.e.,
we have
$\Dom{b'}\supseteq\Dom{\tilde{b}}=\Dom{b}\setminus\myB$,
$\Dom{r'}\supseteq\Dom{\tilde{r}}=\Dom{r}$, and
$\kLI'$ has real guards w.r.t.\ $\tilde{g}$. Thus,
$\myInt{\kLI'}\models\LogGuard{\tilde{g}}$.
From Lemma~\ref{lem:HomomCountQ} we know that
\begin{equation}\label{eq:homCount:formula:reseatR:MainPart}
  n
  \ \deff \
  \hom{\tilde{\kLI}}{\kLI'}
  \ \ =
  \sum_{\ov{e}\in\bV{\I'}^\ell} \hom{\kLI}{\reseatB{\kLI'}{\myB}{\tupel{e}}}.
\end{equation}
Note that $n=0$ $\iff$
$\hom{\kLI}{\reseatB{\kLI'}{\myB}{\tupel{e}}}=0$ for all
$\ov{e}\in\bV{\I'}^\ell$ $\iff$ for all $\ov{e}\in\bV{\I'}^\ell$ we
have 
either $\myInt{\reseatB{\kLI'}{\myB}{\tupel{e}}}\not\models\LogGuard{g}$ or
$\myInt{\reseatB{\kLI'}{\myB}{\tupel{e}}}\models(\LogGuard{g}\und \form{\kLI}{0})$
$\iff$
$\myInt{\kLI'}\models \form{\tilde{\kLI}}{0}$.
Thus, the formula $\form{\tilde{\kLI}}{0}$ has the desired meaning.

Let us now consider the case where $n\geq 1$.
For every number $m'\in[n]$ let
$S'_{m'}$ be the set of all $\ov{e}\in\bV{\I'}^\ell$ such that
$m'=\hom{\kLI}{\reseatB{\kLI'}{\myB}{\tupel{e}}}$.
Let $d'$ be the number of distinct $m'\in[n]$ for which $S'_{m'}\neq \emptyset$, and let
$m'_1<\cdots<m'_{d'}$ be an ordered list of all $m'\in[n]$ with $S'_{m'}\neq \emptyset$.
For each $j\in[d']$ let $c'_j\deff |S'_{m'_j}|$. Let $\tupel{c}'\deff (c'_1,\ldots,c'_{d'})$ and
$\tupel{m}'\deff (m'_1,\ldots,m'_{d'})$.
Note that
$\sum_{j\in[d']}c'_j m'_j = n$. Hence
$(\tupel{c}',\tupel{m}')\in\mySeg{n}$ is a segmentation of $n$
of size $c'\deff\sum_{j\in [d']} c'_j$ and of degree $d'$.

Note that $\myInt{\kLI'}\models \phi'_{\tupel{c}',\tupel{m}'}$. Furthermore,
for every $m\in\NNpos$ we have:
$\myInt{\kLI'}\models \form{\tilde{\kLI}}{m}$ $\iff$ $(\tupel{m}',\tupel{c}')\in\mySeg{m}$
$\iff$ $m=n=\hom{\tilde{\kLI}}{\kLI'}$.
This verifies that for every $m\geq 1$ the formula $\form{\tilde{\kLI}}{m}$ has the
desired meaning.

It remains to verify that
$(\LogGuard{\tilde{g}}\und \form{\tilde{\kLI}}{m})\in\NGCk$ for every $m\in\NN$.
Recall that $\tilde{g}=g$.
Thus, $\tilde{g}$ satisfies the condition~\eqref{eq:syntax:guard} of
rule~\ref{item:syntaxdef:10}) of the syntax definition of $\NGCk$.
Since, by the induction hypothesis,
$(\LogGuard{g}\und \form{\kLI}{m_j})\in\NGCk$, we obtain by rule
\ref{item:syntaxdef:10}) that the conjunction of $\LogGuard{\tilde{g}}$ with
the formula in the second line of $\varphi'_{\tupel{c},\tupel{m}}$
belongs to $\NGCk$.
The same reasoning (combined with rule~\ref{item:syntaxdef:4})) yields that the conjunction of $\LogGuard{\tilde{g}}$ with
the formula in the first line of $\varphi'_{\tupel{c},\tupel{m}}$
belongs to $\NGCk$.
Applying rule \ref{item:syntaxdef:5}) 
yields that each formula 
$(\LogGuard{\tilde{g}}\und \phi'_{\tupel{c},\tupel{m}})$
as well as the formula
$(\LogGuard{\tilde{g}}\und\form{\tilde{\kLI}}{m})$
belongs to $\NGCk$, for every $m\geq 1$.
A similar reasoning shows that also $(\LogGuard{\tilde{g}}\und\form{\tilde{\kLI}}{0})$
belongs to $\NGCk$. This completes the inductive step for the case
where $\tilde{\kLI}=\reclaimB{\kLI}{\myB}$.
\smallskip

The final case to consider is $\tilde{\kLI}\deff \switch{\kLI}{f}$,
where $f$ is a transition for $g$.
A closer inspection shows that for every $\kLI'$ that satisfies the
lemma's assumptions for $\tilde{\kLI}$ instead of $\kLI$, the
following is true:
$\hom{\tilde{\kLI}}{\kLI'}=\sum_{\ov{e}\in\bV{\I'}^\ell}
\hom{\kLI}{\reseatB{\kLI'}{\myB}{\tupel{e}}}$,
for
$\myB\deff \Dom{b}\cap\Img{f}\cap\Img{g}$.
I.e., we have the same situation as in \eqref{eq:homCount:formula:reseatR:MainPart}.
We therefore can choose the exact same formulas
$\form{\tilde{\kLI}}{m}$ as in the previous case
(for all $m\in\NN$) and obtain that these formulas have the desired meaning.

But it remains to verify that
$(\LogGuard{\tilde{g}}\und \form{\tilde{\kLI}}{m})\in\NGCk$ for every $m\in\NN$.
Note that now we have $\tilde{g}=f\cup g$ (whereas previously we had
to deal with the case where $\tilde{g}=g$).
All we have to do is argue that $\tilde{g}$ satisfies the
condition~\eqref{eq:syntax:guard} of 
rule~\ref{item:syntaxdef:10}) of the syntax definition of $\NGCk$.
Once having achieved this, the same reasoning as in the previous case
proves that the formula $(\LogGuard{\tilde{g}}\und
\form{\tilde{\kLI}}{m})$ belongs to 
$\NGCk$ for every $m\in\NN$.

Note that according to the induction hypothesis, for all $m'\in\NN$ we have:
$\Dom{b}=\indfreeB{(\LogGuard{g}\und\form{\kLI}{m'})}\supseteq\Img{g}$.
Therefore, $\myB=\Img{f}\cap\Img{g}$. 

Consider an arbitrary $i\in\Dom{g}$. We have to show that
$\tilde{g}(i)=g(i)$ or $\tilde{g}(i)\in \myB$ or $\tilde{g}(i)\not\in
\Img{g}$. 
If $\tilde{g}(i)=g(i)$
or $\tilde{g}(i)\not\in\Img{g}$,
we are done.
Consider the case that
$\tilde{g}(i)\neq g(i)$
and
$\tilde{g}(i)\in\Img{g}$.
In this case, $\tilde{g}(i)=f(i)$ and there exists an $i'\in\Dom{g}$
with $g(i')=\tilde{g}(i)=f(i)$.
Hence, $\tilde{g}(i)\in\Img{f}\cap\Img{g}=\myB$, and we are done.
This completes the proof sketch of Lemma~\ref{lemma:HomToFormula}.
\end{proof}
  
We proceed with the second lemma necessary for our proof of
Theorem~\ref{thm:GCkVsHomomIndist}.

\begin{lemma}\label{lemma:c_in_hom_count_new}\label{lemma:FormulaToHom}
Let $\chi \deff (\LogGuard{g} \land \psi) \in \NGCk$ and let
$m,d \in \nat$ with $m \geq 1$.
There exists a $Q\deff\quanth{\chi}{m,d}\in\QGLIk$
with $g_Q = g$,
$\dombQ=\indfreeB{\chi}$
and
$\domrQ = \Dom{g}=
\indfreeR{\chi}$
such that for all $k$-labeled incidence graphs
$\kLI'= (\I', b', r', g')$ with
$|\bV{\I'}| = m$ and
$\max \setc{|N_{\I'}(e)|}{e \in \bV{\I'}} \leq d$
and
$\Dom{b'}\supseteq \dombQ$,
$\Dom{r'}\supseteq \domrQ$,
and
$g'\supseteq g$ and with real guards w.r.t.\ $g$ 
we have $\myInt{\kLI'}\models\LogGuard{g}$ and
\begin{equation*}
\hom{Q}{\kLI'} = \begin{cases}
	1 & \text{if } \ \myInt{\kLI'} \models \chi \quad \text{and} \\
	0 & \text{if } \ \myInt{\kLI'} \not\models \chi\;.
\end{cases}
\end{equation*}
\end{lemma}

\begin{proof}[Proof sketch]
We proceed by induction on the construction of $\chi$.
The base cases can be handled straightforwardly (see
Appendix~\ref{appendix:subsec:FormulaToHom}).
To the induction base we also add the \enquote{special case} where
$\psi$ is unfulfillable for all $k$-labeled incidence graphs $\kLI'$ matching
the requirements imposed by the lemma. This case is handled by
choosing $Q\deff 0{\cdot}\kLI$, for a suitably chosen $\kLI$ that has
the desired parameters, i.e., $g_\kLI=g$, $\Dom{b_{\kLI}}=\indfreeB{\chi}$,
and $\Dom{r_{\kLI}}=\Dom{g}=\indfreeR{\chi}$.

Most cases of the induction step can be handled rather
straightforwardly. A notable exception is the case where the formula
is built using rule~\ref{item:syntaxdef:10}) of the syntax definition
of $\NGCk$.
We assume that the lemma's statement is already shown for
$\chi\deff (\LogGuard{g}\und\psi)\in\NGCk$.
Fix arbitrary $m,d\in\NN$ with $m\geq 1$.
Let $Q\in\QGLIk$ be provided by
the lemma's statement for
$\chi$ and the parameters $m,d$. 
The goal is to prove the lemma's statement for parameters $m,d$ and the particular
$\tilde{\chi} \deff
(\LogGuard{\tilde{g}} \land \tilde{\psi})\in\NGCk$ given as follows:
$\tilde{\psi} = \existsi[n]
(\vare_{i_1},\ldots,\vare_{i_\ell}).(\LogGuard{g} \land \psi)$,
where $n,\ell\in\NNpos$ and
$S\deff\set{i_1,\ldots,i_\ell}\subseteq \indfreeB{\chi}$
with $i_1<\cdots <i_\ell$ and $\Dom{\tilde{g}}=\Dom{g}$ and
all $i\in\Dom{g}$ satisfy
$\tilde{g}(i)=g(i)$ or
$\tilde{g}(i)\in S$ or
$\tilde{g}(i)\not\in \Img{g}$.

The case where $n>m^\ell$ is handled by the \enquote{special case}
mentioned above. The case $n\leq m^\ell$ is quite intricate and relies
on a series of claims (see Appendix~\ref{appendix:subsec:FormulaToHom} for proofs):
\smallskip

\textit{Claim 1:} \
Consider an arbitrary $k$-labeled incidence graph $\kLI'$ satisfying
the lemma's assumptions for $\tilde{\chi}$.
\\
For
\,$
  \mysum{\kLI'} \deff 
  \hom{\reclaimB{Q}{S}}{\kLI'}
$\,
we have:\\
\,$
  \mysum{\kLI'}  = 
  \sum_{\tupel{e} \in \bV{\I_{\kLI'}}^{\ell}}
  \hom{Q}{\reseatB{\kLI'}{S}{\tupel{e}}}
$\,
and \\
$0\leq\mysum{\kLI'}\leq |\bV{\I_{\kLI'}}|^\ell$
and $\mysum{\kLI'}$ is exactly the number of tuples
$\tupel{e}\in \bV{\I_{\kLI'}}^\ell$ such that 
$\myInt{\reseatB{\kLI'}{\myB}{\tupel{e}}}\models \chi$.
\smallskip

For the remainder of the proof, our aim is to
construct a $Q'\in\QGLIk$
that has the parameters
$g_{Q'},\domb{Q'}, \domr{Q'}$
required by the
lemma for formula $\tilde{\chi}$,
such
that $\hom{Q'}{\kLI'}=\mysum{\kLI'}$ for all  $k$-labeled incidence
graphs $\kLI'$
satisfying the lemma's assumptions for $\tilde{\chi}$.

\smallskip

Once having achieved this, we can use
Lemma~\ref{lemma:normalizing_quantum_incidence_graphs} to obtain 
$\tilde{Q}\deff Q'[X,Y] \in\QGLIk$ for
$X\deff\set{0,\ldots,n{-}1}$ and $Y\deff\set{n,\ldots,m^\ell}$.
Note that $\tilde{Q}$ has the properties desired for the formula
$\tilde{\chi}$.
But how to construct a suitable $Q'$?
--- We let
\begin{equation*}
  Z\ \deff\ \setc{\,i\in\Dom{g}\ }{\ \tilde{g}(i)\neq g(i) \text{ \ or \ } g(i)\in S\,}\,.
\end{equation*}
If $Z=\emptyset$, we choose
$Q'\deff\reclaimB{Q}{S}$ and are done: $Z=\emptyset$ means that
$\tilde{g}=g$ and $\Img{g}\cap
S=\emptyset$, and hence $\reclaimB{Q}{S}\in\QGLIk$.
If $Z\neq\emptyset$, we cannot simply choose $Q'$ as above, because
there is no guarantee that $\reclaimB{Q}{S}$ belongs to $\QGLIk$.
Instead, we proceed as follows.
Let $f$ be the partial function with $\Dom{f}=Z$ and
$f(i) \deff \tilde{g}(i)$ for all $i\in \Dom{f}$.
\smallskip

\textit{Claim 2:} \ 
$\tilde{g}= f\cup g$, and
$f$ is a transition for $g$, and
all $i\in\Dom{f}$ satisfy: \ $f(i)\in S$ \ or \
$f(i)\not\in \Img{g}$.
\smallskip

Let $D_1\deff (S\cap\Img{f})\setminus \Img{g}$,
let $D_2\deff S\setminus\Img{\tilde{g}}$, and
let $D_3\deff \Img{g}\cap\Img{f}$.
\smallskip

\emph{Claim~3:} $D_1,D_2,D_3$ are pairwise disjoint, $S=\bigcup_{i\in[3]}D_i$.\smallskip

We let $Q_1\deff \reclaimB{Q}{D_1}$ and
$Q_2\deff \switch{Q_1}{f}$ and
$Q'\deff \reclaimB{Q_2}{D_2}$.
\smallskip

\emph{Claim~4:} \ $Q_1,Q_2,Q'\in\QGLIk$ and 
$g_{Q_1}=g$, $\domb{Q_1}=(\dombQ\setminus D_1)$ and 
$g_{Q'}=\tilde{g}$, 
$\domr{Q'}=\indfreeR{\tilde{\chi}}$, $\domb{Q'}=\indfreeB{\tilde{\chi}}$.\smallskip

According to Claim~4, $Q'\in\QGLIk$ and $Q'$ has the desired
parameters $g_{Q'},\domb{Q'},\domr{Q'}$. 
The next claim ensures that $Q'$ indeed satisfies what we aimed for.
\smallskip

\textit{Claim 5:} $\hom{Q'}{\kLI'} = \mysum{\kLI'}$,
for all $k$-labeled incidence graphs $\kLI'$ that satisfy
the lemma's assumptions for $\tilde{\chi}$.\smallskip

This completes the proof sketch for Lemma~\ref{lemma:FormulaToHom}.
\end{proof}

Finally, the proof of Theorem~\ref{thm:GCkVsHomomIndist} can easily be achieved 
by using the Theorems~\ref{thm:normalform} and \ref{thm:IEHWkIsGLIk}
and the Lemmas~\ref{lemma:HomToFormula} (for direction ``$\Longleftarrow$'') and
\ref{lemma:FormulaToHom} (for direction ``$\Longrightarrow$''):

\begin{proof}[\textbf{\upshape{Proof of Theorem~\ref{thm:GCkVsHomomIndist}}}] \
Let $\I$ and $\I'$ be arbitrary incidence graphs, and let
$k\in\NNpos$. Our aim is to prove that
$\I\not\equivGCk\I'$
$\iff$
$\HOM{\IEHW{k}}{\I} \neq \HOM{\IEHW{k}}{\I'}$.

The case where
$|\bVI|\neq |\bV{\I'}|$ can be handled easily:
Choose a sentence in $\GCk$ stating that the number of blue nodes is
exactly $|\bVI|$, and
choose an incidence
graph $J$ that has one blue node and neither red nodes nor edges.
This yields that $\I\not\equivGCk\I'$
and $\HOM{\IEHW{k}}{\I} \neq \HOM{\IEHW{k}}{\I'}$.

Consider the case $|\bVI| = |\bV{\I'}|$.
Let $\kLI,\kLI'$ be the label-free $k$-labeled incidence graphs
corresponding to $\I,\I'$.

For ``$\Longleftarrow$'' let $\J\in\IEHWk$
such that $\hom{\J}{\I}\neq\hom{\J}{\I'}$.
Let $m\deff\hom{\J}{\I}$.
From Theorem\ref{thm:IEHWkIsGLIk}\ref{item:IEHWkInGLIk}
we obtain a label-free $\tilde{\kLI}\in\GLIk$ such that
$\I_{\tilde{\kLI}}\isom\J$.
From Lemma~\ref{lemma:HomToFormula} we obtain a sentence
$\chi\deff (\top\und\form{\tilde{\kLI}}{m})$ in $\NGCk$ (and therefore also in $\GCk$) such that
$\myInt{\kLI}\models \chi$ (because $\hom{\tilde{\kLI}}{\kLI}=m$) and
$\myInt{\kLI'}\not\models \chi$ (because
$\hom{\tilde{\kLI}}{\kLI'}\neq m$).
Thus, $\I\models\chi$ and $\I'\not\models\chi$. I.e., $\chi$ witnesses
that $\I\not\equivGCk\I'$.

For ``$\Longrightarrow$'' let $\phi$
be a sentence
in $\GCk$ such that $\I\models\phi$ and $\I'\not\models\phi$.
By Theorem~\ref{thm:ngck_is_equivalent} there exists a $\chi \isdef (\top \land \psi) \in \NGCk$ such that $\phi \equiv (\top \land \phi) \equiv \chi$, i.e. $\I \models \chi$ and $\I' \not\models \chi$.
Let $m\deff |\bVI|=|\bV{\I'}|$.
Choose a $d\in\NN$ such
that $d\geq |N_{\I}(e)|$ for all $e\in\bVI$ and $d\geq |N_{\I'}(e)|$ for
all $e\in\bV{\I'}$.
Let $Q\deff \quanth{\chi}{m,d}\in\QGLIk$ be provided by
Lemma~\ref{lemma:FormulaToHom}.
From Lemma~\ref{lemma:FormulaToHom} we know that
$\hom{Q}{\kLI}=1$ (because $\I\models\chi$) and
$\hom{Q}{\kLI'}=0$ (because $\I'\not\models\chi$).
Since $Q\in\QGLIk$, there are an $\ell\in\NNpos$ and
$\alpha_1,\ldots,\alpha_\ell\in\RR$ and
$\kLI_1,\ldots,\kLI_\ell\in\GLIk$ such that $Q=\sum_{i=1}^\ell
\alpha_i \kLI_i$. We have:
$\sum_{i=1}^\ell \alpha_i \hom{\kLI_i}{\kLI}
 =
 \hom{Q}{\kLI} \neq \hom{Q}{\kLI'}
 = \sum_{i=1}^\ell \alpha_i \hom{\kLI_i}{\kLI'}$.
Hence there exists an $i\in[\ell]$ with $\alpha_i\neq 0$ and
$\hom{\kLI_i}{\kLI}\neq\hom{\kLI_i}{\kLI'}$.
We know that $\kLI_i\in\GLIk$,
and $\kLI_i=(\I_i,r_i,b_i,g_i)$ has
parameters $\domrQ=\dombQ=\emptyset$ and guard function
$g_i=g_\emptyset$.
I.e., $\Dom{r_i}=\Dom{b_i}=\emptyset$.
From Theorem~\ref{thm:IEHWkIsGLIk}\ref{item:GLIkInIEHWk} we know that
$\I_i\in\IEHWk$. Thus, $\I_i$ witnesses that 
$ \HOM{\IEHW{k}}{\I} \neq \HOM{\IEHW{k}}{\I'}$.
This completes the proof of Theorem~\ref{thm:GCkVsHomomIndist}.
\end{proof}

\section{Conclusion}
\label{sec:conclusion}

Combining the Theorems~\ref{thm:Boeker}, \ref{thm:EHWandIGHW},
\ref{thm:GCkVsHomomIndist} yields:

\begin{theorem}[Main Theorem]\label{thm:Hauptresultat}
Let $H,H'$ be hypergraphs.
\begin{mea}  
\item
$\inzidenz{H}\equiv_{\GCk}\inzidenz{H'}$ \ $\iff$ \
$\HOM{\GHW{k}}{H}=\HOM{\GHW{k}}{H'}$.
\item
If $H$ and $H'$ are simple hypergraphs, then\\
$\inzidenz{H}\equiv_{\GCk}\inzidenz{H'}$ \ $\iff$ \
$\HOM{\SGHW{k}}{H}=\HOM{\SGHW{k}}{H'}$.
\end{mea}  
\end{theorem}

For our proofs it was crucial to consider
\emph{entangled hypertree decompositions} (ehds, for short)
instead of
generalised hypertree decompositions.
To the best of our knowledge,
ehds have not been studied before.
From Theorem~\ref{thm:IEHWsubsetIGHW} we know that there exist
arbitrarily large $k$ such that 
$\IEHW{k}$ is a strict
subclass of $\IGHW{k}$; but nevertheless, according to Theorem~\ref{thm:EHWandIGHW}
homomorphism indistinguishability coincides for both classes.
Many questions remain open, in particular:

How hard is it, given a hypergraph $H$ and a number $k$, to
determine whether $\ehw{H} \leq k$?

For $\ClassC\deff\IEHWk$: \,how hard is it to compute the function
(or, \enquote{vector}) $\HOM{\ClassC}{H}$ for a given hypergraph $H$?
Which properties does it have?

What is the expressive power of the logic $\GCk$?
What properties of hypergraphs can be described by this logic?
How does a suitable pebble game for $\GCk$ look like?
  
Our result lifts Dvořák's result for tree-width $\leq k$
\cite{Dvorak2010} from graphs to hypergraphs.  
Does there also exist a lifting of Grohe's result for
\emph{tree-depth} $\leq k$ \cite{Grohe2020} from graphs to hypergraphs?

Seeing that Dvořák's result lifted nicely to hypergraphs, we
believe that there should also be a lifting of
Cai, Fürer and Immerman's result
\cite{Cai1992}, i.e., 
a hypergraph-variant of the Weisfeiler-Leman
algorithm, whose distinguishing power matches
precisely the logic $\GCk$. We plan to study this in future work.

\addcontentsline{toc}{section}{References}\bibliography{IEEEabrv,references}

\begin{thebibliography}{10}
\providecommand{\url}[1]{#1}
\csname url@samestyle\endcsname
\providecommand{\newblock}{\relax}
\providecommand{\bibinfo}[2]{#2}
\providecommand{\BIBentrySTDinterwordspacing}{\spaceskip=0pt\relax}
\providecommand{\BIBentryALTinterwordstretchfactor}{4}
\providecommand{\BIBentryALTinterwordspacing}{\spaceskip=\fontdimen2\font plus
\BIBentryALTinterwordstretchfactor\fontdimen3\font minus
  \fontdimen4\font\relax}
\providecommand{\BIBforeignlanguage}[2]{{%
\expandafter\ifx\csname l@#1\endcsname\relax
\typeout{** WARNING: IEEEtranS.bst: No hyphenation pattern has been}%
\typeout{** loaded for the language `#1'. Using the pattern for}%
\typeout{** the default language instead.}%
\else
\language=\csname l@#1\endcsname
\fi
#2}}
\providecommand{\BIBdecl}{\relax}
\BIBdecl

\bibitem{Adler2004}
I.~Adler, ``Marshals, monotone marshals, and hypertree-width,'' \emph{J.\ Graph
  Theory}, vol.~47, no.~4, pp. 275--296, 2004.

\bibitem{AdlerDiss}
------, ``Width functions for hypertree decompositions,'' Ph.D. dissertation,
  Albert-Ludwigs Universität Freiburg, 2006, available at
  \url{https://d-nb.info/979896851/34}.

\bibitem{Adler2007}
I.~Adler, G.~Gottlob, and M.~Grohe, ``\BIBforeignlanguage{en}{Hypertree width
  and related hypergraph invariants},'' \emph{\BIBforeignlanguage{en}{European
  Journal of Combinatorics}}, vol.~28, no.~8, pp. 2167--2181, 2007.

\bibitem{DBLP:journals/jphil/AndrekaNB98}
H.~Andr{\'{e}}ka, I.~N{\'{e}}meti, and J.~van Benthem, ``Modal languages and
  bounded fragments of predicate logic,'' \emph{J. Philos. Log.}, vol.~27,
  no.~3, pp. 217--274, 1998.

\bibitem{DBLP:conf/mfcs/BokerCGR19}
J.~B{\"{o}}ker, Y.~Chen, M.~Grohe, and G.~Rattan, ``The complexity of
  homomorphism indistinguishability,'' in \emph{44th International Symposium on
  Mathematical Foundations of Computer Science, {MFCS} 2019, August 26-30,
  2019, Aachen, Germany}, ser. LIPIcs, vol. 138.\hskip 1em plus 0.5em minus
  0.4em\relax Schloss Dagstuhl - Leibniz-Zentrum f{\"{u}}r Informatik, 2019,
  pp. 54:1--54:13.

\bibitem{Boeker-MasterThesis}
J.~Böker, ``Structural similarity and homomorphism counts,'' Master's thesis,
  RWTH Aachen University, 2018.

\bibitem{Boeker2019}
------, ``{C}olor {R}efinement, {H}omomorphisms, and {H}ypergraphs,'' in
  \emph{Graph-Theoretic Concepts in Computer Science}, ser. Lecture Notes in
  Computer Science, I.~Sau and D.~M. Thilikos, Eds.\hskip 1em plus 0.5em minus
  0.4em\relax Springer, Cham, 2019, vol. 11789, pp. 338--350.

\bibitem{Cai1992}
J.-Y. Cai, M.~Fürer, and N.~Immerman, ``An optimal lower bound on the number
  of variables for graph identification,'' \emph{Combinatorica}, vol.~12,
  no.~4, pp. 389--410, Dec. 1992.

\bibitem{Courcelle1993}
B.~Courcelle, ``{G}raph {G}rammars, {M}onadic {S}econd-{O}rder {L}ogic {A}nd
  {T}he {T}heory {O}f {G}raph {M}inors,'' in \emph{Graph Structure Theory},
  ser. Contemporary Mathematics, N.~Robertson and P.~Seymour, Eds., vol.
  147.\hskip 1em plus 0.5em minus 0.4em\relax AMS, 1993, pp. 565--590.

\bibitem{DeAlba2007}
L.~M. DeAlba, ``{D}eterminants and {E}igenvalues,'' in \emph{{H}andbook of
  {L}inear {A}lgebra}, ser. Discrete Mathematics and Its Applications,
  L.~Hogben, Ed.\hskip 1em plus 0.5em minus 0.4em\relax Chapman \& Hall/CRC,
  2007, pp. 4--1 -- 4--12.

\bibitem{DBLP:conf/icalp/DellGR18}
H.~Dell, M.~Grohe, and G.~Rattan, ``Lov{\'{a}}sz meets {W}eisfeiler and
  {L}eman,'' in \emph{45th International Colloquium on Automata, Languages, and
  Programming, {ICALP} 2018, July 9-13, 2018, Prague, Czech Republic}, ser.
  LIPIcs, vol. 107.\hskip 1em plus 0.5em minus 0.4em\relax Schloss Dagstuhl -
  Leibniz-Zentrum f{\"{u}}r Informatik, 2018, pp. 40:1--40:14.

\bibitem{Dvorak2010}
Z.~Dvořák, ``{O}n {R}ecognizing {G}raphs by {N}umbers of {H}omomorphisms,''
  \emph{Journal of Graph Theory}, vol.~64, no.~4, pp. 330--342, 2010.

\bibitem{Gottlob2016}
G.~Gottlob, G.~Greco, N.~Leone, and F.~Scarcello, ``Hypertree decompositions:
  Questions and answers,'' in \emph{Proceedings of the 35th {ACM}
  {SIGMOD-SIGACT-SIGAI} Symposium on Principles of Database Systems, {PODS}
  2016, San Francisco, CA, USA, June 26 - July 01, 2016}, T.~Milo and W.~Tan,
  Eds.\hskip 1em plus 0.5em minus 0.4em\relax ACM, 2016, pp. 57--74.

\bibitem{Gottlob2002}
G.~Gottlob, N.~Leone, and F.~Scarcello, ``{H}ypertree {D}ecompositions and
  {T}ractable {Q}ueries,'' \emph{Journal of Computer and System Sciences},
  vol.~64, no.~3, pp. 579--627, May 2002.

\bibitem{Gottlob2003}
------, ``Robbers, marshals, and guards: game theoretic and logical
  characterizations of hypertree width,'' \emph{Journal of Computer and System
  Sciences}, vol.~66, no.~4, pp. 775--808, Jun. 2003.

\bibitem{DBLP:journals/jsyml/Gradel99}
E.~Gr{\"{a}}del, ``On the restraining power of guards,'' \emph{J. Symb. Log.},
  vol.~64, no.~4, pp. 1719--1742, 1999.

\bibitem{Grohe2020}
M.~Grohe, ``{C}ounting {B}ounded {T}ree {D}epth {H}omomorphisms,'' in
  \emph{Proceedings of the 35th Annual ACM/IEEE Symposium on Logic in Computer
  Science}, ser. LICS'20.\hskip 1em plus 0.5em minus 0.4em\relax ACM, Jul.
  2020, pp. 507--520.

\bibitem{Grohe2020a}
------, ``{W}ord2vec, {N}ode2vec, {G}raph2vec, {X}2vec: {T}owards a {T}heory of
  {V}ector {E}mbeddings of {S}tructured {D}ata,'' in \emph{Proceedings of the
  39th ACM SIGMOD-SIGACT-SIGAI Symposium on Principles of Database Systems},
  ser. PODS'20.\hskip 1em plus 0.5em minus 0.4em\relax ACM, 2020, pp. 1--16.

\bibitem{Kolaitis2000}
P.~G. Kolaitis and M.~Y. Vardi, ``{C}onjunctive-{Q}uery {C}ontainment and
  {C}onstraint {S}atisfaction,'' \emph{Journal of Computer and System
  Sciences}, vol.~61, no.~2, pp. 302--332, Oct. 2000.

\bibitem{Lovasz1967}
L.~Lovász, ``{O}perations with structures,'' \emph{Acta Mathematica Academiae
  Scientiarum Hungaricae}, vol.~18, no.~3, pp. 321--328, 1967.

\bibitem{Lovasz2008}
L.~Lovász and B.~Szegedy, ``{C}ontractors and {C}onnectors of {G}raph
  {A}lgebras,'' \emph{Journal of Graph Theory}, vol.~60, no.~1, pp. 11--30,
  2008.

\bibitem{DBLP:conf/focs/MancinskaR20}
L.~Mancinska and D.~E. Roberson, ``Quantum isomorphism is equivalent to
  equality of homomorphism counts from planar graphs,'' in \emph{61st {IEEE}
  Annual Symposium on Foundations of Computer Science, {FOCS} 2020, Durham, NC,
  USA, November 16-19, 2020}.\hskip 1em plus 0.5em minus 0.4em\relax {IEEE},
  2020, pp. 661--672.

\bibitem{ScheidtSchweikardtMFCS23}
B.~Scheidt and N.~Schweikardt, ``Counting homomorphisms from hypergraphs of
  bounded generalised hypertree width: A logical characterisation,'' in
  \emph{Proceedings of the 48th International Symposium on Mathematical
  Foundations of Computer Science, {MFCS} 2023}, 2023.

\end{thebibliography}

\clearpage

\appendix

\addcontentsline{toc}{section}{APPENDIX}\section*{APPENDIX}

\bigskip

\section{Proof of Theorem~\ref{thm:Boeker}}
\label{appendix:generalization_to_regular_homomorphisms}

{
\newcommand*{\LoMeHom}{\textsf{LoMeHom}}
\newcommand*{\PumpHom}{\textsf{PumpHom}}

Notice that this paper's Theorem~\ref{thm:Boeker} looks similar to
Bökers Lemma 6 in \cite{Boeker2019}.\footnote{Consult the full version
  for a proof of Lemma 6; you can find it at
  \href{https://arxiv.org/abs/1903.12432}{arXiv:1903.12432 [cs.DM]}. 
  Note that the numbering of the lemmas is increased by one, i.e.\ Lemma 6 is Lemma 7 in the full version etc.}
Also observe that this is the only Lemma that he proves just for
Berge-acyclic hypergraphs. Lemmas~3, 4, and 5 are proven for arbitrary
hypergraphs (as he notes himself). So all we have to do is to verify,
that his proof for Lemma~6 works if we consider the class of
hypergraphs of generalised hypertree width at most $k$ instead of the
class of Berge-acyclic hypergraphs. 

Upon examining his proof of Lemma~6 (see the full version on arXiv) we
see that closure under pumping and local merging are the only
properties of Berge-acylic hypergraphs, that Böker uses. Hence, since
these closures also hold for hypergraphs of generalised hypertree
width at most $k$, the proof remains valid. 
} 
\begin{proof}[\textbf{\upshape{Proof of Theorem~\ref{thm:Boeker}:}}] \ \smallskip

\ref{item:a:thm:Boeker}: ``$\Longrightarrow$'' is trivial since $\GHW{k}\supseteq\SGHW{k}$.
``$\Longleftarrow$''
holds because for every hypergraph $F$ and every homomorphism $h$ from
$F$ to a simple hypergraph, all edges $e,e'\in \E{F}$ with
$f_F(e)=f_F(e')$ have to be mapped onto the same edge of the simple
hypergraph. Thus, if $F\in \GHW{k}$ distinguishes between $H$ and
$H'$, then also the ``simplified version'' of $F$ (i.e., the simple
hypergraph obtained by letting each edge of $F$ occur with
multiplicity 1) also distinguishes between $H$ and $H'$. 
\smallskip

\ref{item:b:thm:Boeker}:
Böker \cite{Boeker2019} proved the analogous statement for
$\BA{k}$, $\IBA{k}$ instead of $\GHW{k}$, $\IGHW{k}$, where $\BA{k}$ is the
class of all Berge-acyclic hypergraphs and $\IBA{k}$ is the class of
all incidence graphs of hypergraphs in $\BA{k}$. 
Böker's proof, however, works for all classes $\Class{C}$ of
hypergraphs and the associated class $\Class{IC}$ of all incidence
graphs of hypergraphs in $\Class{C}$, provided that $\Class{C}$ is
closed under \emph{local merging} and \emph{pumping}\footnote{Böker
  calls this \enquote{leaf adding}, since he is working with incidence
  graphs that are also always trees. We call it pumping, because this
  is the more intuitive name in our environment. Note that our
  definition of pumping is the same as Böker's definition of
  leaf-adding.}: 
We say that $\Class{C}$ is closed under \emph{local merging}, if every
$F\in\Class{C}$ remains in $\Class{C}$ when merging two vertices
$u_1,u_2 \in \V{F}$ that are adjacent via a common edge $e \in \E{F}$
(i.e., $u_1,u_2 \in f_F(e)$), into a single new vertex $u$.  
We say that $\Class{C}$ is closed under \emph{pumping}, if every
$F\in\Class{C}$ remains in $\Class{C}$ when inserting a newly created
vertex into an edge $e \in \E{F}$ (but note that we must not insert
vertices in the intersection of multiple edges). 
Now, all that remains to be shown is that the class $\GHW{k}$ is closed under local merging and pumping.

Closure under pumping is verified as follows: Assume we are given a
ghd $D=(T,\DBag,\DCover)$ of a hypergraph $F$ (resp., of its incidence
graph $I_F$). Let $F'$ be the hypergraph obtained from $F$ by pumping
edge $e\in\E{F}$ with a new vertex $u$, i.e.,
$\V{F'}=\V{F}\union\set{u}$, $\E{F'}=\E{F}$, $f_{F'}(e)=
f_{F}(e)\union\set{u}$, and $f_{F'}(e')=f_F(e')$ for all $e'\in
\E{F}\setminus\set{e}$. 
We can turn $D$ into a ghd $D'$ of $F'$ as follows:
We choose one particular tree-node $t\in \V{T}$ such that
$e\in\DCover(t)$ and $f_F(e)\subseteq \DBag(t)$ (such a $t$ exists
because $D$ satisfies the \emph{completeness} condition of
Definition~\ref{def:hypertree_decomposition}), and we insert the new
vertex $u$ into the bag $\DBag(t)$. I.e., $D'=(T',\DBag',\DCover')$
with $T'=T$, $\DCover'=\DCover$, and $\DBag'(t)=\DBag(t)\cup\set{u}$
and $\DBag'(t')=\DBag(t')$ for all $t'\in\V{T}\setminus\set{t}$. 
It is straightforward to verify that $D'$ is a ghd of $F'$ and $\w{D'}=\w{D}$.
Thus, $F\in\GHW{k}$ implies that $F'\in\GHW{k}$.

Closure under local merging is verified as follows:
Assume we are given a ghd $D=(T,\DBag,\DCover)$ of a hypergraph $F$ (resp., of its
incidence graph $\inzidenz{F}$). Let $F'$ be the hypergraph obtained from $F$ by merging
two vertices $u_1,u_2$ for which there exists an $e\in\E{F}$ with
$u_1,u_2\in f_F(e)$. By Definition~\ref{def:hypertree_decomposition},
$D$ satisfies 
the completeness and connectedness conditions. I.e.,
there is a tree-node $t \in V(T)$ with $e\in\DCover(t)$ and $f_F(e)\subseteq\DBag(t)$.
In particular, $u_1,u_2 \in \DBag(t)$.
Hence, the subtrees $T_{u_1}$ and $T_{u_2}$ touch each other in $t$.
The hypergraph $F'$ is obtained from $F$ by 
identifying the vertices $u_1$ and $u_2$ with a single new vertex
$u$. More precisely, let $\pi$ be the mapping with
$\pi(u_1)=\pi(u_2)=u$ and $\pi(v)=v$ for all
$v\in\V{F}\setminus\set{u_1,u_2}$. Then, 
$F'=(\V{F'},\E{F'},f_{F'})$ with $\V{F'}=\setc{\pi(v)}{v\in \V{F}}$, $\E{F'}=\E{F}$ and
$f_{F'}(e')=\setc{\pi(v)}{v\in e'}$ for all $e'\in \E{F'}$.
For every tree-node $t\in \V{T}$ let
$\DBag'(t)\deff\setc{\pi(v)}{v\in\DBag(t)}$. Let $D'\deff
(T,\DBag',\DCover)$. It is straightforward to verify that $D'$ is a
ghd of $F'$. Obviously, $\w{D'}=\w{D}$. Thus, $F\in\GHW{k}$ implies
that $F'\in\GHW{k}$. 
\end{proof}

\medskip

\section{Details omitted in Section~\ref{sec:generalization_to_hw}}
\label{appendix:generalization_to_hw}

\subsection{Detailed Proof of Theorem~\ref{thm:EHWandIGHW}}
\label{appendix:Proof_of_thm:EHWandIGHW}

\begin{proof}[\upshape\textbf{Proof of Theorem~\ref{thm:EHWandIGHW}}] \ \\
  The direction ``$\Longrightarrow$'' is trivial since $\IGHW{k}\supseteq\IEHW{k}$.
  For the direction ``$\Longleftarrow$'' it suffices to prove the following:
  If there is a $\J\in\IGHW{k}$ with  $\hom{\J}{\I} \neq
  \hom{\J}{\I'}$, then there also exists a $\J'\in\IEHW{k}$ with
  $\hom{\J'}{\I} \neq \hom{\J'}{\I'}$. 
  The remainder of this proof is devoted to showing that such a $\J'$
  indeed exists (for illustrations depicting the following
  constructions see
  Appendix~\ref{appendix:illustrations_for_thm:EHWandIGHW}). 

  We start with a $\J\in\IGHW{k}$ with  $\hom{\J}{\I} \neq
  \hom{\J}{\I'}$. Let
  $D=(T,\DBag,\DCover)$ be a ghd of $J$ with $\w{D}\leq k$.
  We use a 2-step process to construct the desired $J'$:
  First, we transform $D$ into a ghd $\Deins$ of an incidence graph $\Jeins$ such that $\w{\Deins}\leq \w{D}$ and 
  $\hom{\Jeins}{\I} \neq \hom{\Jeins}{\I'}$ and $\Deins$ satisfies
  condition~\ref{def:entangled_hypertree_decomposition:precision} of
  Definition~\ref{def:entangled_hypertree_decomposition} (but
  condition~\ref{def:entangled_hypertree_decomposition:connectedness_of_edges}
  might still be violated).
  Afterwards, we transform $\Deins$ into a ghd $\Dzwei$ of an
  incidence graph $\Jzwei$ such that $\w{\Dzwei}\leq \w{\Deins}$ and 
  $\hom{\Jzwei}{\I} \neq \hom{\Jzwei}{\I'}$ and $\Dzwei$ satisfies
  conditions~\ref{def:entangled_hypertree_decomposition:precision}
  and \ref{def:entangled_hypertree_decomposition:connectedness_of_edges}
  of
  Definition~\ref{def:entangled_hypertree_decomposition} and hence 
  is an ehd. Letting $\J'\deff\Jzwei$ then completes the proof.

\emph{Construction of $\Deins$ and $\Jeins$:} \  
  For each tree-node $t\in \V{T}$ and each blue node $e\in \bVJ$
  consider the set of red nodes $s(t,e)\deff N_{\J}(e)\cap\DBag(t)$.
  Let 
  $S\deff \setc{s(t,e)}{t\in\V{T},\ e\in \DCover(t),\ \Nachbarn{e}{J}
    \not\subseteq \DBag(t)}$. Note that $S=\emptyset$ implies that
  condition~\ref{def:entangled_hypertree_decomposition:precision} is
  satisfied and we can choose $\Deins=D$ and $\Jeins=J$. If $S\neq
  \emptyset$, we let $\tilde{\J}\deff \J$, and we loop through all $s\in S$ and proceed as follows:
  We use Lemma~\ref{lemma:adding_edges_keeps_distinguishing} to choose
  a suitable number $n_s\geq 1$ and insert into $\tilde{\J}$ exactly $n_s$ new blue nodes
  $e'_{s,i}$ (for $i\in [n_s]$) along with edges from $e'_{s,i}$ to every
  red node $v\in s$
  --- i.e., we replace $\tilde{\J}$ with $\tilde{\J}+n_s{\cdot}s$.
  Let $\Jeins$ be the resulting $\tilde{\J}$ after having looped
  through all $s\in S$.
  Since we have chosen the $n_s$ (for $s\in S$) according to
  Lemma~\ref{lemma:adding_edges_keeps_distinguishing}, we know that 
  $\hom{\Jeins}{\I} \neq \hom{\Jeins}{\I'}$.
  The ghd $D$ of $\J$ is transformed into a ghd $\Deins=(\Teins,\bageins,\covereins)$ of $\Jeins$
  satisfying
  condition~\ref{def:entangled_hypertree_decomposition:precision} as follows:
  For every $t\in \V{T}$ let $\bageins(t)\deff \DBag(t)$ and
  $\covereins(t)\deff \setc{e\in\DCover(t)}{\Nachbarn{e}{J} \subseteq
    \DBag(t)}
  \cup
  \setc{e'_{s(t,e),1}}{e\in\DCover(t),\ \Nachbarn{e}{J} \not\subseteq
    \DBag(t)}$.
  This ensures that
  condition~\ref{def:entangled_hypertree_decomposition:precision} is
  satisfied when choosing $\Teins\deff T$. Clearly,
  conditions~\ref{def:hypertree_decomposition:connectedness_of_vertices}
  and \ref{def:hypertree_decomposition:covering} are satisfied; and
  condition~\ref{def:hypertree_decomposition:completeness} is 
  satisfied for all those blue nodes that have already been present in
  $\J$ and for $e'_{s,1}$ for all $s\in S$. To meet
  condition~\ref{def:hypertree_decomposition:completeness} also for
  the newly inserted blue nodes $e'_{s,i}$ with $s\in S$ and
  $i\in\set{2,\ldots,n_s}$, we loop trough all $s\in S$ with $n_s\geq
  2$ and proceed as follows: Pick an arbitrary tree-node $t$ in
  $\Teins$ such that $e'_{s,1}\in\covereins(t)$. For each
  $i\in\set{2,\ldots,n_s}$ insert into $\Teins$ a new leaf $t_{s,i}$
  adjacent to $t$ and let $\bageins(t_{s,i})\deff s$ and
  $\covereins(t_{s,i})\deff \set{e'_{s,i}}$.
  This completes the construction of $\Deins$ and $\Jeins$.
  Note that $\Deins$ might violate condition~\ref{def:entangled_hypertree_decomposition:connectedness_of_edges}.

  \emph{Construction of $\Dzwei$ and $\Jzwei$:}
  For each $e\in \bV{\Jeins}$ let $m_e$ denote the
  number of connected components of the subgraph
  $\Teins_{e}$, i.e., the subgraph of $\Teins$ induced on
  $V_e\deff\setc{t\in \V{\Teins}}{e\in\covereins(t)}$;
  and let $V_{e,0},\ldots,V_{e,m_e-1}$
  be the sets of tree-nodes (i.e., nodes in $\V{\Teins}$) of these
  connected components.
  In order to meet
  condition~\ref{def:entangled_hypertree_decomposition:connectedness_of_edges},
  we let $\tilde{J}\deff\Jeins$, and
  we loop through all those $e\in \bV{\Jeins}$ where $m_e\geq 2$ and
  proceed as follows:
  We let $s\deff N_{\Jeins}(e)$ and use Lemma~\ref{lemma:adding_edges_keeps_distinguishing} to choose
  a suitable number $n_e\geq m_e{-}1$ and insert into $\tilde{\J}$ 
  exactly $n_e$ new blue nodes
  $e'_{e,i}$ (for $i\in [n_e]$) and connect each of them to all red
  nodes in $s$
  --- i.e., we replace $\tilde{\J}$ with $\tilde{\J}+n_e{\cdot}s$.
  Let $\Jzwei$ be the resulting $\tilde{\J}$ after having looped
  through all the relevant $e$.
  Since the numbers $n_e$ were chosen according to
  Lemma~\ref{lemma:adding_edges_keeps_distinguishing}, we have
  $\hom{\Jzwei}{\I} \neq \hom{\Jzwei}{\I'}$.
  The ghd $\Deins$ of $\Jeins$ is transformed into an ehd
  $\Dzwei =(\Tzwei,\bagzwei,\coverzwei)$
  of $\Jzwei$
  as follows:
  We start by letting $\Dzwei\deff \Deins$. Then, we loop through all
  $e\in \bV{\Jeins}$ where $m_e\geq 2$ and
  proceed as follows:
  For each $i\in\set{1,\ldots,m_e{-}1}$ we loop through all tree-nodes
  $t\in V_{e,i}$ and replace $e$ with $e'_{e,i}$ in
  $\coverzwei(t)$. This will ensure that
  condition~\ref{def:entangled_hypertree_decomposition:connectedness_of_edges}
  is met. To ensure that also
  condition~\ref{def:hypertree_decomposition:completeness} is
  satisfied, we
  pick an arbitrary tree-node $t\in V_{e,0}$. For each
  $i\in [n_e]$ with $i\geq m_e$, we insert into $\Tzwei$ a new leaf $t_{e,i}$
  adjacent to $t$ and let $\bagzwei(t_{e,i}) \deff N_{\Jzwei}(e'_{e,i})$ and
  $\coverzwei(t_{e,i})\deff \set{e'_{e,i}}$.
  
  This completes the construction of $\Dzwei$ and $\Jzwei$.
  It is straightforward to verify that $\Dzwei$ is an ehd of $\Jzwei$
  and $\w{\Dzwei}=\w{\Deins}\leq\w{D}$. This completes the proof of Theorem~\ref{thm:EHWandIGHW}.
\end{proof}

\medskip

\subsection{Illustrations concerning the proof of
  Theorem~\ref{thm:EHWandIGHW}}
\label{appendix:illustrations_for_thm:EHWandIGHW}

In the first step we construct $J^1$ and $D^1$ that satisfy condition~\ref{def:entangled_hypertree_decomposition:precision} of Definition~\ref{def:entangled_hypertree_decomposition} as shown in Figure~\ref{fig:step1_ehd_to_ghd}.
The left part of this figure depicts the hypergraphs associated with the incidence graphs $\J,\Jeins$.
\begin{figure}
  \centering
\begin{tikzpicture}[
	decoration = {
		snake,
		pre length=0pt,
		post length=4pt,
		amplitude=1.75pt,
		segment length=5pt
	}
]
	\begin{scope}[name=pre, scale=\myscale]
	\begin{scope}[local bounding box=prehg]
		\pic[scale=\myscale] {hypergraph};
	\end{scope}
	\node (prehg-label) at (0,3.5) {$J$};
	\begin{scope}[shift={(2,1)}]
		\node (e) at (.25,2) {$e$};
		\node (f) at (1.25,2) {$f$};
		\node (i) at (.75,1.25) {$i$};
		\node (g) at (.75,.5) {$g$};
		\draw[draw=blau] \convexpath{e,f,i}{.3cm};
		\draw[draw=gruen] \convexpath{i,g}{.3cm};
	\end{scope}
	\begin{scope}[shift={(6,2)}, local bounding box=prenode]
		\node[draw, rectangle, anchor=west] (t) at (0,0) {
			$t$:
			\begin{array}[t]{rcl}
				\DBag(t)& \!\!\!=\!\!\! &\set{ a,b,c,d,e,f,g } \\ 
				\DCover(t)& \!\!\!=\!\!\! &\set{ \set{a,b,c}, \set{c,d}, \textcolor{blau}{\set{ e,f,i }}, \textcolor{gruen}{\set{g,i}}}
			\end{array}
		};
	\end{scope}
	\end{scope}

	\begin{scope}[shift={(0,-5)}, scale=\myscale]
		\begin{scope}[local bounding box=posthg]
			\pic[scale=\myscale] {hypergraph};
		\end{scope}
		\node (posthg-label) at (0,3.5) {$J^{1}$};
		\begin{scope}[shift={(2,1)}]
			\node (e) at (.25,2) {$e$};
			\node (f) at (1.25,2) {$f$};
			\node (i) at (.75,1.25) {$i$};
			\node[draw=orangsch, circle, inner sep=.2cm] (g) at (.75,.5) {$g$};
			\draw[draw=blau] \convexpath{e,f,i}{.3cm};
			\draw[draw=rot] \convexpath{e,f}{.4cm};
			\draw[draw=gruen] \convexpath{i,g}{.3cm};
		\end{scope}
		\begin{scope}[shift={(6,2)}, local bounding box=postnode]
			\node[draw, rectangle, anchor=west] (t) at (0,0) {
				$t$:
				\begin{array}[t]{rcl}
					\DBag^{1}(t)& \!\!\!=\!\!\! &\set{ a,b,c,d,e,f,g } \\ 
					\DCover^{1}(t)& \!\!\!=\!\!\! &\set{ \set{a,b,c}, \set{c,d}, \textcolor{rot}{\set{ e,f }}, \textcolor{orangsch}{\set{g}}}
				\end{array}
			};
		\end{scope}
	\end{scope}

	\draw[-latex, decorate] ($(prehg.south) - (0,.2)$) -- ($(posthg.north) + (0,.2)$) node [midway, label={[]right:{insert $\set{e,f}$, $\set{g}$.}}] {};
	\draw[-latex, decorate] ($(prenode.south |- 0, 0 |- prehg.south) - (0,.2)$) -- ($(prenode.south |- 0, 0 |- posthg.north) + (0,.2)$) node [midway, label={[text width=8em]right:{replace $\set{e,f,i}$, $\set{g,i}$ in cover.}}] {};
\end{tikzpicture}
\caption{Example for a modification occurring in step 1 that ensures that condition \ref{def:entangled_hypertree_decomposition:precision} of Definition~\ref{def:entangled_hypertree_decomposition} is satisfied.}
  \label{fig:step1_ehd_to_ghd}
\end{figure}

In the depicted tree-node $t$ of $D$ (on the top right) condition \ref{def:entangled_hypertree_decomposition:precision} is violated since $i \not\in \DBag(t)$. Our construction inserts new
blue nodes representing the
hyperedges $\set{e,f}$ and $\set{g}$ into $J$, which we can then use to replace the corresponding blue nodes (or, hyperedges) in $\DCover(t)$. After this modification, $t$ no longer violates condition \ref{def:entangled_hypertree_decomposition:precision}. We repeatedly apply this operation until condition \ref{def:entangled_hypertree_decomposition:precision} is no longer violated.

We must be careful, since this \enquote{naive} modification also changes the number of homomorphisms. Therefore, the procedure described above might align the numbers $\hom{J^1}{\I}$ and $\hom{J^1}{\I'}$. To avoid this, we use Lemma \ref{lemma:adding_edges_keeps_distinguishing} and insert in every step of the modification a suitable number of copies (i.e.\ maybe more than just one), so that after every modification the number of homomorphisms remains distinguishing. We deal with the additional copies by inserting additional tree-nodes as leaves, each having one copy in its cover (and nothing else) and the neighbourhood of the copy as its bag.
Altogether, this gives us $\hom{J^1}{\I} \neq \hom{J^1}{\I'}$.

In the second step, we modify $J^1$ and $D^1$ such that they also satisfy condition \ref{def:entangled_hypertree_decomposition:connectedness_of_edges}, i.e.\ such that the resulting $D^2$ is an ehd of $J^2$. Consider
Figure~\ref{fig:step2_ehd_to_ghd}.
The left part of the figure depicts the hypergraphs associated with the incidence graphs $\Jeins$ and $\Jzwei$.
In this example, the hyperedge $\textcolor{blau}{e}$ in $J^1$ induces 3 connected components in $D^1$. Therefore, we insert $2$ copies of $\textcolor{blau}{e}$ --- namely, one for each additional connected component. We apply this procedure until there is no hyperedge left that induces multiple connected components, which means that condition~\ref{def:entangled_hypertree_decomposition:connectedness_of_edges} is met.

Again, this naive approach might equalise the number of
homomorphisms. Therefore, we use Lemma
\ref{lemma:adding_edges_keeps_distinguishing} to insert a suitable,
large enough number of copies, so that after every modification the
number of homomorphisms remains distinguishing. We deal with the
excess copies in the same way as before.
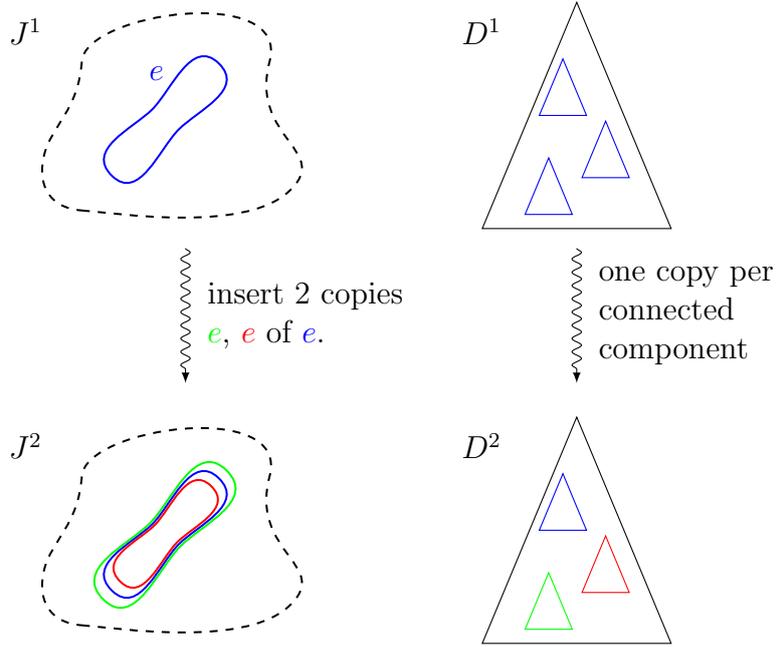
\begin{figure}
  \centering
\begin{tikzpicture}[
	decoration = {
		snake,
		pre length=0pt,
		post length=4pt,
		amplitude=1.75pt,
		segment length=5pt,
	},
	subtree/.style={
		isosceles triangle,
		isosceles triangle apex angle=45,
		draw,
		rotate=90,
		anchor=left corner,
		minimum height =.75cm
	}
]
	\begin{scope}[name=pre, scale=\myscale]
	\begin{scope}[local bounding box=prehg]
		\pic[scale=\myscale] {hypergraph};
	\end{scope}
	\node (prehg-label) at (0,3.5) {$J^{1}$};
	\begin{scope}[shift={(2.5,2)}]
		\pic[draw=blau] {hyperedge};
	\end{scope}
	\node[text=blau] (e) at (2.3, 2.75) {$e$};
	\begin{scope}[shift={(8,0)}, local bounding box=prenode]
		\node (preD-label) at (0,3.5) {$D^{1}$};
		\node[draw=blau, subtree] (sub1) at (1,2){};
		\node[draw=blau, subtree] (sub2) at (1.75,.9){};
		\node[draw=blau, subtree] (sub3) at (.75,.25){};
		\node[isosceles triangle,
			isosceles triangle apex angle=45,
			draw,
			rotate=90,
			anchor=left corner,
			minimum height =3cm
		] (preD)at (0,0){};
		
	\end{scope}
	\end{scope}

	\begin{scope}[shift={(0,-5.5)}, scale=\myscale]
		\begin{scope}[local bounding box=posthg]
			\pic[scale=\myscale] {hypergraph};
		\end{scope}
		\node (posthg-label) at (0,3.5) {$J^{2}$};
		\begin{scope}[shift={(2.5,2)}]
			\pic[draw=blau] {hyperedge};
			\pic[draw=gruen, scale=1.15] {hyperedge};
			\pic[draw=rot, scale=0.85] {hyperedge};
		\end{scope}
		\begin{scope}[shift={(8,0)}, local bounding box=postnode]
			\node (postD-label) at (0,3.5) {$D^{2}$};
			\node[draw=blau, subtree] (sub1) at (1,2){};
			\node[draw=rot, subtree] (sub2) at (1.75,.9){};
			\node[draw=gruen, subtree] (sub3) at (.75,.25){};
			\node[isosceles triangle,
				isosceles triangle apex angle=45,
				draw,
				rotate=90,
				anchor=left corner,
				minimum height =3cm
			] (postD)at (0,0){};
		\end{scope}
	\end{scope}

	\draw[-latex, decorate] ($(prehg.south) - (0,.2)$) -- ($(posthg.north) + (0,.2)$) node [midway, label={[text width=7em]right:{insert $2$ copies \textcolor{gruen}{$e$}, \textcolor{rot}{$e$} of \textcolor{blau}{$e$}.}}] {};
	\draw[-latex, decorate] ($(preD.lower side |- 0, 0 |- prehg.south) - (0,.2)$) -- ($(preD.lower side |- 0, 0 |- posthg.north) + (0,.2)$) node [midway, label={[text width=6em]right:{one copy per connected component}}] {};
\end{tikzpicture}
  \caption{Example for a modification occurring in step 2 that ensures that condition \ref{def:entangled_hypertree_decomposition:connectedness_of_edges} of Definition~\ref{def:entangled_hypertree_decomposition} is satisfied.}
  \label{fig:step2_ehd_to_ghd}
\end{figure}

\medskip

\section{Details omitted in Section~\ref{sec:normal_form}}
\label{appendix:logic}

\noindent
\textbf{The formulas constructed in Example~\ref{example:formula} are in $\NGC{2}$:}

\begin{proof} \hspace{10mm}
\begin{description}
\item[$\phi_1$:]    
Note that $(\top\und\phi_1)$ is a sentence in $\NGC{2}$ (use rules \ref{item:syntaxdef:2})
and \ref{item:syntaxdef:7}) of the syntax definition).

\item[$\phi_2$:]
Note that $(\top\und \phi_2)$ is a sentence in $\NGC{2}$
(use rules \ref{item:syntaxdef:1}), \ref{item:syntaxdef:6}),
\ref{item:syntaxdef:7})).

\item[$\phi_4$:]
Note that $(\top\und\alpha)$ is a sentence in $\NGC{2}$ (by rules 
\ref{item:syntaxdef:2}), \ref{item:syntaxdef:4}),
\ref{item:syntaxdef:1}), \ref{item:syntaxdef:6}), \ref{item:syntaxdef:5})).
Applying rules \ref{item:syntaxdef:1}), \ref{item:syntaxdef:6}),
\ref{item:syntaxdef:5}), \ref{item:syntaxdef:7}), \ref{item:syntaxdef:4}))
yields that
$(\top\und\phi_4)$ is a sentence in $\NGC{2}$.

\item[$\phi_5$:]
Note that $(\LogGuard{g_{i,j}}\und\vartheta_{i,j})\in\NGC{2}$ for
the guard function $g_{i,j}$ with
$\Dom{g_{i,j}}=\set{i,j}$ and $g_{i,j}(i)=g_{i,j}(j)=1$
(use rules \ref{item:syntaxdef:1}), \ref{item:syntaxdef:5}),
\ref{item:syntaxdef:6})).
Furthermore, $\Delta_{i,j}=\LogGuard{g_{i,j}}$ and hence
$(\LogGuard{g_{i,j}}\und\psi_{i,j})\in\NGC{2}$
(use rule \ref{item:syntaxdef:7})).
Rule \ref{item:syntaxdef:5}) yields that
$(\LogGuard{g}\und\psi)\in\NGC{2}$ for the guard function $g$ with
$\Dom{g}=[3]$ and $g(1)=g(2)=g(3)=1$.
Rules \ref{item:syntaxdef:3}) and \ref{item:syntaxdef:5}) yield that
$(\LogGuard{g}\und\chi)\in\NGC{2}$.
By rule \ref{item:syntaxdef:5}) we obtain that
$(\LogGuard{g}\und (\chi\und\psi))\in\NGC{2}$.
Note that $\Delta=\LogGuard{g}$.
Applying rules \ref{item:syntaxdef:6}) and \ref{item:syntaxdef:7})
yields that
$(\top\und\phi_5)\in\NGC{2})$.

\item[$\chi:$]
Rule \ref{item:syntaxdef:5}) yields that $(\top \und
\Und_{i=1}^5\phi_i)$ is a sentence in $\NGC{2}$.
\qedhere
\end{description}

\end{proof}  
 
\medskip

\section{Details omitted in Section~\ref{sec:recursive_def}}
\label{appendix:recursive_def}

\medskip

\subsection{Proof of Lemma~\ref{lemma:glueing_preserves_real_guards}}
\label{appendix:glueing_preserves_real_guards}
\newcounter{glueing_preserves_real_guards}
\begin{proof}[\textbf{\textup{Proof of Lemma~\ref{lemma:glueing_preserves_real_guards}:}}] \ \\
The proof consists of a lengthy (but otherwise straightforward) case
distinction that only checks requirements along Definition~\ref{def:glueing}.

Let $\kLI_1=(I_1,r_1,b_1,g_1)$ and $\kLI_2=(I_2,r_2,b_2,g_2)$.
Let $\kLI \isdef \glue{\kLI_1}{\kLI_2}$ and 
let $(I,r,b,g)=\kLI$.
By assumption, $\kLI_1$ and $\kLI_2$ have real guards. We have to show
that also $\kLI$ has real guards.
To this end,
consider an arbitrary $i \in \Dom{r}$. Let $j = g(i)$. We have to show the following: 
\begin{enumerate}[(A)]
	\item $j \in \Dom{b}$, \ and \label{appendix:glueing_preserves_real_guards:a}
	\item $(b(j),r(i)) \in \E{\inzidenz{}}$. \label{appendix:glueing_preserves_real_guards:b}
\end{enumerate}

\noindent
We distinguish between two cases:
\medskip

\textbf{Case 1:} $i \in \Dom{r_1}$.
	By Definition~\ref{def:glueing} we then have: $r(i)=[(r_1(i),1)]_{\sim_R}$ and $j=g(i)=g_1(i)$.
	Since $\kLI_1$ has real guards, we have 
	\begin{enumerate}[(1)]
		\item $j \in \Dom{b_1}$, \ and \label{appendix:glueing_preserves_real_guards:1}
		\item $(b_1(j),r_1(i)) \in \E{\inzidenz{1}}$. \label{appendix:glueing_preserves_real_guards:2}
		\setcounter{glueing_preserves_real_guards}{\value{enumi}}
	\end{enumerate}
	By Definition~\ref{def:glueing} and \eqref{appendix:glueing_preserves_real_guards:1} we obtain that $j \in \Dom{b}$. Thus, \eqref{appendix:glueing_preserves_real_guards:a} is proved. To prove \eqref{appendix:glueing_preserves_real_guards:b}, we proceed as follows:
	From \eqref{appendix:glueing_preserves_real_guards:2} we know for $e \isdef b_1(j)$ and $v \isdef r_1(i)$ that $(e,v) \in \E{\inzidenz{1}}$. By Definition~\ref{def:glueing} this yields:
	\begin{enumerate}[(1)]
		\setcounter{enumi}{\value{glueing_preserves_real_guards}}
		\item $( [(e,1)]_{\sim B} , [(v,1)]_{\sim_R} ) \ \in \
                  \E{\inzidenz{}}$ \ and \label{appendix:glueing_preserves_real_guards:3}
		\item $[(e,1)]_{\sim_B} \ \in \ \bV{\inzidenz{}}$ \ and \label{appendix:glueing_preserves_real_guards:4}
		\item $[(v,1)]_{\sim R} \ \in \ \rV{\inzidenz{}}$.\label{appendix:glueing_preserves_real_guards:5}
                \end{enumerate}
                \smallskip

                \noindent
By Definition~\ref{def:glueing} we have: $b(j) = [(b_1(j),1)]_{\sim_B} = [(e,1)]_{\sim_B}$.
        Furthermore, recall that $r(i) = [(r_1(i),1)]_{\sim_R} = [(v,1)]_{\sim_R}$.        
	Thus, \eqref{appendix:glueing_preserves_real_guards:3} means that $(b(j),r(i)) \in \E{\inzidenz{}}$, i.e., \eqref{appendix:glueing_preserves_real_guards:a} is proved.\medskip
        
      \textbf{Case 2:} $i \not\in \Dom{r_1}$. Then, $i \in
      \Dom{r_2}$. \eqref{appendix:glueing_preserves_real_guards:a} and
      \eqref{appendix:glueing_preserves_real_guards:b} can be proven
      in the same way as in Case~1 by replacing every subscript 1 with
      subscript 2.
      This completes the proof of Lemma~\ref{lemma:glueing_preserves_real_guards}.
\end{proof}

\medskip

\subsection{Detailed Proof of Part~\ref{item:GLIkInIEHWk} of Theorem~\ref{thm:IEHWkIsGLIk}}
\label{appendix:GLIk-in-IEHWk}

\begin{proof}[\textbf{\textup{Proof of part \ref{item:GLIkInIEHWk} of Theorem~\ref{thm:IEHWkIsGLIk}:}}] \ \\
By induction on the definition of $\GLIk$ we show for every
$\kLI\in\GLIk$ that there exists an 
entangled hypertree decomposition $D=(T,\DBag,\DCover)$ of $\I_\kLI$ of width $\leq k$.
To be able to carry out the induction step, we additionally ensure that there is a tree-node $\omega\in\V{T}$ with $\DCover(\omega)=\Img{b_\kLI}$.

The idea is fairly simple: From Lemma~\ref{lemma:GLIk_has_real_guards}
we know that $\kLI$ has real guards. Thus, the (at most $k$) blue
labels given by $b$ guard all the red labels given by $r$ in the sense
that $(b({g(i)}),r(i))\in\E{\I}$ for every $i\in\Dom{r}$. Hence, we
can use the set of labeled blue vertices to cover the bag containing
all the labeled red vertices. 
This is utilised as follows.
\smallskip

\noindent \textbf{Base case}: \begin{enumerate}[wide]
\item
Let $\kLI=(\I,r,b,g)$ be a $k$-labeled incidence graph with $\rVI=\Img{r}$, $\bVI=\Img{b}$, and with real guards.
We choose $D \isdef (T, \DBag, \DCover)$ as follows: Let $t$ be a single tree-node,
$\V{T}\deff\set{t}$, $\E{T}\deff\emptyset$, $\DBag(t) \isdef \rV{\inzidenz{}} = \Img{r}$,
$\DCover(t) \isdef \bV{\inzidenz{}} = \Img{b}$.
Clearly, $D$ is an ehd of $\I$ of width $|\Img{b}|\leq k$, and $\omega:=t$ satisfies $\DCover(\omega)=\Img{b}$.
\setcounter{InductionConuter}{\value{enumi}}
\end{enumerate}
\smallskip

\noindent \textbf{Inductive step}: \
Let $\kLI = (\I, b, r, g) \in\GLIk$, let 
$D = (T, \DBag, \DCover)$ be an ehd of $\I$ of width $\leq k$,
and let 
$\omega\in \V{T}$ be such that
$\DCover(\omega) = \Img{b}$.
\smallskip

\begin{enumerate}[wide]
\setcounter{enumi}{\value{InductionConuter}}
\item Let $\kLI' \isdef \reclaimR{\kLI}{\myR}$ for some
$\myR \subseteq \Def{r}$. Obviously, $D$ also is an ehd of $\I_{\kLI'}$,
since $\I_{\kLI'} = \I$. And since $b_{\kLI'}=b$, we also have
$\DCover(\omega)=\Img{b_{\kLI'}}$.
I.e., $D$ is an ehd of  $\I_{\kLI'}$ with the desired properties. 
\smallskip
  
\item Let $\kLI'\isdef \reclaimB{\kLI}{\myB}$ for some $\myB \subseteq \Def{b} \setminus \Img{g}$.
Then, $D$ is also on ehd of $I_{\kLI'}$, since $\I_{\kLI'}=\I$.
But note that we might have $\Img{b_{\kLI'}}\varsubsetneq\Img{b}$ and hence
$\DCover(\omega)\varsupsetneq\Img{b_{\kLI'}}$.
To fix this, we create a new tree-node $\omega'$ that we insert into
$T$ as a new leaf adjacent to $\omega$, and we let
$\DCover(\omega')\deff \Img{b_{\kLI'}}$ and $\DBag(\omega')\deff
\bigcup_{e\in\DCover(\omega')}N_{\I}(e)$. 
This results in an ehd $D'$ of $\I_{\kLI'}$ with the desired
properties.
\smallskip

\item
Let $\kLI' \deff \switch{\kLI}{f}$ for a transition $f$ for $g_\kLI$. Let $\I'\deff\I_{\kLI'}
$. As in the previous step, we create a new tree-node $\omega'$ that we insert into $T$ as
a new leaf adjacent to $\omega$, and we let $\DCover'(\omega')\deff\Img{b_{\kLI'}}$ and
$\DBag'(\omega')=\bigcup_{e\in\DCover'(\omega')}N_{\I'}(e)$.

To define $\DCover'(t)$ for all $t\in\V{T}\setminus\set{\omega'}$,
need the following notation. By definition,
$\kLI'=\glue{\kLI_1}{\kLI_2}$ where $\kLI_1\deff \kLIf$ and 
$\kLI_2\deff \reclaimB{\kLI}{\myB}$ for $\myB \isdef \Def{g}
\intersect \Img{f}$.
For $i\in[2]$ let $\I_i\deff \I_{\kLI_i}$, and let
$\pi_{i,R}$ and $\pi_{i,B}$ be the mappings provided at the end of
Definition~\ref{def:glueing}.
For all $t\in\V{T}\setminus\set{\omega'}$ we let $\DCover'(t)\deff \setc{\pi_{2,B}(e)}{e\in\DCover(t)}$ and
$\DBag'(t)\deff\bigcup_{e'\in\DCover'(t)}N_{\I'}(e)$.

It is straightforward to verify that this results in an ehd $D'$ of $\I'$ with the desired properties.
\smallskip

\item
Let $\kLI'\deff \glue{\kLI_1}{\kLI_2}$ where, for each $i\in[2]$,
$\kLI_i=(\I_i,r_i,b_i,g_i)\in\GLIk$, $D_i=(T_i,\DBag_i,\DCover_i)$
is an ehd of $\I_i$ of width $\leq k$, and $\omega_i\in\V{T_i}$ is a
tree-node satisfying $\DCover(\omega_i)=\Img{b_i}$. 
W.l.o.g.\ we assume that $\V{T_1}\cap\V{T_2}=\emptyset$.
Let $\I'\deff \I_{\kLI'}$.

We create a new tree-node $\omega'$ and we let
$\DCover'(\omega')\deff\Img{b_{\kLI'}}$ and
$\DBag'(\omega')=\bigcup_{e\in\DCover'(\omega')}N_{\I'}(e)$.
We choose
$T'= (\V{T'},\E{T'})$ with
$\V{T'}\deff \V{T_1}\cup\V{T_2}\cup\set{\omega'}$
and
$\E{T'}\deff \E{T_1}\cup\E{T_2}\cup\bigset{\set{\omega',\omega_1}\;,\;\set{\omega',\omega_2}}$.

To define $\DCover'(t)$ for all $t\in\V{T_1}\cup\V{T_2}$, we
need the following notation.
For $i\in[2]$ let
$\pi_{i,R}$ and $\pi_{i,B}$ be the mappings provided at the end of
Definition~\ref{def:glueing}.
For all $i\in[2]$ and all $t\in\V{T_i}$ we let $\DCover'(t)\deff \setc{\pi_{i,B}(e)}{e\in\DCover_i(t)}$ and
$\DBag'(t)\deff\bigcup_{e'\in\DCover'(t)}N_{\I'}(e)$.

It is straightforward to verify that $D'\deff (T',\DBag',\DCover')$ is an ehd of $\I'$  with the desired properties.
\end{enumerate}
This completes the proof of part \ref{item:GLIkInIEHWk} of Theorem~\ref{thm:IEHWkIsGLIk}.
\end{proof}
 
\medskip

\subsection{Detailed Proof of Part~\ref{item:IEHWkInGLIk} of Theorem~\ref{thm:IEHWkIsGLIk}}
\label{appendix:IEHWk-in-GLIk}

\begin{proof}[\textbf{\textup{Proof of part \ref{item:IEHWkInGLIk} of Theorem~\ref{thm:IEHWkIsGLIk}:}}] \ \\
Given an $\I\in\IEHWk$, the proof proceeds as follows.
We pick an ehd $D=(T,\DBag,\DCover)$ of $\I$ of width $\leq k$ that
has a particularly suitable shape (details follow below), and we
single out a suitable tree-node $\omega$ as the ``root'' of $T$. For
every $t\in\V{T}$, the information provided by $\DCover(t)$ is viewed
as a description of a guarded $k$-labeled incidence graph $\kLI_t$
corresponding to the base case of
Definition~\ref{def:k_guarded_incidence_graphs}. 
We then perform a bottom-up traversal of the rooted tree $(T,\omega)$ and glue 
together
all the $\kLI_t$'s.
But this has to be done with care: before glueing them, we have to
ensure that the guards are compatible; we achieve this by adequately
changing labels and applying transitions, so that finally we end up
with an $\kLI\in\GLIk$ whose incidence graph $\I_\kLI$ is isomorphic
to $\I$.
Below, we give a detailed description of the construction;
a specific example of the construction can be found
in Appendix~\ref{appendix:subsec:example}.

We start by picking an ehd $D=(T,\DBag,\DCover)$ of $\I$ and an
$\omega\in\V{T}$ as provided by the following lemma.

\begin{lemma}\label{lemma:monotone_ehd}\label{lem:BinaryMonotone}
For every $\I\in\IEHWk$ there is an ehd $D=(T,\DBag,\DCover)$ of $\I$
and an $\omega\in\V{T}$ with $|\DCover(\omega)|\leq k$ such that the
rooted tree $(T,\omega)$ is binary and monotone, i.e.,
every tree-node has at most 2 children, and
for all parent-child pairs $(t_p,t_c)$ we have $|\DCover(t_p)| \geq |\DCover(t_c)|$.
\end{lemma}

\begin{proof}
We start with an arbitrary ehd $D=(T,\DBag,\DCover)$ of $\I$ of width $\leq k$.
Fix an $\omega\in \V{T}$ where $|\DCover(\omega)|$ is as large as
possible. Perform a bottom-up traversal of the rooted tree
$(T,\omega)$.
Whenever encountering a parent-child pair $(t_p,t_c)$ where
$|\DCover(t_p)| < |\DCover(t_c)|$, insert elements from $\DCover(t_c)$
into $\DCover(t_p)$ to ensure that $|\DCover(t_p)|=|\DCover(t_c)|$,
and modify $\DBag(t_p)$ such that condition
\ref{def:entangled_hypertree_decomposition:precision} of 
Definition~\ref{def:entangled_hypertree_decomposition} is satisfied.

To ensure that the tree is binary, we perform a top-down traversal of $(T,\omega)$, and for every node
$t$ that has $n\geq 3$ children $t_1,\ldots,t_n$ we introduce $n{-}2$
new nodes $t'_2,\ldots,t'_{n-1}$ with $\DBag(t'_i)=\DBag(t)$ and
$\DCover(t'_i)=\DCover(t)$,
we insert edges $\set{t,t'_2}$, $\set{t'_{n-1},t_n}$, and $\set{t'_i,t'_{i+1}}$ for
all $i\in\set{2,\ldots,n{-}2}$; and
for every $i\in\set{2,\ldots,n{-}1}$ we replace the edge $\set{t,t_i}$
by the new edge $\set{t'_i,t_i}$.
It is straightforward to verify that the resulting tree is an ehd of
$\I$ with the desired properties.
This completes the proof of Lemma~\ref{lem:BinaryMonotone}.
\end{proof}

We perform a top-down traversal of the rooted tree $(T,\omega)$ to construct a
\emph{$k$-colouring of $\I$ according to $D$}, i.e.\ a mapping
$c:\bVI\to [k]$ with the following properties, where $c(t)\deff
\setc{c(e)}{e\in\DCover(t)}$ for all $t\in\V{T}$:  
For all $t \in \V{T}$ and all $e,e' \in \DCover(t)$ with $e\neq e'$ we
have $c(e) \neq c(e')$, and for all parent-child pairs $(t_p,t_c)$ we
have $c(t_p) \supseteq c(t_c)$. 

For all $t\in\V{T}$ and $j\in [k]$ we let $c(t,j)\deff e$ for the
particular  $e\in\DCover(t)$ with $c(e)=j$; and in case that no such
$e$ exists we let 
$c(t,j)\deff\undefined$.
\medskip

We perform a further top-down traversal of $(T,\omega)$ to construct a \emph{schedule for $D$}, i.e., a mapping         
$s: \V{T} \times \rVI \to \bVI$ with the following properties:
\begin{enumerate}
\item
For all $t \in \V{T}$ and all $v \in \DBag(t)$,
$s(t,v)$ is a blue node of $\I$ that belongs to $\DCover(t)$ and is
adjacent in $\I$ to the red node $v$.
\item
For all parent-child pairs $(t_p, t_c)$ and for all
$v \in \DBag(t_p) \intersect \DBag(t_c)$, the following is satisfied:
If $s(t_p, v) \in \DCover(t_c)$, then $s(t_c, v) = s(t_p, v)$. 
\end{enumerate}
\smallskip

Finally, we let $n\deff |\rVI|$, and we fix an arbitrary list $v_1,\ldots,v_n$ of all
red nodes of $\I$.
We use this list and the mappings $s$ and $c$ to define for every
$t\in\V{T}$
the guarded $k$-labeled incidence graph $\kLI_t$ \emph{associated with the
  tree-node $t$} as follows:
$\kLI_t\deff (\I_t,r_t,b_t,g_t)$, where
$\rV{\I_t}\deff \DBag(t)$,
$\bV{\I_t}\deff\DCover(t)$,
$\E{\I_t}\deff \E{\I}\cap (\bV{\I_t}\times \rV{\I_t})$,
and we let $\Dom{r_t}\deff\Dom{g_t}\deff \setc{i\in [n]}{v_i\in\DBag(t)}$ and $r_t(i)\deff v_i$ and $g_t(i)\deff c(s(t,v_i))$
for all $i\in\Dom{g_t}$,
and we let $\Dom{b_t}\deff \setc{c(e)}{e\in\DCover(t)}$ and $b_t(j)\deff c(t,j)$ for all
$j\in\Dom{b_t}$.
It is easy to see that $\kLI_t\in\GLIk$ (according to the base case of
Definition~\ref{def:k_guarded_incidence_graphs}).
Below, we will make extensive use of the particular
functions $r_t,b_t,g_t$ of $\kLI_t$. 

All that remains to be done is to glue together the $\kLI_t$'s for all
$t\in\V{T}$.
This has to be done with care. To describe the
construction, we use the following notation.
For every $t\in\V{T}$ let $T_t$ denote the subtree of the rooted tree
$(T,\omega)$ rooted at $t$. By $\assig{t}$ we denote the \emph{incidence graph
associated with $T_t$}, i.e.,
$\rV{\assig{t}}=\bigcup_{t'\in\V{T_t}}\DBag(t')$,
$\bV{\assig{t}}=\bigcup_{t'\in\V{T_t}}\DCover(t')$, and
$\E{\assig{t}}=\E{\I}\cap (\bV{\assig{t}} \times \rV{\assig{t}})$.

We perform a bottom-up traversal of $(T,\omega)$ to construct for each
$t\in\V{T}$ a guarded $k$-labeled incidence graph $\kLI'_t=(\I'_t,r'_t,b'_t,g'_t)\in\GLIk$
and an isomorphism $\pi_t=(\pi_{t,R},\pi_{t,B})$ from $\assig{t}$ to
$\I'_t$ and ensure that $g'_t=g_t$, $\Dom{r'_t}=\Dom{r_t}$ and
$r'_t(i)=\pi_{t,R}\big(r_t(i)\big)$ for all $i\in\Dom{r_t}$, and
$\Dom{b'_t}=\Dom{b_t}$ and
$b'_t(j)=\pi_{t,B}\big(b_t(j)\big)$ for all $j\in\Dom{b_t}$.

Note that once having achieved this, the proof of
part~\ref{item:IEHWkInGLIk} of Theorem~\ref{thm:IEHWkIsGLIk} is
complete because $ \I=\assig{\omega} \isom \I'_{\omega}$; we are
therefore done by choosing
$\kLI\deff \reclaimB{\big(\reclaimR{\kLI'_\omega}{\Dom{r'_\omega}}\big)}{\Dom{b'_\omega}}$.

For each \emph{leaf} $t$ of the rooted tree $(T,\omega)$, we can
simply choose $\kLI'_t\deff \kLI_t$.

Let us now consider a tree-node $t_p$ that has exactly one child $t$.
Our goal is to combine the $k$-labeled incidence graph $\kLI_{t_p}$
associated with the tree-node $t_p$ with the $k$-labeled incidence
graph $\kLI'_{t}$ already constructed for $T_t$ into a single
$k$-labeled incidence graph $\kLI'_{t_p}$ for $T_{t_p}$. We do this by
glueing them at the red nodes in $\DBag(t)\cap\DBag(t_p)$ and the blue
nodes in $\DCover(t)\cap\DCover(t_p)$. But we have to take care to use the
correct labels and end up with the correct guard function.
To this end, we first release from $\kLI'_t$ all the red labels of
nodes in $\DBag(t)\setminus\DBag(t_p)$, i.e., we let
\[
   A_1 \ \ \deff\ \ \reclaimR{\kLI'_t\,}{\,(\Dom{r_t}\setminus\Dom{r_{t_p}})\,}\;.
\]
Note that $g_{A_1}=g_t-(\Dom{r_t}\setminus\Dom{r_{t_p}})$. 
Let $\myB\deff\setc{j\in c(t)}{c(t,j)\neq c(t_p,j)}$; this is the set
of blue labels that are used in $t$ for a different node than in $t_p$.
Consider the set
$\myR\deff\setc{i\in\Dom{r_{t}}\cap\Dom{r_{t_p}}}{s(t,v_i)\neq
  s(t_p,v_i)}$;
this is the set of red labels used in both $t$ and $t_p$ but whose
guard changes when moving from $t$ to its parent $t_p$.

In case that $\myR=\emptyset$, we let
\[
  A_2 \ \deff \ \reclaimB{A_1\,}{\myB}
  \quad\text{and}\quad
  \kLI'_{t_p}\deff\ \glue{\kLI_{t_p}}{A_2}\;.
\]
It can be
verified\footnote{see Appendix~\ref{appendix:IEHW--in-GLIk-MissingPieces}}
that $\myB\subseteq\Dom{b_{A_1}}\setminus\Img{g_{A_1}}$;
hence $A_2\in\GLIk$. It can also be
verified\footnote{see Appendix~\ref{appendix:IEHW--in-GLIk-MissingPieces}}
that $g_{t_p}$ and
$g_{A_2}$ are compatible, and thus also $\kLI'_{t_p}\in\GLIk$.

In case that $\myR\neq\emptyset$ we let $f$ be the restriction of
$g_{t_p}$ to $\myR$, i.e., $\Dom{f}=\myR$ and $f(i)=g_{t_p}(i)$ for
all $i\in\Dom{f}$.
It can be
verified\footnote{see Appendix~\ref{appendix:IEHW--in-GLIk-MissingPieces}}
that $\Img{f}\subseteq \myB$ and that $f$ is a transition for $g_{A_1}$.
We let $\myBHat \isdef (\Img{b_{A_1}} \intersect \Img{f}) \setminus
\Img{g_{A_1}}$ and $\myBTilde\deff \myB\setminus\Img{f}$. Note that
$\myBHat \intersect \myBTilde = \emptyset$. We choose $\hat{A}_1
\isdef \reclaimB{A_1}{\myBHat}$. Obviously $\hat{A}_1 \in \GLIk$ and
$f$ is also a transition for $\hat{A}_1$. Let $\tilde{A}_1 \isdef
\switch{\hat{A}_1}{f}$, which is also obviously in $\GLIk$ and then
$A_2 \isdef \reclaimB{\tilde{A}_1}{\myBTilde}$. I.e.\ we have 
\[
  A_2 \ \deff \ 
  \reclaimB{
    \; \big( \,
      \switch{
        \; \big( \,
          \reclaimB{A_1}{\myBHat} 
        \, \big) \; }
      {f}
    \,\big) \;}
  {\myBTilde} \;.
\]
Finally, we let $\kLI'_{t_p} \isdef \glue{\kLI_{t_p}}{A_2}$.
It can be
verified\footnote{see Appendix~\ref{appendix:IEHW--in-GLIk-MissingPieces}} that
$\myBTilde\subseteq\Dom{b_{\tilde{A}_1}}\setminus\Img{g_{\tilde{A}_1}}$; hence
$A_2\in\GLIk$.
It can also be
verified\footnote{see Appendix~\ref{appendix:IEHW--in-GLIk-MissingPieces}}
that $g_{t_p}$ and
$g_{A_2}$ are compatible, and thus also $\kLI'_{t_p}\in\GLIk$.

In both cases it can be verified that $\kLI'_{t_p}$ has the desired
properties.
\smallskip

The case of a tree-node $t_p$ that has two children $t$ and $t'$ can
be handled in the same way: Let $A_2$ be constructed for $t$ as
described above, and let $A'_2$ be constructed for $t'$ in the
analogous way; the construction is completed by letting
\[
  \kLI'_{t_p} \ \deff \ \glue{\;\glue{\kLI_{t_p}}{A_2}\,}{A'_2\,}\,.
\]
This completes the proof of part \ref{item:IEHWkInGLIk} of
Theorem~\ref{thm:IEHWkIsGLIk}. 
\end{proof}

See Appendix~\ref{appendix:subsec:example} for an example illustrating
the above construction. It helps to get some insight into $k$-labeled
incidence graphs and especially into transitions.
 
\medskip

\subsection{Missing pieces in the proof of part \ref{item:IEHWkInGLIk}
  of Theorem~\ref{thm:IEHWkIsGLIk}}
\label{appendix:IEHW--in-GLIk-MissingPieces}

\smallskip

\noindent
\textbf{In case that $\myR = \emptyset$:}

\begin{proof}[{Proof that $\myB \subseteq \Dom{b_{A_1}} \setminus \Img{g_{A_1}}$}]
It is clear that $\myB \subseteq \Dom{b_{A_1}}$ holds since $c(t) =
\Dom{b_{\kLI_{t}}} = \Dom{b_{\kLI_t}}$. Assume for contradiction, that
there exists an $i \in \Dom{g_{A_1}}$ such that $g_{A_1}(i) = \ell \in
\myB$. We have $b_{A_1}(\ell) = b_{\kLI_{t}'}(\ell)$ and $i \in
\Dom{r_{t}} \intersect \Dom{r_{t_p}}$, thus $c(s(t, v_i)) =
\ell$. Since $\ell \in \myB$ we have $c(t, \ell) \neq c(t_p, \ell)$,
which means $s(t, v_i) \neq s(t_p, v_i)$. But this contradicts our
assumption that $\myR = \emptyset$. Therefore, such an $i$ cannot
exist. 
\end{proof}

\begin{proof}[{Proof that $g_{t_p}$ and $g_{A_2}$ are compatible}]
Assume for contradiction that there exists an $i \in \Dom{g_{t_p}} \intersect \Dom{g_{A_2}}$ such that $g_{t_p}(i) \neq g_{A_2}(i)$. 
We have that $g_{A_2}(i) = g_{t}(i)$. Therefore, $g_{t_p}(i) =
c(s(t_p, v_i)) \neq c(s(t, v_i)) = g_{t}(i)$. Then $s(t_p, v_i) \neq
s(t, v_i)$ must hold, which contradicts $\myR = \emptyset$. 
\end{proof}

\smallskip

\noindent
\textbf{In case that $\myR \neq \emptyset$:}

\begin{proof}[{Proof that $\Img{f} \subseteq \myB$ and $f$ is a transition for $g_{A_1}$}]
Let $\ell$ be an element of $\Img{f}$, and $i \in (\Dom{r_t}
\intersect \Dom{r_{t_p}})$ such that $s(t, v_i) \neq s(t_p, v_i)$ and
$f(i) = g_{t_p}(i) = \ell$. Now, assume for contradiction that $\ell
\not\in \myB$, i.e.\ $c(t, \ell) = c(t_p, \ell)$. This would mean that
$s(t, v_i) = s(t_p, v_i)$, because $s$ is a schedule. This clearly is
a contradiction. 
\end{proof}

\begin{proof}[{Proof that $\myBTilde \subseteq \Dom{b_{\tilde{A}_1}} \setminus \Img{g_{\tilde{A}_1}}$}]
It is clear that
$\myBTilde \subseteq (\Dom{b_{\kLI_t'}} \setminus \myBHat) \subseteq
\Dom{b_{\tilde{A}_1}}$. Assume for contradiction, that there exists an
$i \in \Dom{g_{\tilde{A}_1}}$ such that $g_{\tilde{A}_1}(i) = \ell \in
\myBTilde$. Since $\ell \not\in \Img{f}$ we have $g_{\tilde{A}_1}(i) =
g_{\kLI_t'}(i) = \ell$. Since $\ell \in \myBTilde \subseteq \myB$ we
have $c(t, \ell) \neq c(t_p, \ell)$, which means $s(t, v_i) \neq
s(t_p, v_i)$. Therefore, $i \in \Dom{f}$ and $f(i) = \ell$, i.e.\
$\ell \in \Img{f}$. This contradicts our assumption that $\ell \in
\myBTilde = \myB \setminus \Img{f}$. 
\end{proof}

\begin{proof}[{Proof that $g_{t_p}$ and $g_{A_2}$ are compatible}]
Assume for contradiction that there exists an $i \in \Dom{g_{t_p}} \intersect \Dom{g_{A_2}}$ such that $g_{t_p}(i) \neq g_{A_2}(i)$. 
If $i \in \Def{f}$, then $f(i) = g_{t_p}(i)$ and since $f \subseteq
g_{A_2}$ we have a contradiction. Therefore, let $i \not\in
\Def{f}$. Then we have that $g_{A_2}(i) = g_{t}(i)$. Therefore,
$g_{t_p}(i) = c(s(t_p, v_i)) \neq c(s(t, v_i)) = g_{t}(i)$. Then
$s(t_p, v_i) \neq s(t, v_i)$ must hold, which contradicts $i \not\in
\Def{f}$. 
\end{proof}
 
\medskip

\subsection{An example illustrating the construction in the proof of
  part \ref{item:IEHWkInGLIk} of Theorem~\ref{thm:IEHWkIsGLIk}}
\label{appendix:subsec:example} 

In this section we will define mappings \enquote{inline} like this: $f
\isdef \set{ a \to b, u \to v, x \to y }$ means $\Dom{f} = \set{ a, u,
  x}$ and $f(a) = b$, $f(u) = v$ and $f(x) = y$. 

We will draw a lot of guarded $2$-labeled incidence graphs in this
example. In the same way as in the proof, the red vertex $v_i$ is
either labeled with $i$ or it is not labeled at all. Therefore, we
colour $v_i$ \textcolor{cred}{red} to indicate, that $r(i) = v_i$ and
black to indicate that $i \not\in \Dom{r}$, i.e.\ the label $i$ is not
used. Since we only have two blue labels, we label a blue vertex with
$\textcolor{cblue}{1}$ or $\textcolor{cblue}{2}$ respectively to
indicate, that the corresponding label is mapped to that vertex. If a
blue vertex has no label, we give it an alphabetical name -- but we
colour it blue anyway to better distinguish it from red vertices. 

Edges between blue and red vertices are indicated by a dotted line
\tikz[inner sep=0pt, baseline=(base)]{ \node (base) at (0,-.5ex) {};
  \draw[edge] (0,0) -- (0.5,0); }. Since all guarded $2$-labeled
incidence graphs have real guards, a blue vertex with label $j$ is
only a guard of a red vertex $i$, if they are neighbours, therefore we
draw the edge dotted and thick \tikz[inner sep=0pt, baseline=(base)]{
  \node (base) at (0,-.5ex) {}; \draw[edge, guard] (0,0) -- (0.5,0); }
to indicate that the blue vertex is the red vertex' guard.  

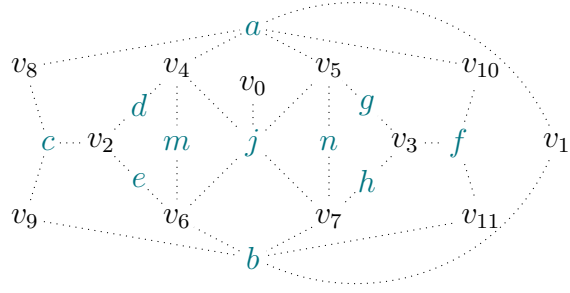
\begin{figure}
\centering

\begin{tikzpicture}[
	every node/.style={regular},
	]
	\node (S) at (4,0) {$v_1$};
	\node (Z) at (-2,0) {$v_2$};
	\node (ZP) at (2,0) {$v_3$};

	\node (X) at (-1,1) {$v_4$};
	\node (XP) at (1,1) {$v_5$};
	\node (Y) at (-1,-1) {$v_6$};
	\node (YP) at (1,-1) {$v_7$};

	\node (C) at (-3,1) {$v_8$};
	\node (CP) at (-3,-1) {$v_9$};
	\node (F) at (3,1) {$v_{10}$};
	\node (FP) at (3,-1) {$v_{11}$};

	\node[blue] (a) at ($(X)!0.5!(XP) + (0pt, 15pt)$) {$a$};
	\draw[dotted] (C) -- (a) -- (X);
	\draw[dotted] (XP) -- (a) -- (F);
	\draw[dotted] (S) edge[bend right=45] (a);

	\node[circle, blue] (b) at ($(Y)!0.5!(YP) + (0pt, -15pt)$) {$b$};
	\draw[dotted] (CP) -- (b) -- (Y);
	\draw[dotted] (YP) -- (b) -- (FP);
	\draw[dotted] (S) edge[bend left=45] (b);

	\node[blue] (c) at ($(C)!0.5!(CP)!0.3!(Z)$) {$c$};
	\draw[dotted] (C) -- (c) -- (Z);
	\draw[dotted] (CP) -- (c);

	\node[blue] (d) at ($(Z)!0.5!(X)$) {$d$};
	\draw[dotted] (Z) -- (d) -- (X);

	\node[blue] (e) at ($(Z)!0.5!(Y)$) {$e$};
	\draw[dotted] (Z) -- (e) -- (Y);

	\node[blue] (f) at ($(F)!0.5!(FP)!0.3!(ZP)$) {$f$};
	\draw[dotted] (F) -- (f) -- (ZP);
	\draw[dotted] (FP) -- (f);

	\node[blue] (g) at ($(ZP)!0.5!(XP)$) {$g$};
	\draw[dotted] (ZP) -- (g) -- (XP);

	\node[blue] (h) at ($(ZP)!0.5!(YP)$) {$h$};
	\draw[dotted] (ZP) -- (h) -- (YP);

	\node[blue] (m) at ($(X)!0.5!(Y)$) {$m$};
	\draw[dotted] (X) -- (m) -- (Y);

	\node[blue] (n) at ($(XP)!0.5!(YP)$) {$n$};
	\draw[dotted] (XP) -- (n) -- (YP);

	\node[blue] (j) at (0,0) {$j$};
	\node (J) at ($(a)!0.5!(j)$) {$v_0$};
	\draw[dotted] (X) -- (j) -- (XP);
	\draw[dotted] (Y) -- (j) -- (YP);
	\draw[dotted] (J) -- (j);
\end{tikzpicture}
 	\caption[]{An incidence graph $\inzidenz{}$ ($\inzidenz{}$ is a modified version of $\mathcal{Q}_2$ in \cite{Gottlob2003}).
	Blue nodes are depicted in blue and named $a,b,\ldots{}$, red nodes are depicted in black and named $v_1,v_2,\ldots{}$, and edges are indicated by dotted lines.}\label{fig:example_incidence_graph}
\end{figure}

\begin{figure}
	\centering
	\includegraphics[max width = \columnwidth]{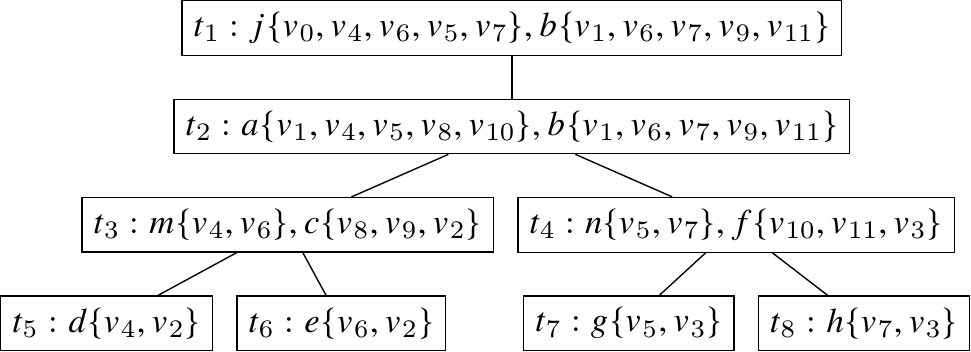}
	\caption[]{An ehd $D$ of the incidence graph $\I$ of
          Figure~\ref{fig:example_incidence_graph}, witnessing that
          $\ehw{\I} \leq 2$. The names of the tree-nodes are
          $t_1,t_2,\ldots{}$, the tree-node's bag is the set of all the
          depicted nodes $v_i$, and the tree-node's cover consists of the
          nodes depicted directly left to each symbol \enquote{$\{$}.}\label{fig:example_decomposition}
\end{figure}

We consider the particular incidence graph $\I$ and ehd $D=(T,\DBag,\DCover)$ depicted in
the Figures~\ref{fig:example_incidence_graph} and
\ref{fig:example_decomposition}, and we let $k\deff 2$.
Note that $D$ has the properties provided by Lemma~\ref{lemma:monotone_ehd}.
Defining a $2$-colouring $c$ is easy: we colour each blue node of $\I$
that is depicted first in its tree-node with the number 1, and all
other blue nodes are coloured with the number 2
(e.g., $c(j) = c(m) = 1$ and $c(b) = c(f) = 2$). In total, we get the following $2$-colouring $c$.
\begin{center}
\begin{tabular}{c|c c c c c c c c c c c}
$e \in \bV{\inzidenz{}}$ & $a$ & $b$ & $c$ & $d$ & $e$ & $f$ & $g$ & $h$ & $j$ & $m$ & $n$ \\\hline
$c(e)$ & $1$ & $2$ & $2$ & $1$ & $1$ & $2$ & $1$ & $1$ & $1$ & $1$ & $1$
\end{tabular}	
\end{center}

Our schedule is easy to define as well, since in many cases, the
neighbourhoods of the blue vertices covering the bags do not
intersect. For example, we have no other choice but to set $s(t_5,
v_2) = s(t_5, v_4) = d$ and $s(t_3, v_4) = m$, $s(t_3, v_9) = c$ and
so on. The complete schedule $s$ looks like this: 

\begin{center}
\begin{tabular}{c|c c c c c c c c}
$s$ 		& $t_1$	& $t_2$	& $t_3$	& $t_4$	& $t_5$	& $t_6$	& $t_7$	& $t_8$ 
\\\hline
$v_0$		& $j$	& --	& --	& --	& --	& --	& --	& --	\\
$v_1$		& $b$	& $b$	& --	& --	& --	& --	& --	& --	\\
$v_2$		& --	& --	& $c$	& --	& $d$	& $e$	& --	& --	\\
$v_3$		& --	& --	& --	& $f$	& --	& --	& $g$	& $h$	\\
$v_4$		& $j$	& $a$	& $m$	& $d$	& --	& --	& --	& --	\\
$v_5$		& $j$	& $a$	& --	& $n$	& --	& --	& $g$	& --	\\
$v_6$		& $j$	& $b$	& $m$	& --	& --	& $e$	& --	& --	\\
$v_7$		& $j$	& $b$	& --	& $n$	& --	& --	& --	& $h$	\\
$v_8$		& --	& $a$	& $c$	& --	& --	& --	& --	& --	\\
$v_9$		& $b$	& $b$	& $c$	& --	& --	& --	& --	& --	\\
$v_{10}$	& --	& $a$	& --	& $f$	& --	& --	& --	& --	\\
$v_{11}$	& $b$	& $b$	& --	& $f$	& --	& --	& --	& --	\\
\end{tabular}	
\end{center}

Now, to construct $\inzidenz{}$ we start at the bottom of $D$, so let
us take a look at the nodes $t_5$, $t_6$ and their parent $t_3$. Then,
$\kLI_{t_3}$, $\kLI_{t_5}$ and $\kLI_{t_6}$ look like this: 
\begin{center}
	\begin{tabularx}{\columnwidth}{C C C}
		$\kLI_{t_3}$:
	 	&
		$\kLI_{t_5}$:
		&
		$\kLI_{t_6}$:
		\\
		\adjustbox{valign=c}{

\begin{tikzpicture}[
	every node/.style={regular}
	]
	\node[red] (Z) at (-2,0) {$v_2$};

	\node[red] (X) at (-1,1) {$v_4$};
	\node[red] (Y) at (-1,-1) {$v_6$};

	\node[red] (C) at (-3,1) {$v_8$};
	\node[red] (CP) at (-3,-1) {$v_9$};

	\node[blue] (c) at ($(C)!0.5!(CP)!0.3!(Z)$) {$2$};
	\draw[guard, edge] (C) -- (c) -- (Z);
	\draw[guard, edge] (CP) -- (c);

	\node[blue] (m) at ($(X)!0.5!(Y)$) {$1$};
	\draw[guard, edge] (X) -- (m) -- (Y);
\end{tikzpicture}
 		}
		&
		\adjustbox{valign=c}{

\begin{tikzpicture}[
	every node/.style={regular},
	]
	\node[red] (Z) at (-2,0) {$v_2$};

	\node[red] (X) at (-1,1) {$v_4$};

	\node[blue] (d) at ($(Z)!0.5!(X)$) {$1$};
	\draw[edge, guard] (Z) -- (d) -- (X);
\end{tikzpicture}
 		}
		&
		\adjustbox{valign=c}{

\begin{tikzpicture}[
	every node/.style={regular},
	]
	\node[red] (Z) at (-2,0) {$v_2$};
	\node[red] (Y) at (-1,-1) {$v_6$};

	\node[blue] (e) at ($(Z)!0.5!(Y)$) {$1$};
	\draw[edge, guard] (Z) -- (e) -- (Y);
\end{tikzpicture}
 		}
	\end{tabularx}
\end{center}

We want to glue $\kLI_{t_5}$ and $\kLI_{t_6}$ with $\kLI_{t_3}$, since
$t_3$ is the parent of $t_5$ and $t_6$. If we approach this in a naive
way, we would immediately fail, since we would merge the blue vertices
with label $1$, but the $1$-labeled vertex in e.g.\ $\kLI_{t_5}$
should represent $d$, whereas the one in $t_3$ should correspond to
$m$. Hence, simply glueing these would result in the wrong incidence
graph. 

Here come the transitions into play. We want to apply a transition
$f_5$ to $\kLI_{t_5}$ and $f_6$ to $\kLI_{t_6}$ respectively, so that
we mitigate the aforementioned issue. If we want to remove the label
$1$ and insert a new one with this label, we have to find new guards
for $v_4$ and $v_2$. The schedule $s$ helps us with this. It tells us,
that $2$ is a suitable guard for $v_2$ and $1$ is suitable for $v_4$
since $c(s(t_3, v_2)) = c(c) = 2$ and $c(s(t_3, v_4)) = c(m) = 1$. We
use the guards according to $t_3$, because then we can be sure that we
do not make the same mistake as outlined above when glueing. Hence, we
set $f_5 \isdef \set{  2 \to 2, 4 \to 1 }$. With the same argument we
get $f_6 \isdef \set{ 2 \to 2, 6 \to 1 }$. Applying these transitions
gives us the following: 
\begin{center}
	\begin{tabularx}{\columnwidth}{C C}
		$A_{t_5} \isdef \switch{\kLI_{t_5}}{f_5}$:\smallskip
		&
		$A_{t_6} \isdef \switch{\kLI_{t_6}}{f_6}$:
		\\
		\adjustbox{valign=c}{

\begin{tikzpicture}[
	every node/.style={regular}
	]
	\node[red] (Z) at (-2,0) {$v_2$};

	\node[red] (X) at (-1,1) {$v_4$};

	\node[blue] (c) at (-2.7,0) {$2$};
	\draw[edge, guard] (c) -- (Z);

	\node[blue] (d) at ($(Z)!0.5!(X)$) {$d$};
	\draw[edge] (Z) -- (d) -- (X);

	\node[blue] (m) at (-1,0) {$1$};
	\draw[edge, guard] (X) -- (m);
\end{tikzpicture}
 		}
		&
		\adjustbox{valign=c}{

\begin{tikzpicture}[
	every node/.style={regular},
	]
	\node[red] (Z) at (-2,0) {$v_2$};

	\node[red] (Y) at (-1,-1) {$v_6$};
	\node[blue] (c) at (-2.7,0) {$2$};
	\draw[edge, guard] (c) -- (Z);

	\node[blue] (e) at ($(Z)!0.5!(Y)$) {$e$};
	\draw[edge] (Z) -- (e) -- (Y);

	\node[blue] (m) at (-1,0) {$1$};
	\draw[edge, guard] (m) -- (Y);
\end{tikzpicture}
 		}
	\end{tabularx}
\end{center}
If we glue these with $\kLI_{t_3}$, we see that nothing bad happens any more:
\begin{center}
	$\kLI_{t_3}' \isdef \glue{\glue{\kLI_{n_3}}{A_{t_5}}}{A_{t_6}}$:\\

\begin{tikzpicture}[
	every node/.style={circle},
	]
	\node[red] (Z) at (-2,0) {$v_2$};

	\node[red] (X) at (-1,1) {$v_4$};
	\node[red] (Y) at (-1,-1) {$v_6$};

	\node[red] (C) at (-3,1) {$v_8$};
	\node[red] (CP) at (-3,-1) {$v_9$};
	\node[blue] (c) at ($(C)!0.5!(CP)!0.3!(Z)$) {$2$};
	\draw[edge, guard] (C) -- (c) -- (Z);
	\draw[edge, guard] (CP) -- (c);

	\node[blue] (d) at ($(Z)!0.5!(X)$) {$d$};
	\draw[edge] (Z) -- (d) -- (X);

	\node[blue] (e) at ($(Z)!0.5!(Y)$) {$e$};
	\draw[edge] (Z) -- (e) -- (Y);

	\node[blue] (m) at ($(X)!0.5!(Y)$) {$1$};
	\draw[edge, guard] (X) -- (m) -- (Y);
\end{tikzpicture}
 \end{center}
We now have constructed a guarded $2$-labeled incidence graph
$\kLI_{t_3}'$, that is isomorphic to $\assig{t_3}$ -- the subgraph
induced by $T_{t_3}$. 

Next we apply the same procedure to the other branch of $D$, i.e.\ to $t_7$, $t_8$ and their parent $t_4$:
\begin{center}
	\begin{tabularx}{\columnwidth}{C C C}
		$\kLI_{t_4}$:
		&
		$\kLI_{t_7}$:
		&
		$\kLI_{t_8}$:
		\\
		\adjustbox{valign=c}{

\begin{tikzpicture}[
	every node/.style={regular},
	]
	\node[red] (ZP) at (2,0) {$v_3$};

	\node[red] (XP) at (1,1) {$v_5$};
	\node[red] (YP) at (1,-1) {$v_7$};

	\node[red] (F) at (3,1) {$v_{10}$};
	\node[red] (FP) at (3,-1) {$v_{11}$};

	\node[blue] (f) at ($(F)!0.5!(FP)!0.3!(ZP)$) {$2$};
	\draw[edge, guard] (F) -- (f) -- (ZP);
	\draw[edge, guard] (FP) -- (f);

	\node[blue] (n) at ($(XP)!0.5!(YP)$) {$1$};
	\draw[edge, guard] (XP) -- (n) -- (YP);
\end{tikzpicture}
 		}
		&
		\adjustbox{valign=c}{

\begin{tikzpicture}[
	every node/.style={regular},
	]
	\node[red] (ZP) at (2,0) {$v_3$};
	\node[red] (XP) at (1,1) {$v_5$};
	
	\node[blue] (g) at ($(ZP)!0.5!(XP)$) {$1$};
	\draw[edge, guard] (ZP) -- (g) -- (XP);
\end{tikzpicture}
 		}
		&
		\adjustbox{valign=c}{

\begin{tikzpicture}[
	every node/.style={circle},
	]
	\node[red] (ZP) at (2,0) {$v_3$};
	\node[red] (YP) at (1,-1) {$v_7$};

	\node[blue] (h) at ($(ZP)!0.5!(YP)$) {$1$};
	\draw[edge, guard] (ZP) -- (h) -- (YP);
\end{tikzpicture}
 		}
	\end{tabularx}
\end{center}
We set $f_7 \isdef \set{ 3 \to 2, 5 \to 1 }$ and $f_8 \isdef \set{ 3 \to 2, 7 \to 1 }$, applying these yields:
\begin{center}
	\begin{tabularx}{\columnwidth}{C C}
		$A_{t_7} \isdef \switch{\kLI_{t_7}}{f_7}$:\smallskip
		&
		$A_{t_8} \isdef \switch{\kLI_{t_8}}{f_8}$:
		\\
		\adjustbox{valign=c}{

\begin{tikzpicture}[
	every node/.style={regular},
	]
	\node[red] (ZP) at (2,0) {$v_3$};
	\node[red] (XP) at (1,1) {$v_5$};

	\node[blue] (f) at (2.7,0) {$2$};
	\draw[edge, guard] (f) -- (ZP);

	\node[blue] (g) at ($(ZP)!0.5!(XP)$) {$g$};
	\draw[edge] (ZP) -- (g) -- (XP);

	\node[blue] (n) at (1,0) {$1$};
	\draw[edge, guard] (XP) -- (n);
\end{tikzpicture}
 		}
		&
		\adjustbox{valign=c}{

\begin{tikzpicture}[
	every node/.style={circle},
	]
	\node[red] (ZP) at (2,0) {$v_3$};

	\node[red] (YP) at (1,-1) {$v_7$};
	\node[blue] (f) at (2.7,0) {$2$};
	\draw[edge, guard] (f) -- (ZP);

	\node[blue] (h) at ($(ZP)!0.5!(YP)$) {$h$};
	\draw[edge] (ZP) -- (h) -- (YP);

	\node[blue] (n) at (1,0) {$1$};
	\draw[edge, guard] (n) -- (YP);
\end{tikzpicture}
 		}
	\end{tabularx}
\end{center}
And then glueing gives us again a guarded $2$-labeled incidence graph $\kLI_{t_4}'$ isomorphic to the subgraph $\assig{t_4}$ induced by $T_{t_4}$.
\begin{center}
	$\kLI_{t_4}' \isdef \glue{\glue{\kLI_{t_4}}{A_{t_7}}}{A_{t_8}}$:\\

\begin{tikzpicture}[
	every node/.style={regular},
	]
	\node[red] (ZP) at (2,0) {$v_3$};

	\node[red] (XP) at (1,1) {$v_5$};
	\node[red] (YP) at (1,-1) {$v_7$};

	\node[red] (F) at (3,1) {$v_{10}$};
	\node[red] (FP) at (3,-1) {$v_{11}$};

	\node[blue] (f) at ($(F)!0.5!(FP)!0.3!(ZP)$) {$2$};
	\draw[edge, guard] (F) -- (f) -- (ZP);
	\draw[edge, guard] (FP) -- (f);

	\node[blue] (g) at ($(ZP)!0.5!(XP)$) {$g$};
	\draw[edge] (ZP) -- (g) -- (XP);

	\node[blue] (h) at ($(ZP)!0.5!(YP)$) {$h$};
	\draw[edge] (ZP) -- (h) -- (YP);

	\node[blue] (n) at ($(XP)!0.5!(YP)$) {$1$};
	\draw[edge, guard] (XP) -- (n) -- (YP);
\end{tikzpicture}
 \end{center}

In the next step we want to handle $t_3$, $t_4$ and their parent
$t_2$. Obviously, this time we do not glue $\kLI_{t_3}$, $\kLI_{t_4}$
with $\kLI_{t_2}$. Instead, we glue the incidence graphs we just
constructed, that is $\kLI_{t_3}'$ and $\kLI_{t_4}'$. 

Again, we first have to apply transitions to make them compatible to
$\kLI_{t_2}$. For $t_3$, we want $v_4$, $v_8$ to be guarded by a new
vertex labeled $1$ and $v_6, v_9$ by one labeled $2$. So we set $f_3
\isdef \set{ 4 \to 1, 6 \to 2, 8 \to 1, 9 \to 2 }$ and analogously
$f_4 \isdef \set{ 5 \to 1, 7 \to 2, 10 \to 1, 11 \to 2 }$. But, what
do we do with $v_2$ and $v_3$? Since they no longer appear in the bag
of $t_2$, we know that they will no longer appear in the construction
\emph{at all} because of the connectedness of vertices, so there is no
need to transition them to new guards. But, because we do not include
them, $f_3$ is not a transition for $\kLI_{t_3}'$ (and $f_4$ not for
$\kLI_{t_4}'$), since $g_{\kLI_{t_3}}(2) = 2 \in \Img{f_3}$ but $2
\not\in \Def{f_3}$, thus violating the requirements of a
transition. Analogously $g_{t_4}(3) = 2 \in \Img{f_4}$, but $3 \not\in
\Def{f_4}$. But, as we have seen a moment ago, we can easily remedy
this situation for $\kLI_{t_3}'$ by removing the label $2$ from the
red vertex $v_2$, by setting $X_{r,t_3} \isdef \set{ 2 }$. 
\begin{center}
$A_{t_3} \isdef \switch{\reclaimR{\kLI_{t_3}'}{X_{r,t_3}}}{f_3}$:
\\

\begin{tikzpicture}[
	every node/.style={regular},
	]
	\node (Z) at (-2,0) {$v_2$};

	\node[red] (X) at (-1,1) {$v_4$};
	\node (XP) at (1,1) {$\phantom{v_5}$};
	\node[red] (Y) at (-1,-1) {$v_6$};
	\node (YP) at (1,-1) {$\phantom{v_7}$};

	\node[red] (C) at (-3,1) {$v_8$};
	\node[red] (CP) at (-3,-1) {$v_9$};

	\node[blue] (a) at ($(X)!0.5!(XP) + (0pt, 15pt)$) {$1$};
	\draw[edge, guard] (C) -- (a) -- (X);

	\node[blue] (b) at ($(Y)!0.5!(YP) + (0pt, -15pt)$) {$2$};
	\draw[edge, guard] (CP) -- (b) -- (Y);

	\node[blue] (c) at ($(C)!0.5!(CP)!0.3!(Z)$) {$c$};
	\draw[edge] (C) -- (c) -- (Z);
	\draw[edge] (CP) -- (c);

	\node[blue] (d) at ($(Z)!0.5!(X)$) {$d$};
	\draw[edge] (Z) -- (d) -- (X);

	\node[blue] (e) at ($(Z)!0.5!(Y)$) {$e$};
	\draw[edge] (Z) -- (e) -- (Y);

	\node[blue] (m) at ($(X)!0.5!(Y)$) {$m$};
	\draw[edge] (X) -- (m) -- (Y);
\end{tikzpicture}
 \end{center}
In the same way, removing the label $3$ from $v_3$ in $\kLI_{t_4}'$
fixes our violation of the requirements, turning $f_4$ into a proper
transition. I.e.\ we set $X_{r,t_4} \isdef \set{ 3 }$. 
\begin{center}
$A_{t_4} \isdef \switch{\reclaimR{\kLI_{t_4}'}{X_{r,t_4}}}{f_4}$:
\\

\begin{tikzpicture}[
	every node/.style={regular},
	]
	\node (ZP) at (2,0) {$v_3$};

	\node (X) at (-1,1) {$\phantom{v_4}$};
	\node[red] (XP) at (1,1) {$v_5$};
	\node (Y) at (-1,-1) {$\phantom{v_6}$};
	\node[red] (YP) at (1,-1) {$v_7$};

	\node[red] (F) at (3,1) {$v_{10}$};
	\node[red] (FP) at (3,-1) {$v_{11}$};

	\node[blue] (a) at ($(X)!0.5!(XP) + (0pt, 15pt)$) {$1$};
	\draw[guard, edge] (XP) -- (a) -- (F);

	\node[blue] (b) at ($(Y)!0.5!(YP) + (0pt, -15pt)$) {$2$};
	\draw[guard, edge] (YP) -- (b) -- (FP);

	\node[blue] (f) at ($(F)!0.5!(FP)!0.3!(ZP)$) {$f$};
	\draw[edge] (F) -- (f) -- (ZP);
	\draw[edge] (FP) -- (f);

	\node[blue] (g) at ($(ZP)!0.5!(XP)$) {$g$};
	\draw[edge] (ZP) -- (g) -- (XP);

	\node[blue] (h) at ($(ZP)!0.5!(YP)$) {$h$};
	\draw[edge] (ZP) -- (h) -- (YP);

	\node[blue] (n) at ($(XP)!0.5!(YP)$) {$n$};
	\draw[edge] (XP) -- (n) -- (YP);
\end{tikzpicture}
 \end{center}

Glueing these with $\kLI_{t_2}$ gives us the following:
\begin{center}
$\kLI_{t_2}' \isdef \glue{\glue{\kLI_{t_2}}{A_{t_3}}}{A_{t_4}}$:
\\

\begin{tikzpicture}[
	every node/.style={regular},
	]
	\node[red] (S) at (4,0) {$v_1$};
	\node (Z) at (-2,0) {$v_2$};
	\node (ZP) at (2,0) {$v_3$};

	\node[red] (X) at (-1,1) {$v_4$};
	\node[red] (XP) at (1,1) {$v_5$};
	\node[red] (Y) at (-1,-1) {$v_6$};
	\node[red] (YP) at (1,-1) {$v_7$};

	\node[red] (C) at (-3,1) {$v_8$};
	\node[red] (CP) at (-3,-1) {$v_9$};
	\node[red] (F) at (3,1) {$v_{10}$};
	\node[red] (FP) at (3,-1) {$v_{11}$};

	\node[blue] (a) at ($(X)!0.5!(XP) + (0pt, 15pt)$) {$1$};
	\draw[guard, edge] (C) -- (a) -- (X);
	\draw[guard, edge] (XP) -- (a) -- (F);
	\draw[edge] (S) edge[bend right=45] (a);

	\node[blue] (b) at ($(Y)!0.5!(YP) + (0pt, -15pt)$) {$2$};
	\draw[guard, edge] (CP) -- (b) -- (Y);
	\draw[guard, edge] (YP) -- (b) -- (FP);
	\draw[edge, guard] (S) edge[bend left=45] (b);

	\node[blue] (c) at ($(C)!0.5!(CP)!0.3!(Z)$) {$c$};
	\draw[edge] (C) -- (c) -- (Z);
	\draw[edge] (CP) -- (c);

	\node[blue] (d) at ($(Z)!0.5!(X)$) {$d$};
	\draw[edge] (Z) -- (d) -- (X);

	\node[blue] (e) at ($(Z)!0.5!(Y)$) {$e$};
	\draw[edge] (Z) -- (e) -- (Y);

	\node[blue] (f) at ($(F)!0.5!(FP)!0.3!(ZP)$) {$f$};
	\draw[edge] (F) -- (f) -- (ZP);
	\draw[edge] (FP) -- (f);

	\node[blue] (g) at ($(ZP)!0.5!(XP)$) {$g$};
	\draw[edge] (ZP) -- (g) -- (XP);

	\node[blue] (h) at ($(ZP)!0.5!(YP)$) {$h$};
	\draw[edge] (ZP) -- (h) -- (YP);

	\node[blue] (m) at ($(X)!0.5!(Y)$) {$m$};
	\draw[edge] (X) -- (m) -- (Y);

	\node[blue] (n) at ($(XP)!0.5!(YP)$) {$n$};
	\draw[edge] (XP) -- (n) -- (YP);
\end{tikzpicture}
 \end{center}

Now we have to repeat this procedure one last time. We want $v_4, v_5,
v_6$ and $v_7$ to be guarded by a new vertex labeled $1$. For that, we
remove the labels from $v_8$ and $v_{10}$. Thus, we define $f_2 \isdef
\set{ 4 \to 1, 5 \to 1, 6 \to 1, 7 \to 1 }$ and $X_{r,t_2} \isdef
\set{ 8, 10 }$. Glueing the resulting $2$-guarded incidence graph
$A_{t_2} \isdef \switch{(\reclaimR{\kLI_{t_2}'}{X_{r,t_2}})}{f_2}$
with $\kLI_{t_1}$ we get: 
\begin{center}
$\kLI_{t_1}' \isdef \glue{\kLI_{t_1}}{A_{t_2}}$:
\\

\begin{tikzpicture}[
	every node/.style={regular},
	]
	\node[red] (S) at (4,0) {$v_1$};
	\node (Z) at (-2,0) {$v_2$};
	\node (ZP) at (2,0) {$v_3$};

	\node[red] (X) at (-1,1) {$v_4$};
	\node[red] (XP) at (1,1) {$v_5$};
	\node[red] (Y) at (-1,-1) {$v_6$};
	\node[red] (YP) at (1,-1) {$v_7$};

	\node (C) at (-3,1) {$v_8$};
	\node[red] (CP) at (-3,-1) {$v_9$};
	\node (F) at (3,1) {$v_{10}$};
	\node[red] (FP) at (3,-1) {$v_{11}$};

	\node[blue] (a) at ($(X)!0.5!(XP) + (0pt, 15pt)$) {$a$};
	\draw[edge] (C) -- (a) -- (X);
	\draw[edge] (XP) -- (a) -- (F);
	\draw[edge] (S) edge[bend right=45] (a);

	\node[blue] (b) at ($(Y)!0.5!(YP) + (0pt, -15pt)$) {$2$};
	\draw[edge, guard] (CP) -- (b) -- (FP);
	\draw[edge, guard] (S) edge[bend left=45] (b);
	\draw[edge] (YP) -- (b) -- (Y);

	\node[blue] (c) at ($(C)!0.5!(CP)!0.3!(Z)$) {$c$};
	\draw[edge] (C) -- (c) -- (Z);
	\draw[edge] (CP) -- (c);

	\node[blue] (d) at ($(Z)!0.5!(X)$) {$d$};
	\draw[edge] (Z) -- (d) -- (X);

	\node[blue] (e) at ($(Z)!0.5!(Y)$) {$e$};
	\draw[edge] (Z) -- (e) -- (Y);

	\node[blue] (f) at ($(F)!0.5!(FP)!0.3!(ZP)$) {$f$};
	\draw[edge] (F) -- (f) -- (ZP);
	\draw[edge] (FP) -- (f);

	\node[blue] (g) at ($(ZP)!0.5!(XP)$) {$g$};
	\draw[edge] (ZP) -- (g) -- (XP);

	\node[blue] (h) at ($(ZP)!0.5!(YP)$) {$h$};
	\draw[edge] (ZP) -- (h) -- (YP);

	\node[blue] (m) at ($(X)!0.5!(Y)$) {$m$};
	\draw[edge] (X) -- (m) -- (Y);

	\node[blue] (n) at ($(XP)!0.5!(YP)$) {$n$};
	\draw[edge] (XP) -- (n) -- (YP);

	\node[blue] (j) at (0,0) {$1$};
	\node[red] (J) at ($(a)!0.5!(j)$) {$v_0$};
	\draw[guard, edge] (X) -- (j) -- (XP);
	\draw[guard, edge] (Y) -- (j) -- (YP);
	\draw[guard, edge] (J) -- (j);
\end{tikzpicture}
 \end{center}
Which is obviously isomorphic to the incidence graph $\inzidenz{E}$ depicted in Figure~\ref{fig:example_incidence_graph}.

Note that we did not remove blue labels in this construction,
which is possible (and necessary) by the definition of $k$-labeled
incidence graphs. We need to do this in cases where we do not
have labeled red vertices in the bags connecting the blue ones, so
that we can apply a transition. Imagine for example the case that
neighbouring bags are disjoint. We can handle this by removing all
red labels and then all blue labels from one of these incidence
graphs. On the other hand, we do need transitions, as we can easily
verify that we can not construct the incidence graph $\inzidenz{E}$
by removing blue labels instead of applying transitions, without
\enquote{temporarily leaving labeled red vertices unguarded}, which
must be prohibited. Because if we were to allow this, then it would be
easy to construct $k$-guarded incidence graphs with entangled
hypertree width larger than $k$. 
 
\medskip

\subsection{Proof of Lemma \ref{lem:hom_of_glued_is_product}}
\label{appendix:subsec:lemma:hom_of_glued_is_product}

We first prove the lemma on
\enquote{ordinary}
$k$-labeled incidence graphs.

\begin{enumerate}[wide]
	\item[\textbf{Claim~1:}] 
	Let $\kLI$, $\kLI_1$, $\kLI_2$ be $k$-labeled incidence
        graphs. It holds that $\hom{\glue{\kLI_1}{\kLI_2}}{\kLI} =
        \hom{\kLI_1}{\kLI} \cdot \hom{\kLI_2}{\kLI}$. 
	\begin{proof}
          We
          provide
          a bijection $f$ from $\Hom{\glue{\kLI_1}{\kLI_2}}{\kLI}$ to
          $\Hom{\kLI_1}{\kLI} \times \Hom{\kLI_2}{\kLI}$. \\ For $(h_R, h_B) \in \Hom{\glue{\kLI_1}{\kLI_2}}{\kLI}$ we define 
		\[
			f((h_R, h_B)) \isdef ((h_{1,R}, h_{1,B}), (h_{2,R}, h_{2,B}))
		\]
		where
		\begin{equation*}
			h_{i,R}(v) = h_R(\pi_{i,R}(v)) \; \text{ and } \; h_{i,B}(e) = h_B(\pi_{i,B}(e))
		\end{equation*}
		for every $i \in [2]$ and all $v \in \rV{\I_{\kLI_i}}$ and all $e \in \bV{\I_{\kLI_i}}$.
		We have to prove that
		\begin{itemize}
			\item $f$ is well-defined, i.e.\ $(h_{1,R}, h_{1,B})$ and $(h_{2,R}, h_{2,B})$ are homomorphisms;
			\item $f$ is injective;
			\item and $f$ is surjective.
		\end{itemize}
		
		We prove that $(h_{1,R}, h_{1,B})$ is a
                homomorphism. In an analogous
                way, one can prove that $(h_{2,R}, h_{2,B})$ is homomorphism as well.
		
		Let $(e,v) \in \E{\I_{\kLI_1}}$. We must show that
                this edge exists in the image as well, i.e.\
                $(h_{1,B}(e),\allowbreak h_{1,R}(v)) \in
                \E{\I_{\kLI}}$. We have $h_{1,R}(v) =
                h_R(\pi_{1,R}(v)) = h_R([(v,1)]_{\sim_R})$ and
                $h_{1,B}(e) = h_B(\pi_{1,B}(e)) =
                h_B([(e,1)]_{\sim_B})$. By definition, we have
                $([(e,1)]_{\sim_B}, [(v,1)]_{\sim_R}) \in
                \E{\I_{\glue{\kLI_1}{\kLI_2}}}$,
                therefore $(h_B([(e,1)]_{\sim_B}), h_R([(v,1)]_{\sim_R})) \in \E{\I_{\kLI}}$, since $(h_R, h_B)$ is a homomorphism.

		Let $i \in \Dom{r_{\kLI_1}}$ and $j \in
                \Dom{b_{\kLI_1}}$. We must show, that
                $h_{1,R}(r_{\kLI_1}(i)) = r_{\kLI}(i)$ and
                $h_{1,B}(b_{\kLI_1}(j)) = b_{\kLI}(j)$. This is easy
                as well, since  
		\begin{align*}
			h_{1,R}(r_{\kLI_1}(i)) &= h_R(\pi_{1,R}(r_{\kLI_1}(i))) \\
			&= h_R([(1, r_{\kLI_1(i)})]_{\sim_R}) \\
			&= h_R(r_{\glue{\kLI_1}{\kLI_2}}(i)) 
			= r_{\kLI}(i) \;,
		\end{align*} 
		and
		\begin{align*}
			h_{1,B}(b_{\kLI_1}(j)) &= h_B(\pi_{1,B}(b_{\kLI_1}(j))) \\
			&= h_B([(1, b_{\kLI_1(j)})]_{\sim_B}) \\
			&= h_B(b_{\glue{\kLI_1}{\kLI_2}}(j)) 
			= b_{\kLI}(j) \;.
		\end{align*}
		Thus, $(h_{1,R}, h_{1,B})$ and therefore, $f$ is well-defined.
		
		It is easy to see, that $f$ is injective: Assume that
                there exists a homomorphism $(h_R', h_B') \in
                \Hom{\glue{\kLI_1}{\kLI_2}}{\kLI}$ such that $f(h_R',
                h_B') = f(h_R, h_B)$. Then we have for every $i \in
                [2]$ that $h_R([(v,i)]_{\sim_R}) =
                h_R'([(v,i)]_{\sim_R})$ holds for all $v \in
                \rV{\I_{\kLI_i}}$ and $h_B([(e,i)]_{\sim_B}) =
                h_B'([(e,i)]_{\sim_B})$ holds for all $e \in
                \bV{\I_{\kLI_i}}$. Since these
                make up
                all vertices, $(h_R', h_B')$ and $(h_R, h_B)$ agree on all values, therefore they are equal.
		
		Consider $(h_{1,R}, h_{1,B}) \in \Hom{\kLI_1}{\kLI}$
                and $(h_{2,R}, h_{2,B}) \in \Hom{\kLI_2}{\kLI}$. For
                every $i \in [2]$ let $h_R([(v,i)]_{\sim_R}) =
                h_{i,R}(v)$ for all $v \in \rV{\I_{\kLI_i}}$ and
                $h_B([(e,i)]_{\sim_B}) = h_{i,B}(e)$ for all $e \in
                \bV{\I_{\kLI_i}}$.  
		Assume that $h_R$ is ill-defined. Then there must be
                $v_1 \in \rV{\I_{\kLI_1}}$ and $v_2 \in
                \rV{\I_{\kLI_2}}$ such that $v_1 \in [(v_2,
                2)]_{\sim_R}$ and $h_{1,R}(v_1) \neq
                h_{2,R}(v_2)$. $v_1 \in [(v_2, 2)]_{\sim_R}$ means
                that there exists a $j \in \Dom{r_{\kLI_1}} \intersect
                \Dom{r_{\kLI_2}}$ such that $r_{\kLI_1}(j) = v_1$ and
                $r_{\kLI_2}(j) = v_2$. But this means $h_{1,R}(v_1) =
                r_{\kLI}(j) = h_{2,R}(v_2)$, which is a
                contradiction. Thus, $h_R$ is well-defined. The same
                argument works for $h_B$. With similar arguments as we
                have used so far, it is easy to verify that $(h_R,
                h_B) \in \Hom{\glue{\kLI_1}{\kLI_2}}{\kLI}$. 

		We showed that $f$ is well-defined, injective and
                surjective, thus bijective. Therefore,
                $|\Hom{\glue{\kLI_1}{\kLI_2}}{\kLI}| =
                |\Hom{\kLI_1}{\kLI}| \cdot |\Hom{\kLI_2}{\kLI}|$. 
	\end{proof}
	
	\item[\textbf{Claim~2:}] Let $\kLI$, $\kLI'$ be $k$-labeled incidence graphs and
          let $\myR \isdef \set{i_1, \dots, i_\ell} \subseteq
          \Dom{r_{\kLI}}$ with $i_1 < \cdots < i_\ell$.
          It holds that 
	\[
		\hom{\reclaimR{\kLI}{\myR}}{\kLI'} = \sum_{\tupel{v} \in \rV{\I_{\kLI'}}^\ell} \hom{\kLI}{\reseatR{\kLI'}{\myR}{\tupel{v}}} \;.
	\]
	\begin{proof}
		Let $(h_R, h_B) \in
                \Hom{\reclaimR{\kLI}{\myR}}{\kLI'}$ be a
                homomorphism.
                Recall
                that $\I_{\kLI} = \I_{\reclaimR{\kLI}{\myR}}$,
                therefore we can define the following tuple of red
                vertices in $\kLI'$:
		\[
			 \tupel{v} = (v_{i_1}, \dots, v_{i_\ell}) \isdef (h_R(r_{\kLI}(i_1)), \dots, h_R(r_{\kLI}(i_\ell)))\;.
		\]
		Clearly, $(h_R, h_B) \in
                \Hom{\kLI}{\reseatR{\kLI'}{\myR}{\tupel{v}}}$, since
                it is a homomorphism from $\I_{\kLI} =
                \I_{\reclaimR{\kLI}{\myR}}$ to
                $\I_{\reseatR{\kLI'}{\myR}{\tupel{v}}} = \I_{\kLI'}$,
                and we moved the red labels onto $(v_{i_1}, \dots,
                v_{i_\ell})$. 

		For any tuple $\tupel{v} \in \rV{\I_{\kLI'}}^\ell$ it is easy to see that 
		\[
			\Hom{\kLI}{\reseatR{\kLI'}{\myR}{\tupel{v}}} \
                        \subseteq \ \Hom{\reclaimR{\kLI}{\myR}}{\kLI'} \;.
		\]
		It is also easy to see, that for all $\tupel{v} \neq \tupel{v}' \in \rV{\I_{\kLI'}}^\ell$
		\[
			\Hom{\kLI}{\reseatR{\kLI'}{\myR}{\tupel{v}}} \
                        \intersect\
                        \Hom{\kLI}{\reseatR{\kLI'}{\myR}{\tupel{v}'}}
                        \ = \ \emptyset \; .
		\]

		This finishes our proof.
	\end{proof}
	
	\item[\textbf{Claim~3:}] Let $\kLI$, $\kLI'$ be $k$-labeled incidence graphs and let $\myB \subseteq \Dom{b_{\kLI}}$ and $\ell \isdef |\myB|$. It holds that 
	\[
		\hom{\reclaimB{\kLI}{\myB}}{\kLI'} = \sum_{\tupel{e} \in \bV{\I_{\kLI'}}^\ell} \hom{\kLI}{\reseatB{\kLI'}{\myB}{\tupel{e}}} \;.
	\]
	\begin{proof}
		The same argument as in the previous case works for this claim as well.
	\end{proof}
	
	\item[\textbf{Claim~4:}] Let $\kLI$, $\kLI'$ be $k$-labeled incidence graphs and
          let $f$ be a transition for $\kLI$. Let $\myB \isdef
          \Dom{b_{\kLI}} \intersect \Img{f} \intersect \Img{g}$. It
          holds that $\hom{\switch{\kLI}{f}}{\kLI'} =$  
	\[
		\hom{\kLIf}{\kLI'} \ \cdot \!\!\! \sum_{\tupel{e} \in
                  \bV{\I_{\kLI'}} ^\ell} \!\!\!
                \hom{\kLI}{\reseatB{\kLI'}{\myB}{\tupel{e}}} \;. 
	\]
	\begin{proof}
          This follows directly from 1) and 3) and
          Definition~\ref{def:transition}.
	\end{proof}
\end{enumerate}

Generalising these results to quantum $k$-labeled incidence graphs is straightforward using simple linear arguments:
\begin{enumerate}[wide]
	\item[\textbf{Part~1:}] \ $\hom{\glue{Q_1}{Q_2}}{\kLI} =$
	\begin{eqnarray*}
		& \phantom{=} & \sum_{i \in [d] \atop j \in [d']} (\alpha_i \cdot \alpha_j')\hom{\glue{\kLI_i}{\kLI_j'}}{\kLI} \\
		& = & \sum_{i \in [d] \atop j \in [d']} (\alpha_i \cdot \alpha_j') \cdot \hom{\kLI_i}{\kLI} \cdot \hom{\kLI'_j}{\kLI} \\
		& = & \sum_{i \in [d]} \Big( \alpha_i \hom{\kLI_i}{\kLI} \cdot \sum_{j \in [d']} \alpha_j' \hom{\kLI_j'}{\kLI}  \Big) \\
		& = & \sum_{i \in [d]} \Big( \alpha_i \hom{\kLI_i}{\kLI} \cdot \hom{Q_2}{\kLI} \Big) \\
		& = & \hom{Q_2}{\kLI} \cdot \sum_{i \in [d]} \alpha_i \hom{\kLI_i}{\kLI} \\
		& = & \hom{Q_2}{\kLI} \cdot \hom{Q_1}{\kLI}
	\end{eqnarray*}

	\item[\textbf{Part~2:}] \ $\hom{\reclaimR{Q}{\myR}}{\kLI} =$
	\begin{eqnarray*}
		& \phantom{=} & \sum_{i \in [d]} \alpha_i \hom{\reclaimR{\kLI_i}{\kLI}}{\kLI} \\
		& = & \sum_{i \in [d]} \alpha_i \; \Big( \!\!\!\!\!
                      \sum_{\tupel{v} \in \rV{\I_{\kLI}}^\ell}
                      \!\!\!\!
                      \hom{\kLI_i}{\reseatR{\kLI}{\myR}{\tupel{v}}}
                      \Big) \\ 
		& = & \!\!\!\!\! \sum_{\tupel{v} \in
                      \rV{\I_{\kLI}}^\ell} \ \Big( \sum_{i \in [d]}
                      \alpha_i
                      \hom{\kLI_i}{\reseatR{\kLI}{\myR}{\tupel{v}}}
                      \Big) \\ 
		& = & \!\!\!\!\! \sum_{\tupel{v} \in \rV{\I_{\kLI}}^\ell} \ \hom{Q}{\reseatR{\kLI}{\myR}{\tupel{v}}}
	\end{eqnarray*}

	\item[\textbf{Part~3:}] \ The same argument as in the previous case works for this part as well.
	\item[\textbf{Part~4:}] \  This follows directly from 1) and 3).\qed\textsubscript{Lemma \ref{lem:hom_of_glued_is_product}}
\end{enumerate} 
\medskip

\section{Details ommitted in Section~\ref{sec:main_theorem}}\label{appendix:main_theorem}

\medskip

\subsection{Complete proof of Lemma \ref{lemma:HomToFormula}}
\label{appendix:subsec:HomToFormula}

\begin{proof}[\textbf{\upshape{Proof of
      Lemma~\ref{lemma:HomToFormula}}}] \ \\
Throughout this proof we use the following notation:
A \emph{segmentation} for a number $m \in \natpos$ is a pair of $d$-tuples
\[
  (\tupel{c}, \tupel{m}) \ \ = \ \  \big((c_1, \dots, c_d), (m_1,
  \dots, m_d)\big)
\]
of numbers in $\NNpos$,
such that $d\in\NNpos$ and $m_1<\cdots <m_d$ and
$\sum_{i \in [d]} c_i m_i = m$.
We call $d$ the \emph{degree} of $(\tupel{c}, \tupel{m})$
and $c \deff \sum_{i \in [d]} c_i$ its \emph{size}.

By $\mySeg{m}$ we denote the set of all segmentations of $m$.
Note that $\mySeg{m}$ is finite and non-empty for every $m\in\NNpos$.
  
The proof of Lemma~\ref{lemma:HomToFormula} proceeds by induction based on Definition~\ref{def:k_guarded_incidence_graphs}.
\smallskip

\noindent \textbf{Base case}:

1) \ $\rVI=\Img{r}$, $\bVI=\Img{b}$, and $\kLI$ has real guards.
Obviously, $\hom{\kLI}{\kLI'}\in\set{0,1}$.
Let
\begin{eqnarray*}
  S_E & \deff & \setc{(j,i)\in\Dom{b}\times\Dom{r}}{(b(j),r(i))\in\E{\I}},
\\                
  S_R & \deff & \setc{(i,i')}{i,i'\in\Dom{r},\ i<i',\ r(i)=r(i')}\,,
\\                
  S_B & \deff & \setc{(j,j')}{j,j'\in\Dom{b},\ j<j',\ b(j)=b(j')}\,.
\end{eqnarray*}                
We choose \;$\form{\kLI}{1} \deff$
\[
\Und_{(j,i)\in S_E}\!\!\! E(\vare_j,\varv_i)
\ \ \und
\Und_{(i,i')\in S_R}\!\!\! \varv_{i}{=}\varv_{i'}
\ \ \und
\Und_{(j,j')\in S_B}\!\!\! \vare_{j}{=}\vare_{j'}\,.
\]  
It is straightforward to verify
that $(\LogGuard{g}\und\form{\kLI}{1})\in\NGCk$. 

Let $\kLI'$ satisfy the lemma's assumptions. Since $\kLI'$ has real guards w.r.t.\ $g$ we
have $\myInt{\kLI'}\models\LogGuard{g}$. By construction of the
formula $\form{\kLI}{1}$ there is a homomorphism from $\kLI$ to
$\kLI'$ if, and only if, 
$\myInt{\kLI'}\models\form{\kLI}{1}$.
Thus, we are done for the case $m=1$.
\\
For $m=0$ we are done by choosing $\form{\kLI}{0}\deff\nicht\,\form{\kLI}{1}$.
\\
For $m\geq 2$ we know that $\hom{\kLI}{\kLI'}\neq m$. Thus, we are done by choosing an arbitrary unsatisfiable formula; to be specific we choose
$\form{\kLI'}{m}\deff (\form{\kLI'}{0}\und\form{\kLI'}{1})$.
\medskip

\noindent
\textbf{Inductive step:}
We assume that the lemma's statement is already shown for $\kLI=(\I,r,b,g)\in\GLIk$ and
for $\kLI_i=(\I_i,r_i,b_i,g_i)\in\GLIk$ for $i\in[2]$ where $g_1,g_2$ are compatible.
Our goal is to prove the lemma's statement for the particular $\tilde{\kLI}=(\tilde{\I},\tilde{r},\tilde{b},\tilde{g})\in\GLIk$
considered in the following case distinction.
\medskip

2) \
Let $\tilde{\kLI}\deff\reclaimR{\kLI}{\myR}$ for some $\myR\subseteq\Dom{r}$.
If $\myR=\emptyset$ we are done since $\tilde{\kLI}=\kLI$.
If $\myR\neq \emptyset$ let $\ell\deff |\myR|$ and let $\myR=\set{i_1,\ldots,i_\ell}$ with $i_1<\cdots<i_\ell$.
Note that $\tilde{r}=r-\myR$, $\tilde{b}=b$, and $\tilde{g}=g-\myR$.
We let 
\[
  \form{\tilde{\kLI}}{0} \ \ \deff \ \ 
  \forall (\varv_{i_1}, \dots, \varv_{i_\ell}).(\LogGuard{g} \impl
  \form{\kLI}{0})\,.
\]
It is straightforward to verify that
$(\LogGuard{\tilde{g}}\und \form{\tilde{\kLI}}{0})\in\NGCk$.
\\
For every $m\geq 1$ we let
\[
  \form{\tilde{\kLI}}{m} \ \ \isdef \ \
  \biglor_{(\tupel{c}, \tupel{m}) \in \mySeg{m}} \phi_{\tupel{c}, \tupel{m}}
\]
where for every
$(\tupel{c},\tupel{m})=((c_1,\ldots,c_d),(m_1,\ldots,m_d))\in\mySeg{m}$
and $c\deff\sum_{i\in [d]} c_i$
we let 
\[
 \begin{array}{ll}
  \phi_{\tupel{c},\tupel{m}} \ \ \isdef 
&
  \existsex[c] (\varv_{i_1}, \dots, \varv_{i_\ell}).
  (\LogGuard{g} \land \lnot \form{\kLI}{0}) \ \ \land
\\[1ex]
  & \bigland_{j\in[d]} \left(
    \existsex[c_j] (\varv_{i_1}, \dots, \varv_{i_\ell}).
    (\LogGuard{g} \land \form{\kLI}{m_j} \right)\,.
 \end{array}
\] 
It is straightforward to verify
that $(\LogGuard{\tilde{g}}\und \phi_{\tupel{c},\tupel{m}})
\in \NGCk$; and hence also $(\LogGuard{\tilde{g}}\und\form{\tilde{\kLI}}{m})\in\NGCk$.

Let $\kLI'=(\I',r',b',g')$ satisfy the lemma's assumptions for $\tilde{\kLI}$ instead of $\kLI$. I.e.,
we have
$\Dom{b'}\supseteq\Dom{\tilde{b}}=\Dom{b}$,
$\Dom{r'}\supseteq\Dom{\tilde{r}}=\Dom{r}\setminus\myR$, and
$\kLI'$ has real guards w.r.t.\ $\tilde{g}$. Thus,
$\myInt{\kLI'}\models\LogGuard{\tilde{g}}$.
\\
From Lemma~\ref{lem:HomomCountQ} we know that
\begin{equation}\label{eq:homCount:formula:reseatR}
  n
  \ \deff \
  \hom{\tilde{\kLI}}{\kLI'}
  \ \ =
  \sum_{\ov{v}\in\rV{\I'}^\ell} \hom{\kLI}{\reseatR{\kLI'}{\myR}{\tupel{v}}}.
\end{equation}
Note that $n=0$ $\iff$
$\hom{\kLI}{\reseatR{\kLI'}{\myR}{\tupel{v}}}=0$ for all
$\ov{v}\in\rV{\I'}^\ell$ $\iff$ for all $\ov{v}\in\rV{\I'}^\ell$ we
have either
$\myInt{\reseatR{\kLI'}{\myR}{\tupel{v}}}\not\models\LogGuard{g}$
(this is in case that $\myInt{\reseatR{\kLI'}{\myR}{\tupel{v}}}$ does
not have real guards w.r.t.\ $g$) or
$\myInt{\reseatR{\kLI'}{\myR}{\tupel{v}}}\models(\LogGuard{g}\und \form{\kLI}{0})$
$\iff$ $\myInt{\kLI'}\models \form{\tilde{\kLI}}{0}$.
Thus, the formula $\form{\tilde{\kLI}}{0}$ has the desired meaning.

Let us now consider the case where $n\geq 1$.
For every number $m'\in[n]$ let
$S_{m'}$ be the set of all $\ov{v}\in\rV{\I'}^\ell$ such that
$m'=\hom{\kLI}{\reseatR{\kLI'}{\myR}{\tupel{v}}}$.
Let $d'$ be the number of distinct $m'\in[n]$ for which $S_{m'}\neq \emptyset$, and let
$m'_1<\cdots<m'_{d'}$ be an ordered list of all $m'\in[n]$ with $S_{m'}\neq \emptyset$.
For each $j\in[d']$ let $c'_j\deff |S_{m'_j}|$. Let $\tupel{c}'\deff (c'_1,\ldots,c'_{d'})$ and
$\tupel{m}'\deff (m'_1,\ldots,m'_{d'})$.
Note that
$\sum_{j\in[d']}c'_j m'_j = n$. Hence
$(\tupel{c}',\tupel{m}')\in\mySeg{n}$ is a segmentation of $n$
of size $c'\deff\sum_{j=1}^{d'} c'_j$ and of degree $d'$.

Note that $\myInt{\kLI'}\models \phi_{\tupel{c}',\tupel{m}'}$. Furthermore,
for every $m\in\NNpos$ we have:
$\myInt{\kLI'}\models \form{\tilde{\kLI}}{m}$ $\iff$ $(\tupel{m}',\tupel{c}')\in\mySeg{m}$
$\iff$ $m=n=\hom{\tilde{\kLI}}{\kLI'}$.
This verifies that for every $m\geq 1$ the formula $\form{\tilde{\kLI}}{m}$ has the
desired properties.
\medskip

3) \
Let $\tilde{\kLI}\deff\reclaimB{\kLI}{\myB}$ for some $\myB\subseteq\Dom{b}\setminus\Img{g}$.
If $\myB=\emptyset$ we are done since $\tilde{\kLI}=\kLI$.
If $\myB\neq \emptyset$ let $\ell\deff |\myB|$ and let
$\myB=\set{i_1,\ldots,i_\ell}$ with $i_1<\cdots<i_\ell$. Note that
$\tilde{r}=r$, $\tilde{b}=b-\myB$, and $\tilde{g}=g$. 
We let 
\[
  \form{\tilde{\kLI}}{0} \ \ \deff \ \ 
  \forall (\vare_{i_1}, \dots, \vare_{i_\ell}).(\LogGuard{g} \impl
  \form{\kLI}{0})\,.
\]
For every $m\geq 1$ we let
\[
  \form{\tilde{\kLI}}{m} \ \ \isdef \ \
  \biglor_{(\tupel{c}, \tupel{m}) \in \mySeg{m}} \phi'_{\tupel{c}, \tupel{m}}
\]
where for every
$(\tupel{c},\tupel{m})=((c_1,\ldots,c_d),(m_1,\ldots,m_d))\in\mySeg{m}$
and $c\deff\sum_{i\in [d]}c_i$
we let
\[
 \begin{array}{ll}
  \phi'_{\tupel{c},\tupel{m}} \ \ \isdef 
&
  \existsex[c] (\vare_{i_1}, \dots, \vare_{i_\ell}).
  (\LogGuard{g} \land \lnot \form{\kLI}{0}) \ \ \land
\\[1ex]
  & \bigland_{j\in[d]} \left(
    \existsex[c_j] (\vare_{i_1}, \dots, \vare_{i_\ell}).
    (\LogGuard{g} \land \form{\kLI}{m_j} \right)\,.
 \end{array}
\] 
Let $\kLI'=(\I',r',b',g')$ satisfy the lemma's assumptions for $\tilde{\kLI}$ instead of $\kLI$. I.e.,
we have
$\Dom{b'}\supseteq\Dom{\tilde{b}}=\Dom{b}\setminus\myB$,
$\Dom{r'}\supseteq\Dom{\tilde{r}}=\Dom{r}$, and
$\kLI'$ has real guards w.r.t.\ $\tilde{g}$. Thus,
$\myInt{\kLI'}\models\LogGuard{\tilde{g}}$.
From Lemma~\ref{lem:HomomCountQ} we know that
\begin{equation}\label{eq:homCount:formula:reseatR}
  n
  \ \deff \
  \hom{\tilde{\kLI}}{\kLI'}
  \ \ =
  \sum_{\ov{e}\in\bV{\I'}^\ell} \hom{\kLI}{\reseatB{\kLI'}{\myB}{\tupel{e}}}.
\end{equation}
Note that $n=0$ $\iff$
$\hom{\kLI}{\reseatB{\kLI'}{\myB}{\tupel{e}}}=0$ for all
$\ov{e}\in\bV{\I'}^\ell$ $\iff$ for all $\ov{e}\in\bV{\I'}^\ell$ we
have 
either $\myInt{\reseatB{\kLI'}{\myB}{\tupel{e}}}\not\models\LogGuard{g}$ or
$\myInt{\reseatB{\kLI'}{\myB}{\tupel{e}}}\models(\LogGuard{g}\und \form{\kLI}{0})$
$\iff$
$\myInt{\kLI'}\models \form{\tilde{\kLI}}{0}$.
Thus, the formula $\form{\tilde{\kLI}}{0}$ has the desired meaning.

Let us now consider the case where $n\geq 1$.
For every number $m'\in[n]$ let
$S'_{m'}$ be the set of all $\ov{e}\in\bV{\I'}^\ell$ such that
$m'=\hom{\kLI}{\reseatB{\kLI'}{\myB}{\tupel{e}}}$.
Let $d'$ be the number of distinct $m'\in[n]$ for which $S'_{m'}\neq \emptyset$, and let
$m'_1<\cdots<m'_{d'}$ be an ordered list of all $m'\in[n]$ with $S'_{m'}\neq \emptyset$.
For each $j\in[d']$ let $c'_j\deff |S'_{m'_j}|$. Let $\tupel{c}'\deff (c'_1,\ldots,c'_{d'})$ and
$\tupel{m}'\deff (m'_1,\ldots,m'_{d'})$.
Note that
$\sum_{j\in[d']}c'_j m'_j = n$. Hence
$(\tupel{c}',\tupel{m}')\in\mySeg{n}$ is a segmentation of $n$
of size $c'\deff\sum_{j\in [d']} c'_j$ and of degree $d'$.

Note that $\myInt{\kLI'}\models \phi'_{\tupel{c}',\tupel{m}'}$. Furthermore,
for every $m\in\NNpos$ we have:
$\myInt{\kLI'}\models \form{\tilde{\kLI}}{m}$ $\iff$ $(\tupel{m}',\tupel{c}')\in\mySeg{m}$
$\iff$ $m=n=\hom{\tilde{\kLI}}{\kLI'}$.
This verifies that for every $m\geq 1$ the formula $\form{\tilde{\kLI}}{m}$ has the
desired meaning.

It remains to verify that
$(\LogGuard{\tilde{g}}\und \form{\tilde{\kLI}}{m})\in\NGCk$ for every $m\in\NN$.
Recall that $\tilde{g}=g$.
Thus, $\tilde{g}$ satisfies the condition~\eqref{eq:syntax:guard} of
rule~\ref{item:syntaxdef:10}) of the syntax definition of $\NGCk$.
Since, by the induction hypothesis,
$(\LogGuard{g}\und \form{\kLI}{m_j})\in\NGCk$, we obtain by rule
\ref{item:syntaxdef:10}) that the conjunction of $\LogGuard{\tilde{g}}$ with
the formula in the second line of $\varphi'_{\tupel{c},\tupel{m}}$
belongs to $\NGCk$.
The same reasoning (combined with rule~\ref{item:syntaxdef:4})) yields that the conjunction of $\LogGuard{\tilde{g}}$ with
the formula in the first line of $\varphi'_{\tupel{c},\tupel{m}}$
belongs to $\NGCk$.
Applying rule \ref{item:syntaxdef:5}) 
yields that each formula 
$(\LogGuard{\tilde{g}}\und \phi'_{\tupel{c},\tupel{m}})$
as well as the formula
$(\LogGuard{\tilde{g}}\und\form{\tilde{\kLI}}{m})$
belongs to $\NGCk$, for every $m\geq 1$.
A similar reasoning shows that also $(\LogGuard{\tilde{g}}\und\form{\tilde{\kLI}}{0})$
belongs to $\NGCk$.

\medskip

4) \ Let $\tilde{\kLI}\deff \switch{\kLI}{f}$, where $f$ is a transition for $g$.
I.e.,
$\emptyset\neq\Dom{f}\subseteq\Dom{g}$, such that for every $i\in\Dom{g}$ with
$g(i)\in\Img{f}$ we have $i\in\Dom{f}$. By definition,
$\tilde{\kLI}=\glue{\kLIf}{\reclaimB{\kLI}{\myB}}$ for $\myB\deff\Img{g}\cap\Img{f}\cap\Dom{b}$.
In particular, $\tilde{g}=f\cup g$, i.e., $\tilde{g}(i)=f(i)$ for all $i\in\Dom{f}$, and
$\tilde{g}(i)=g(i)$ for all $i\in\Dom{g}\setminus\Dom{f}$.

Consider an arbitrary $\kLI'=(\I',r',b',g')$ that satisfies the lemma's assumptions for $\tilde{\kLI}$ instead of $\kLI$. I.e.,
we have
$\Dom{b'}\supseteq\Dom{\tilde{b}}=\Img{f}\cup\Dom{b}$,
$\Dom{r'}\supseteq\Dom{\tilde{r}}=\Dom{f}\cup\Dom{r}=\Dom{f}\cup\Dom{g}=\Dom{\tilde{g}}$,
and 
$\kLI'$ has real guards w.r.t.\ $\tilde{g}$.
Since $\kLI'$ has real guards w.r.t.\ $\tilde{g}$ we have
$\myInt{\kLI'}\models\LogGuard{\tilde{g}}$.
From Lemma~\ref{lem:HomomCountQ} we know that $n\deff\hom{\tilde{\kLI}}{\kLI'}=$
\begin{equation*}
  \hom{\kLIf}{\kLI'}\ \ {\cdot} \!\!\sum_{\ov{e}\in\bV{\I'}^\ell}
  \hom{\kLI}{\reseatB{\kLI'}{\myB}{\tupel{e}}} 
\end{equation*}
for $\myB\deff \Dom{b}\cap\Img{f}\cap\Img{g}$.
Furthermore, we know that $\hom{\kLIf}{\kLI'}\in\set{0,1}$.
Note that $\hom{\kLIf}{\kLI'}=1$ $\iff$
$\myInt{\kLI'}\models\LogGuard{f}$.
Since $\myInt{\kLI'}\models\LogGuard{\tilde{g}}$ and $\tilde{g}\supseteq f$,
we obviously have $\myInt{\kLI'}\models\LogGuard{f}$.
Thus, $n=\sum_{\ov{e}\in\bV{\I'}^\ell} \hom{\kLI}{\reseatB{\kLI'}{\myB}{\tupel{e}}}$, as
in case~3) in equation~\eqref{eq:homCount:formula:reseatR}.
We therefore can choose the exact same formulas $\form{\tilde{\kLI}}{m}$ as in case~3)
for all $m\in\NN$ and obtain that these formulas have the desired meaning.

It remains to verify that
$(\LogGuard{\tilde{g}}\und \form{\tilde{\kLI}}{m})\in\NGCk$ for every $m\in\NN$.
Note that now we have $\tilde{g}=f\cup g$ (whereas in case~3) we had
to deal with the case where $\tilde{g}=g$).
All we have to do is argue that $\tilde{g}$ satisfies the
condition~\eqref{eq:syntax:guard} of 
rule~\ref{item:syntaxdef:10}) of the syntax definition of $\NGCk$.
Once having achieved this, the same reasoning as in case~3)
proves that the formula $(\LogGuard{\tilde{g}}\und
\form{\tilde{\kLI}}{m})$ belongs to 
$\NGCk$ for every $m\in\NN$.

Note that according to the induction hypothesis, for all $m'\in\NN$ we have:
$\Dom{b}=\indfreeB{(\LogGuard{g}\und\form{\kLI}{m'})}\supseteq\Img{g}$.
Therefore, $\myB=\Img{f}\cap\Img{g}$. 

Consider an arbitrary $i\in\Dom{g}$. We have to show that
$\tilde{g}(i)=g(i)$ or $\tilde{g}(i)\in \myB$ or $\tilde{g}(i)\not\in
\Img{g}$. 

If $\tilde{g}(i)=g(i)$
or $\tilde{g}(i)\not\in\Img{g}$, we are done.
Consider the case that
$\tilde{g}(i)\neq g(i)$
and
$\tilde{g}(i)\in\Img{g}$.
In this case, $\tilde{g}(i)=f(i)$ and there exists an $i'\in\Dom{g}$
with $g(i')=\tilde{g}(i)=f(i)$.
Hence, $\tilde{g}(i)\in\Img{f}\cap\Img{g}=\myB$, and we are done.

\medskip

5) \ Let $\tilde{\kLI}\deff \glue{\kLI_1}{\kLI_2}$. Recall from the
inductive step's assumption that the lemma's statement is already
shown for $\kLI_i=(\I_i,r_i,b_i,g_i)$ for each $i\in[2]$ and that
$g_1,g_2$ are compatible.
Recall from Definition~\ref{def:glueing} that
$\tilde{\kLI}=(\tilde{I},\tilde{r},\tilde{b},\tilde{g})$ with
$\Dom{\tilde{r}}=\Dom{r_1}\cup\Dom{r_2}$, 
$\Dom{\tilde{b}}=\Dom{b_1}\cup\Dom{b_2}$, and
$\tilde{g}=g_1\cup g_2$.
We let
\[
  \form{\tilde{\kLI}}{0} \ \ \deff \ \
  \big(\,
    \form{\kLI_1}{0}\ \oder \ \form{\kLI_2}{0}
  \,\big)\,.
\]
For each $m\geq 1$ we let $S_m$ be the set of all tuples $(m_1,m_2)$
of positive integers such that $m_1\cdot m_2=m$ (note that $S_m$ is
finite and non-empty), and let
\[
  \form{\tilde{\kLI}}{m} \ \ \deff \ \
  \Oder_{(m_1,m_2)\in S_m}\big(\,
    \form{\kLI_1}{m_1}\ \und \ \form{\kLI_2}{m_2}
  \,\big)\,.
\]
By the induction hypothesis we know that $(\LogGuard{g_i}\und
\form{\kLI_i}{m_i})\in\NGCk$
for each $i\in[2]$ and every $m_i\in\NN$. Since $\tilde{g}=g_1\cup
g_2$ and $g_1,g_2$ are compatible, we obtain that the formula
$(\LogGuard{\tilde{g}}\und \form{\tilde{\kLI}}{m})$ belongs to $\NGCk$
for every $m\in \NN$.

All that remains to be done is to show that these formulas have the
desired meaning.
Consider an arbitrary $\kLI'=(\I',r',b',g')$ that satisfies the
lemma's assumptions
for $\tilde{\kLI}$ instead of $\kLI$. I.e., we have
$\Dom{b'}\supseteq\Dom{\tilde{b}}=\Dom{b_1}\cup\Dom{b_2}$,
$\Dom{r'}\supseteq\Dom{\tilde{r}}=\Dom{b_1}\cup\Dom{b_2}$, and
$g'\supseteq \tilde{g}= g_1 \cup g_2$.

Since $\kLI'$ has real guards w.r.t.\ $\tilde{g}$ we have
$\myInt{\kLI'}\models\LogGuard{\tilde{g}}$.
From Lemma~\ref{lem:HomomCountQ} we know that
\begin{equation}\label{eq:homCount:formula:glue}
  n
  \ \deff \
  \hom{\tilde{\kLI}}{\kLI'}
  \ = \ 
  \hom{\kLI_1}{\kLI'} \cdot 
  \hom{\kLI_2}{\kLI'}\,.
\end{equation}
Note that $n=0$ $\iff$ $\hom{\kLI_1}{\kLI'}=0$ or
$\hom{\kLI_2}{\kLI'}=0$.
Thus, $\myInt{\kLI'}\models\form{\tilde{\kLI}}{0}$ $\iff$ $n=0$.
Hence, the formula $\form{\tilde{\kLI}}{0}$ has the desired meaning.

Let us now consider the case where $n\geq 1$.
For $i\in[2]$ let $n_i\deff\hom{\kLI_i}{\kLI'}$. 
By \eqref{eq:homCount:formula:glue} we have $n=n_1\cdot n_2$. Thus,
$(n_1,n_2)\in S_n$.
Note for every $m\geq 1$ that
$\myInt{\kLI'}\models\form{\tilde{\kLI}}{m}$ $\iff$ $(n_1,n_2)\in S_m$
$\iff$ $n=n_1\cdot n_2 = m$ $\iff$ $\hom{\tilde{\kLI}}{\kLI'}=m$.
Hence, for every $m\geq 1$ the formula $\form{\tilde{\kLI}}{m}$ has the desired meaning.
This completes the proof of Lemma~\ref{lemma:HomToFormula}.
\end{proof}
 
\medskip

\subsection{Complete proof of Lemma \ref{lemma:FormulaToHom}}
\label{appendix:subsec:FormulaToHom}

\begin{proof}[\textbf{\textup{Proof of Lemma~\ref{lemma:FormulaToHom}}}] \ \\
We proceed by induction on the construction of $\chi$.
\smallskip

\noindent \textbf{Base cases}:
$\chi$ is $(\LogGuard{g} \land \psi)$ and matches one of the
following:
\smallskip

\begin{enumerate}[wide, label={\textbf{Case \arabic*}:}]

\item
$\psi$ is $E(\vare_j, \varv_i)$ and $\Dom{g} = \set{i}$ and $ j\in
[k]$. Note that $\chi= (E(\vare_{g(i)},\varv_i)\und E(\vare_j,\varv_i))$.
Let $\I$ be the incidence graph defined as follows.
$\I$ has exactly one red node, i.e., $\rVI=\set{v_i}$.
If $g(i)=j$ then $\I$ has exactly one blue node, i.e.,
$\bVI=\set{e_j}$; otherwise $\I$ has exactly two blue nodes, i.e.,
$\bVI=\set{e_j,e_{g(i)}}$.
$\I$ has edges from every blue node to the red node, i.e.,
$\E{\I}=\bVI\times\rVI$.
Consider the $k$-labeled incidence graph $\kLI=(\I,r,b,g)$ where
$\Dom{r}=\set{i}$ and $r(i)=v_i$, and $\Dom{b}=\set{j,g(i)}$ and
$b(j)=e_j$ and $b(g(i))=e_{g(i)}$ (and $g$ is the formula $\chi$'s
guard function).
Note that $\kLI\in\GLIk$ (use the base case of
Definition~\ref{def:k_guarded_incidence_graphs}).
For every $k$-labeled incidence graph $\kLI'$ we have:
$\hom{\kLI}{\kLI'}\in\set{0,1}$, and
$\hom{\kLI}{\kLI'}=1$ $\iff$ $\myInt{\kLI'}\models \chi$.
Thus we are done by choosing $Q\deff \kLI$.
\smallskip

\item
$\psi$ is $\vare_{j} {=} \vare_{j'}$ and $\Dom{g}=\emptyset$ and
$j,j'\in[k]$.
Note that $\chi= (\top\und \vare_j{=}\vare_{j'})$.
Let $\I$ be the incidence graph with no red node and
exactly one blue node, i.e., $\rVI=\emptyset$, $\bVI=\set{e}$, and
$\E{\I}=\emptyset$.
Consider the $k$-labeled incidence graph $\kLI=(\I,r,b,g)$ where
$\Dom{r}=\emptyset$, $\Dom{b}=\set{j,j'}$ and $b(j)=b(j')=e$ (and $g$
is the formula $\chi$'s (empty) guard function).
Note that $\kLI\in\GLIk$ (use the base case of
Definition~\ref{def:k_guarded_incidence_graphs}).
For every $k$-labeled incidence graph $\kLI'$ we have:
$\hom{\kLI}{\kLI'}\in\set{0,1}$, and
$\hom{\kLI}{\kLI'}=1$ $\iff$ $\myInt{\kLI'}\models \chi$.
Thus we are done by choosing $Q\deff \kLI$.
\smallskip

\item
$\psi$ is $\varv_i {=} \varv_{i'}$ with $\Dom{g}=\set{i,i'}$.
Note that $\chi= (E(\vare_{g(i)},\varv_i)\und
E(\vare_{g(i')},\varv_{i'})\und \varv_{i}{=}\varv_{i'})$.
Let $\I$ be the incidence graph defined as follows.
$\I$ has exactly one red node, i.e., $\rVI=\set{v}$.
If $g(i)=g(i')$ then $\I$ has exactly one blue node, i.e.,
$\bVI=\set{e}$; otherwise $\I$ has exactly two blue nodes, i.e.,
$\bVI=\set{e_{g(i)},e_{g(i')}}$.
$\I$ has edges from every blue node to the red node, i.e.,
$\E{\I}=\bVI\times\rVI$.
Consider the $k$-labeled incidence graph $\kLI=(\I,r,b,g)$ where
$g$ is the formula $\chi$'s
guard function, and
$\Dom{r}=\set{i,i'}$ and $r(i)=r(i')=v$, and
$\Dom{b}=\set{g(i),g(i')}$.
In case that $g(i)=g(i')$ we have $b(g(i))=e$, and otherwise we have
$b(g(i))=e_{g(i)}$ and $b(g(i'))=e_{g(i')}$.
Note that $\kLI\in\GLIk$ (use the base case of
Definition~\ref{def:k_guarded_incidence_graphs}).
For every $k$-labeled incidence graph $\kLI'$ we have:
$\hom{\kLI}{\kLI'}\in\set{0,1}$, and
$\hom{\kLI}{\kLI'}=1$ $\iff$ $\myInt{\kLI'}\models \chi$.
Thus we are done by choosing $Q\deff \kLI$.
\smallskip

\item[{\textbf{Special Case}:}]
$\psi$ is unfulfillable for all $k$-labeled incidence graphs $\kLI'$ matching
the requirements imposed by the lemma.

If $\Dom{g}\neq\emptyset$ we choose $\kLI\deff\kLIg$ and $Q\deff
0{\cdot}\kLI$.
Note that $Q\in\QGLIk$ and $g_Q,\dombQ,\domrQ$ are as required by the
lemma. Furthermore, $\hom{Q}{\kLI'}=0$ for every $k$-labeled incidence
graph $\kLI'$.

If $\Dom{g}=\emptyset$ and $\free{\chi}\neq \emptyset$ we choose $\I$
to consist of a single blue node $e$ and no red nodes, and we let
$\kLI=(\I,r,b,g)$ with $\Dom{r}=\Dom{g}=\emptyset$ and 
$\Dom{b}=\indfreeB{\chi}$ and
$b(j)=e$ for
every $j\in\indfreeB{\chi}$. Let $Q\deff 0{\cdot}\kLI$.
Note that $Q\in\QGLIk$ and $g_Q,\dombQ,\domrQ$ are as required by the
lemma. Furthermore, $\hom{Q}{\kLI'}=0$ for every $k$-labeled incidence
graph $\kLI'$.

If $\Dom{g}=\emptyset$ and $\free{\chi}= \emptyset$ we choose $\I$ to
be empty, i.e., $\rVI=\bVI=\E{\I}=\emptyset$,
and we let
$\kLI=(\I,r,b,g)$ with $\Dom{r}=\Dom{b}=\Dom{g}=\emptyset$.
Let $Q\deff 0{\cdot}\kLI$.
Note that $Q\in\QGLIk$ and $g_Q,\dombQ,\domrQ$ are as required by the
lemma. Furthermore, $\hom{Q}{\kLI'}=0$ for every $k$-labeled incidence
graph $\kLI'$.
\setcounter{LogicCounter}{\value{enumi}}
\end{enumerate}

\medskip

\noindent \textbf{Inductive step}: We assume that the lemma's
statement is already shown for
$\chi\deff (\LogGuard{g}\und\psi)\in\NGCk$ and for
$\chi_i\deff(\LogGuard{g_i}\und\psi_i)\in\NGCk$ for $i\in[2]$, where
$g_1,g_2$ are compatible.

Fix arbitrary $m,d\in\NN$ with $m\geq 1$.
Let $Q,Q_1,Q_2\in\QGLIk$ be the quantum incidence graphs provided by
the lemma's statement for
$\chi,\chi_1,\chi_2$ and the parameters $m,d$. 

Our goal is to prove the lemma's statement for parameters $m,d$ and the particular
$\tilde{\chi} \deff
(\LogGuard{\tilde{g}} \land \tilde{\psi})\in\NGCk$ considered in the
following case distinction.
We will write $\tilde{Q}$ to denote the quantum incidence
graph that has to be constructed for $\tilde{\chi}$.
\smallskip

\begin{enumerate}[wide, label={\textbf{Case \arabic*}:}]
\setcounter{enumi}{\value{LogicCounter}}

\item
$\tilde{\psi} = \lnot \psi$ and $\tilde{g} = g$.
We choose $\tilde{Q} \isdef Q[\set{1},\set{0}]$
according to Lemma~\ref{lemma:normalizing_quantum_incidence_graphs}.
Note that $\tilde{Q}\in\QGLIk$ and $\tilde{Q}$ has the parameters
$g_{\tilde{Q}},\domb{\tilde{Q}},\domr{\tilde{Q}}$ required by the
lemma.
Furthermore, for all $k$-labeled incidence graphs $\kLI'$ satisfying
the lemma's assumptions we have: $\myInt{\kLI'}\models\LogGuard{g}$
and $\hom{\tilde{Q}}{\kLI'}\in\set{0,1}$. Furthermore,
$\hom{\tilde{Q}}{\kLI'}=1$ $\iff$ $\hom{Q}{\kLI'}=0$ $\iff$
$\myInt{\kLI'}\not\models \chi$. Since
$\myInt{\kLI'}\models\LogGuard{g}$ by assumption, we have:
$\myInt{\kLI'}\not\models \chi$ $\iff$ $\myInt{\kLI'}\not\models\psi$
$\iff$ $\myInt{\kLI'}\models\nicht\psi$
$\iff$ $\myInt{\kLI'}\models\tilde{\chi}$. 
This completes Case~4.
\smallskip

\item
$\tilde{\psi} = (\psi_1 \land \psi_2)$ and $\tilde{g} = g_1 \union
g_2$ (recall that $g_1$
and $g_2$ are compatible).
We choose $\tilde{Q} \isdef \glue{Q_1}{Q_2}$.
Note that $\tilde{Q}\in\QGLIk$ and $\tilde{Q}$ has the parameters
$g_{\tilde{Q}},\domb{\tilde{Q}},\domr{\tilde{Q}}$ required by the lemma.

Consider an arbitrary $k$-labeled incidence graph $\kLI'$ satisfying
the lemma's assumptions.
By Lemma~\ref{lem:hom_of_glued_is_product} we have
$\hom{\tilde{Q}}{\kLI'} = \hom{Q_1}{\kLI'} \cdot
\hom{Q_2}{\kLI'}$.
By the induction hypothesis we know for each $i\in[2]$ that
$\hom{Q_i}{\kLI'}\in\set{0,1}$; and
$\hom{Q_i}{\kLI'}=1$ $\iff$ $\myInt{\kLI'}\models\chi_i$.
Thus, $\hom{\tilde{Q}}{\kLI'}\in\set{0,1}$. Furthermore,
$\hom{\tilde{Q}}{\kLI'}=1$ $\iff$ $\myInt{\kLI'}\models\chi_1$ and
$\myInt{\kLI'}\models\chi_2$ $\iff$ $\myInt{\kLI'}\models \tilde{\chi}$.
This completes Case~5.
\smallskip
	
\item
$\tilde{\psi} = \existsi[n]
(\varv_{i_1},\ldots,\varv_{i_\ell}).(\LogGuard{g} \land \psi)$ and
$\tilde{g} = g - \set{ i_1, \ldots, i_\ell }$ and $n,\ell\in\NNpos$
and $i_1<\cdots<i_\ell$ and 
$\myR \deff \set{ i_1, \ldots, i_\ell } \subseteq \Dom{g}$.

If $n > (md)^\ell$, the
formula $\tilde{\psi}$
is unfulfillable for all $k$-labeled incidence graphs $\kLI'$ matching
the requirements imposed by the lemma (simply because each such
$\kLI'$ has at most $md$ red nodes). Thus, we are done by choosing
$\tilde{Q}$ according to the ``special~case'' handled above.

If $n \leq (md)^\ell$, we proceed as follows.
Let $Q'\deff \reclaimR{Q}{\myR}$ and
$\tilde{Q}\deff Q'[X,Y]$ for $X\deff\set{0,\ldots,n{-}1}$ and
$Y\deff\set{n,\ldots,(md)^\ell}$.
Note that $Q'\in\QGLIk$ and $\tilde{Q}\in\QGLIk$, and
$\tilde{Q}$ has the parameters
$g_{\tilde{Q}},\domb{\tilde{Q}},\domr{\tilde{Q}}$ required by the
lemma.
Consider an arbitrary $k$-labeled incidence graph $\kLI'$ satisfying
the lemma's assumptions for $\tilde{\chi}$. In particular, $\kLI'$ has
real guards w.r.t.\ $\tilde{g}$ and therefore satisfies
$\LogGuard{\tilde{g}}$. Thus, $\myInt{\kLI'}\models\tilde{\chi}$
$\iff$ $\myInt{\kLI'}\models\tilde{\psi}$.
By Lemma~\ref{lem:hom_of_glued_is_product} we have
\[
  \hom{Q'}{\kLI'} \ = \
  \sum_{\tupel{v} \in \rV{\I_{\kLI'}}^{\ell}}
  \hom{Q}{\reseatR{\kLI'}{\myR}{\tupel{v}}}\,.
\]  
Consider an arbitrary $\tupel{v}$ with
$\hom{Q}{\reseatR{\kLI'}{\myR}{\tupel{v}}} > 0$.
Since $Q$ has real guards, the $k$-labeled incidence graph
$\reseatR{\kLI'}{\myR}{\tupel{v}}$ satisfies the lemma's assumptions
w.r.t.\ the formula $\chi$ (in particular, it has real guards w.r.t.\
$g$). Therefore, by the induction hypothesis we obtain that
$\hom{Q}{\reseatR{\kLI'}{\myR}{\tupel{v}}}=1$ and
$\myInt{\reseatR{\kLI'}{\myR}{\tupel{v}}}\models\chi$.

Therefore,
$\hom{Q'}{\kLI'}\leq |\rV{\I_{\kLI'}}|^\ell \leq (md)^\ell$.
Moreover, $\hom{Q'}{\kLI'}$ is exactly the number of tuples
$\tupel{v}\in \rV{\I_{\kLI'}}^\ell$ such that
$\myInt{\reseatR{\kLI'}{\myR}{\tupel{v}}}\models \chi$.

By our choice of $X,Y$ and $\tilde{Q}= Q'[X,Y]$ we obtain that
$\hom{\tilde{Q}}{\kLI'}\in\set{0,1}$ and, moreover,
$\hom{\tilde{Q}}{\kLI'}=1$ $\iff$ $\myInt{\kLI'}\models\tilde{\chi}$.
This completes Case~6.
\medskip

\item
$\tilde{\psi} = \existsi[n]
(\vare_{i_1},\ldots,\vare_{i_\ell}).(\LogGuard{g} \land \psi)$ and
$\tilde{g}$, where $n,\ell\in\NNpos$ and
$S\deff\set{i_1,\ldots,i_\ell}\subseteq \indfreeB{\chi}$ for
$\chi\deff(\LogGuard{g}\und\psi)$
with
$i_1<\cdots <i_\ell$ and $\Dom{\tilde{g}}=\Dom{g}$ and
all $i\in\Dom{g}$ satisfy
\begin{equation}\label{eq:neuerSyntaxfall:tildeg}
  \tilde{g}(i)=g(i)
  \quad\text{or}\quad
  \tilde{g}(i)\in S
  \quad \text{or} \quad
  \tilde{g}(i)\not\in
  \Img{g}.
\end{equation}

If $n>m^\ell$, the formula $\tilde{\psi}$ is unfulfillable for all
$k$-labeled incidence graphs $\kLI'$ matching the requirements imposed
by the lemma (simply because each such $\kLI'$ has only $m$ blue
nodes).
Thus, we are done by choosing $\tilde{Q}$ according to the ``special
case'' handled above.

If $n\leq m^\ell$, we proceed as follows.
Since $Q$ is obtained by the induction hypothesis, we know in
particular that 
$g_Q=g$, $\dombQ=\indfreeB{\chi}$, and $\domrQ=\Dom{g}=\indfreeR{\chi}$.
In particular, we have: $S\subseteq \dombQ$.

\bigskip

\noindent
\textit{Claim 1:} \
Consider an arbitrary $k$-labeled incidence graph $\kLI'$ satisfying
the lemma's assumptions for $\tilde{\chi}$.
For
\[
  \mysum{\kLI'} \ \deff \
  \hom{\reclaimB{Q}{S}}{\kLI'}
\]
we have:
\,$
  \mysum{\kLI'}  = 
  \sum_{\tupel{e} \in \bV{\I_{\kLI'}}^{\ell}}
  \hom{Q}{\reseatB{\kLI'}{S}{\tupel{e}}}
$\,
and $0\leq\mysum{\kLI'}\leq |\bV{\I_{\kLI'}}|^\ell$
and $\mysum{\kLI'}$ is exactly the number of tuples
$\tupel{e}\in \bV{\I_{\kLI'}}^\ell$ such that 
$\myInt{\reseatB{\kLI'}{\myB}{\tupel{e}}}\models \chi$.
\medskip

\noindent
\textit{Proof of Claim 1:} \
Since $\kLI'$ satisfies the lemma's assumptions for $\tilde{\chi}$,
it has
real guards w.r.t.\ $\tilde{g}$ and therefore satisfies
$\LogGuard{\tilde{g}}$.
Thus, $\myInt{\kLI'}\models\tilde{\chi}$
$\iff$ $\myInt{\kLI'}\models\tilde{\psi}$.

Furthermore, $\dombQ=\indfreeB{\chi}\supseteq S$. 
Lemma~\ref{lem:hom_of_glued_is_product} yields:
\[
  \mysum{\kLI'} \ = \ 
  \sum_{\tupel{e} \in \bV{\I_{\kLI'}}^{\ell}}\!\!\!
  \hom{Q}{\reseatB{\kLI'}{S}{\tupel{e}}}\,.
\]
Consider an arbitrary $\tupel{e}$ with
$\hom{Q}{\reseatB{\kLI'}{S}{\tupel{e}}} > 0$.
Since $Q$ has real guards, the $k$-labeled incidence graph
$\reseatB{\kLI'}{S}{\tupel{e}}$ satisfies the lemma's assumptions
w.r.t.\ the formula $\chi$ (in particular, it has real guards w.r.t.\
$g_Q=g$). Therefore, by the induction hypothesis we obtain that
$\hom{Q}{\reseatB{\kLI'}{S}{\tupel{e}}}=1$ and
$\myInt{\reseatB{\kLI'}{S}{\tupel{e}}}\models\chi$.

Hence,
$0\leq \mysum{\kLI'}\leq |\bV{\I_{\kLI'}}|^\ell \leq m^\ell$.
Moreover, $\mysum{\kLI'}$ is exactly the number of tuples
$\tupel{e}\in \bV{\I_{\kLI'}}^\ell$ such that
$\myInt{\reseatB{\kLI'}{\myB}{\tupel{e}}}\models \chi$.
This completes the proof of Claim~1.
\qed$_\textit{Claim 1}$
\bigskip

\noindent
For the remaining parts of the proof, our aim is to
\begin{equation}\label{eq:Goal:neuerFall}
\begin{array}{l}
\parbox{12cm}{construct a $Q'\in\QGLIk$
that has the parameters
$g_{\tilde{Q}},\domb{\tilde{Q}}, \domr{\tilde{Q}}$ required by the
lemma for formula $\tilde{\chi}$,
such that $\hom{Q'}{\kLI'}=\mysum{\kLI'}$ for all  $k$-labeled incidence graphs $\kLI'$
satisfying the lemma's assumptions for $\tilde{\chi}$.}
\end{array}
\end{equation}

\noindent
Once having achieved this, we can use
Lemma~\ref{lemma:normalizing_quantum_incidence_graphs} to obtain 
a $k$-labeled quantum incidence graph
$\tilde{Q}\deff Q'[X,Y] \in\QGLIk$ for $X\deff\set{0,\ldots,n{-}1}$ and $Y\deff\set{n,\ldots,m^\ell}$ such that
$\hom{\tilde{Q}}{\kLI'}\in\set{0,1}$ and, moreover,
$\hom{\tilde{Q}}{\kLI'}=1$ $\iff$
$\mysum{\kLI'}\geq n$ $\iff$
$\myInt{\kLI'}\models\tilde{\chi}$.

Hence, all that remains to be done to finish the proof is to achieve
\eqref{eq:Goal:neuerFall}.
We let
\begin{equation}\label{eq:Def:Z}
  Z\ \deff\ \setc{\,i\in\Dom{g}\ }{\ \tilde{g}(i)\neq g(i) \text{ \ or \ } g(i)\in S\,}\,.
\end{equation}

\noindent
If $Z=\emptyset$, we have $\tilde{g}=g$ and $\Img{g}\cap
S=\emptyset$. Hence, $S\subseteq \dombQ\setminus \Img{g}$.
Using part \ref{item:GLI-reclaimB}) of
Definition~\ref{def:k_guarded_incidence_graphs}
we obtain that $Q'\deff\reclaimB{Q}{S}$ is in $\QGLIk$.
Note that $Q'$ has the parameters
$g_{\tilde{Q}},\domb{\tilde{Q}}, \domr{\tilde{Q}}$ required by the
lemma for formula $\tilde{\chi}$.
Furthermore, as stated in Claim~1, we have
$\mysum{\kLI'}=\hom{Q'}{\kLI'}$ for all $k$-labeled incidence graphs
$\kLI'$ satisfying the lemma's assumptions for $\tilde{\chi}$.
Thus, we have achieved \eqref{eq:Goal:neuerFall}, and therefore
the proof is finished in case that $Z=\emptyset$.
\medskip

\noindent
If $Z\neq\emptyset$, we cannot simply choose $Q'$ as above, because
there is no guarantee that $\reclaimB{Q}{S}$ belongs to $\QGLIk$.
Instead, we proceed as follows.
\\
Let $f$ be the partial function with $\Dom{f}=Z$ and
\[
f(i) \deff \tilde{g}(i) \quad \text{for all $i\in \Dom{f}$.}
\]
\textit{Claim 2:} \ 
1.) \ $\tilde{g}= f\cup g$.\\
2.) \ $f$ is a transition for $g$.\\
3.) \ All
$i\in\Dom{f}$ satisfy: \ $f(i)\in S$ \ or \
$f(i)\not\in \Img{g}$.
\smallskip

\noindent
\textit{Proof of Claim~2:} \\
1.) \ For all $i\in\Dom{f}$ we have: $(f\cup g)(i)=f(i)=\tilde{g}(i)$.\\
For all $i\in\Dom{g}\setminus\Dom{f}$ we have $(f\cup g)(i)=g(i)$. Since $i\not\in \Dom{f}=Z$, we have $\tilde{g}(i)=g(i)$.
I.e., $(f\cup g)=\tilde{g}$.
\smallskip

\noindent
2.) \ Consider an arbitrary $i\in\Dom{g}$ with $g(i)\in\Img{f}$. We have to show that $i\in\Dom{f}$, i.e., $i\in Z$.
Since $g(i)\in\Img{f}$, there is a $j\in Z$ such that $g(i)=f(j)$.
From \eqref{eq:neuerSyntaxfall:tildeg} we know that $\tilde{g}(j)=g(j)$ or
$\tilde{g}(j)\in S$ or
$\tilde{g}(j)\not\in\Img{g}$.
But
\begin{equation}\label{eq:ProofNewCaseStrangeEq1}
  \tilde{g}(j)\ = \ (f\cup g)(j) \ = \ f(j) \ = \ g(i)\,,
\end{equation}
hence
$\tilde{g}(j)\in\Img{g}$.
This leaves us with two remaining cases: $\tilde{g}(j)=g(j)$ or $\tilde{g}(j)\in S$.

In case that $\tilde{g}(j)=g(j)$, from the fact that $j\in Z$ we obtain (by the definition of $Z$) that $g(j)\in S$.
Hence, $g(j)=\tilde{g}(j)=f(j)=g(i)\in S$.
Thus, by our choice of $Z$ we have: $i\in Z=\Dom{f}$.

In case that $\tilde{g}(j)\in S$, by \eqref{eq:ProofNewCaseStrangeEq1} we have:
$\tilde{g}(j)=f(j)=g(i)\in S$. Thus, by our choice of $Z$ we have: $i\in Z=\Dom{f}$.

Hence, we have shown that $f$ is a transition for $g$.
\smallskip

\noindent
3.) \ Consider an arbitrary $i\in\Dom{f}$. By definition we have $f(i)=\tilde{g}(i)$.

In case that $\tilde{g}(i)= g(i)$, by definition of $Z$ we obtain: $g(i)\in S$.
Hence, $g(i)=\tilde{g}(i)=f(i)\in S$.

In case that $\tilde{g}(i)\neq g(i)$, by \eqref{eq:neuerSyntaxfall:tildeg} we know that
$\tilde{g}(i)\in S$ or
$\tilde{g}(i)\not\in\Img{g}$.
Since $f(i)=\tilde{g}(i)$, we hence have: $f(i)\in S$ or
$f(i)\not\in\Img{g}$.

This completes the proof of Claim~2.
\qed$_{\textit{Claim 2}}$
\bigskip

\noindent
Let $D_1\deff (S\cap\Img{f})\setminus \Img{g}$,
let $D_2\deff S\setminus\Img{\tilde{g}}$, and
let $D_3\deff \Img{g}\cap\Img{f}$.
\bigskip

\noindent
\emph{Claim~3:} $D_1,D_2,D_3$ are pairwise disjoint, and $S=\bigcup_{i\in[3]}D_i$.\smallskip

\noindent
\emph{Proof of Claim~3:} \
We first show that $D_1,D_2,D_3$ are pairwise disjoint:
$D_1\cap D_2=\emptyset$ because
$D_1\subseteq\Img{f}\subseteq\Img{\tilde{g}}$.
$D_1\cap D_3=\emptyset$ because $D_3\subseteq \Img{g}$.
$D_2\cap D_3=\emptyset$ because $D_3\subseteq\Img{f}\subseteq\Img{\tilde{g}}$.

We next show that $\bigcup_{i\in[3]}D_i\subseteq S$:
By definition, $D_i\subseteq S$ for each $i\in[2]$. 
For showing that $D_3\subseteq S$ consider an arbitrary $j\in
D_3=\Img{g}\cap\Img{f}$.
Since $j\in\Img{f}$, there is an $i\in\Dom{f}=Z$ such that $j=f(i)$.
From Claim~2, part 3.) we obtain that $f(i)\in S$ or
$f(i)\not\in\Img{g}$. Since $f(i)=j\in\Img{g}$, we hence have
$j\in S$.

Finally, we show that $S\subseteq\bigcup_{i\in[3]}D_i$:
Consider an arbitrary $j\in S$.
If $j\in D_2$, we are done.
If $j\not\in D_2$, then $j\in \Img{\tilde{g}}$.
Hence, there is an $i\in\Dom{\tilde{g}}$ such that $j=\tilde{g}(i)$.
We obtain: $i\in Z$ (the reasoning is as follows: If $g(i)\neq
\tilde{g}(i)$, then $i\in Z$ by \eqref{eq:Def:Z}. Otherwise,
$g(i)=\tilde{g}(i)=j\in S$, and by \eqref{eq:Def:Z} we obtain that
$i\in Z$).
Hence, $i\in Z=\Dom{f}$. Thus $j=f(i)\in\Img{f}$.
If $j\in\Img{g}$ then $j\in D_3$. Otherwise, $j\not\in\Img{g}$, and
hence $j\in D_1$.
\qed$_\textit{Claim~3}$
\bigskip

We let $Q_1\deff \reclaimB{Q}{D_1}$ and
$Q_2\deff \switch{Q_1}{f}$ and
$Q'\deff \reclaimB{Q_2}{D_2}$.
\bigskip

\noindent
\emph{Claim~4:} \ $Q_1,Q_2,Q'\in\QGLIk$ and 
$g_{Q_1}=g$, $\domb{Q_1}=(\dombQ\setminus D_1)$ and 
$g_{Q'}=\tilde{g}$, 
$\domr{Q'}=\indfreeR{\tilde{\chi}}$, $\domb{Q'}=\indfreeB{\tilde{\chi}}$.
\smallskip

\noindent
\emph{Proof of Claim~4:} \
Since $Q\in\QGLIk$ with $g_Q=g$ and $S\subseteq\dombQ$, we have
$D_1\subseteq \dombQ\setminus\Img{g_Q}$. Hence $Q_1=\reclaimB{Q}{D_1}\in\QGLIk$.
Note that $g_{Q_1}=g$ and $\domb{Q_1}=\dombQ\setminus D_1$.

By Claim~2, $f$ is a transition for $g$; hence we have
$Q_2=\switch{Q_1}{f}\in\QGLIk$.
Note that $g_{Q_2}=f\cup g=\tilde{g}$ and
$\domb{Q_2}=\Img{\tilde{g}}\cup (\dombQ \setminus D_1)$.
Note that $(\dombQ\setminus D_1)\supseteq (S\setminus D_1) = D_2\cup D_3$
(by Claim~3). Hence, $\domb{Q_2}\supseteq D_2$. Since
$D_2=S\setminus\Img{\tilde{g}}$ we have $D_2\subseteq \domb{Q_2}\setminus\Img{\tilde{g}}$. 
Hence, $Q'= \reclaimB{Q_2}{D_2}\in\QGLIk$.
Note that $g_{Q'}=\tilde{g}$ and
$\domr{Q'}=\domr{Q}=\indfreeR{\chi}=\indfreeR{\tilde{\chi}}$.

To complete the proof of Claim~4, we show that
$\domb{Q'}=\indfreeB{\tilde{\chi}}$.
Note that $\indfreeB{\tilde{\chi}}=\Img{\tilde{g}}\cup
(\dombQ\setminus S)$.
And $\domb{Q'}=(\domb{Q_2}\setminus D_2) = ((\Img{\tilde{g}}\cup
(\dombQ\setminus D_1))\setminus D_2) =
(\Img{\tilde{g}}\setminus D_2)\cup (\dombQ\setminus (D_1\cup D_2))$
(for the latter, recall from Claim~3 that $D_1\cap D_2=\emptyset$).
Since $D_2 \intersect \Img{\tilde{g}} = \emptyset$ we can simplify even further and get $\domb{Q'} = 
\Img{\tilde{g}}\cup (\dombQ\setminus (D_1\cup D_2))$.
\smallskip

We first show that $\domb{Q'}\subseteq \indfreeB{\tilde{\chi}}$:
Consider an arbitrary $j\in\domb{Q'}$. Then, $j\not\in
D_2$. Furthermore, $j\in\Img{\tilde{g}}$ or $j\in(\dombQ\setminus
D_1)$.

If $j\in\Img{\tilde{g}}$, then $j\in\indfreeB{\tilde{\chi}}$ and we are done.
Otherwise, $j\in(\dombQ\setminus D_1)$. Hence, $j\not\in D_1\cup D_2$.
From Claim~3 we know that $S=D_1\cup D_2\cup D_3$.

In case that $j\not\in D_3$ we have: $j\not\in S$, and hence
$j\in (\dombQ\setminus S)$. Thus, $j\in\indfreeB{\tilde{\chi}}$.

In case that $j\in D_3$ we have: $j\in
D_3\subseteq\Img{f}\subseteq\Img{\tilde{g}}$, and hence
$j\in\indfreeB{\tilde{\chi}}$.

This completes the proof that $\domb{Q'}\subseteq
\indfreeB{\tilde{\chi}}$.
\smallskip

Next, we show that $\indfreeB{\tilde{\chi}}\subseteq \domb{Q'}$:
Consider an arbitrary $j\in \indfreeB{\tilde{\chi}}$.
Then, $j\in\Img{\tilde{g}}$ or $j\in (\dombQ\setminus S)$.
Recall from Claim~3 that $S$ is the union of the pairwise disjoint
sets $D_1,D_2,D_3$.

\emph{Case 1:} $j\in (\dombQ\setminus S)$. Then, $j\in\dombQ$ and
$j\not\in D_1\cup D_2\cup D_3$. In particular, $j\in (\dombQ\setminus
(D_1\cup D_2))\subseteq \domb{Q'}$.

\emph{Case 2:} $j\not\in (\dombQ\setminus S)$. Then,
$j\in\Img{\tilde{g}}$ 
and we are done since $\Img{\tilde{g}} \subseteq \domb{Q'}$.

This completes the proof of Claim~4.
\qed$_{\textit{Claim 4}}$
\bigskip

According to Claim~4, $Q'\in\QGLIk$ and $Q'$ has the desired
parameters $g_{Q'},\domb{Q'},\domr{Q'}$. 
The next claim ensures that $Q'$ indeed satisfies what we aimed for in
\eqref{eq:Goal:neuerFall}.  
\medskip

\noindent
\textit{Claim 5:}
For every $k$-labeled incidence graph $\kLI'$ that satisfies
the lemma's assumptions for $\tilde{\chi}$ we have:
\[
  \hom{Q'}{\kLI'} \ = \ 
  \mysum{\kLI'}.
\]
\textit{Proof of Claim 5:} \
Let $\kLI'$ be an arbitrary $k$-labeled incidence graph $\kLI'$ that satisfies
the lemma's assumptions for $\tilde{\chi}$ (in particular, $\kLI'$ has
real guards w.r.t.\ $\tilde{g}=f\cup g$).
For $i\in[3]$ let $\ell_i\deff |D_i|$.
From Lemma~\ref{lem:HomomCountQ} we obtain:
\\
$\hom{Q'}{\kLI'}=
  \hom{\reclaimB{Q_2}{D_2}}{\kLI'}=$
\[
  \sum_{\tupel{e}_2\in \bV{\I_{\kLI'}}^{\ell_2}} \hom{Q_2}{\reseatB{\kLI'}{D_2}{\tupel{e}_2}}\,.
\]
Furthermore, for every $\tupel{e}_2\in \bV{\I_{\kLI'}}^{\ell_2}$ and for
 $\kLI'_{\tupel{e}_2}\deff\reseatB{\kLI'}{D_2}{\tupel{e}_2}$
we have: \\
$\hom{Q_2}{\kLI'_{\tupel{e}_2}} =
 \hom{\switch{Q_1}{f}}{\kLI'_{\tupel{e}_2}} =$
\[
  \hom{\kLIf}{\kLI'_{\tupel{e}_2}} \ \cdot \!\!
  \sum_{\tupel{e}_3\in \bV{I_{\kLI'}}^{\ell_3}} \hom{Q_1}{\reseatB{\kLI'_{\tupel{e}_2}}{\myB}{\tupel{e}_3}}
\]
where
$\myB\deff\Img{g_{Q_1}}\cap\Img{f}\cap\domb{Q_1} = D_3$ (the latter
follows from Claims~3 and 4).

Since $\kLI'$ has real guards w.r.t.\ $\tilde{g}=f\cup g$, it also has
real guards w.r.t.\ $f$.
Note that $\kLI'_{\tupel{e}_2}$ differs from $\kLI'$ only w.r.t.\ the
blue labels in $D_2=S\setminus\Img{\tilde{g}}$; none of these labels
belongs to $\Img{f}$ (because $\Img{f}\subseteq\Img{\tilde{g}}$). 
Hence, also $\kLI'_{\tupel{e}_2}$ has real
guards w.r.t.\ $f$. This implies that
$\hom{\kLIf}{\kLI'_{\tupel{e}_1}}=1$.

In summary, we obtain that $\hom{Q'}{\kLI'} =$
\[
  \sum_{\tupel{e}_2\in \bV{\I_{\kLI'}}^{\ell_2}} \
  \sum_{\tupel{e}_3\in \bV{\kLI'}^{\ell_3}}
  \hom{Q_1}{\reseatB{\kLI'_{\tupel{e}_2}}{D_3}{\tupel{e}_3}}\,.
\]
Since $D_2,D_3$ are disjoint, we obtain:
$\hom{Q'}{\kLI'} =$
\[
  \sum_{\tupel{e}\in \bV{\I_{\kLI'}}^{\ell_2+\ell_3}} \
  \hom{Q_1}{\reseatB{\kLI'}{(D_2\cup D_3)}{\tupel{e}}}\,.
\]
Recall that $Q_1=\reclaimB{Q}{D_1}$. Hence,
a further application of Lemma~\ref{lem:HomomCountQ} yields:
$\hom{Q'}{\kLI'} =$
\[
  \sum_{\tupel{e}\in \bV{\I_{\kLI'}}^{\ell_2+\ell_3}} \ \
  \sum_{\tupel{e}_1\in\bV{\I_{\kLI'}}^{\ell_1}} 
  \hom{Q}{
    \reseatB{
      \reseatB{\kLI'}{(D_2\cup D_3)}{\tupel{e}}}
      {D_1}
      {\tupel{e}_1}
  }\,.
\]
Note that $S$ is the union of the disjoint sets $D_1,D_2,D_3$ and
$\ell=|S|=\ell_1+\ell_2+\ell_3$. Hence, $\hom{Q'}{\kLI'}=$
\[
  \sum_{\tupel{e}\in \bV{\I_{\kLI'}}^{\ell}} \hom{Q}{\reseatB{\kLI'}{S}{\tupel{e}}}
  \quad = \quad \mysum{\kLI'}
\]
(the last equality is provided by Claim~1).
Hence, the proof of Claim~5 is complete.
\qed$_{\textit{Claim 5}}$

\bigskip
From Claims~4 and 5 we obtain that we have achieved
\eqref{eq:Goal:neuerFall}.
This completes Case~7, and it completes the proof of Lemma~\ref{lemma:FormulaToHom}.\qedhere
\end{enumerate}
\end{proof}
 
\medskip

\subsection{Detailed proof of Theorem \ref{thm:GCkVsHomomIndist}}
\label{appendix:subsec:MainThm}

\begin{proof}[\textbf{\textup{Proof of Theorem~\ref{thm:GCkVsHomomIndist}:}}] \ \smallskip
  
\noindent
Let $\I$ and $\I'$ be arbitrary incidence graphs, and let
$k\in\NNpos$. Our aim is to prove the following:
\[
  \I\not\equivGCk\I'
  \ \ \iff \ \
  \HOM{\IEHW{k}}{\I} \neq \HOM{\IEHW{k}}{\I'}\,.
\smallskip
\]

\noindent
\textbf{Case 1:} \ $|\bVI|\neq |\bV{\I'}|$.
\\
Let $\J$ be the incidence graph with exactly $1$
blue node, no red node, and no edge. Clearly,
$\J\in\IEHWk$. Furthermore,
$\hom{\J}{\I}=|\bV{\I}|$ and $\hom{\J}{\I'}=|\bV{\I'}|$.
As $|\bVI|\neq |\bV{\I'}|$, we obtain that
$\HOM{\IEHW{k}}{\I} \neq \HOM{\IEHW{k}}{\I'}$.

Let $m\deff|\bV{\I}|$ and consider the sentence $\psi\deff
\existsex[m](\vare_1).\vare_1{=}\vare_1$ in $\GCk$. Then, $\I\models\psi$ and
$\I'\not\models\psi$.
Hence, $\I\not\equivGCk\I'$.
Thus, in Case~1 we are done.
\medskip

\noindent
\textbf{Case 2:} \ $|\bVI|=|\bV{\I'}|$.
\medskip

\noindent
For ``$\Longrightarrow$'' let $\phi$ be a sentence
in $\GCk$ such that $\I\models\phi$ and $\I'\not\models\phi$.
Let $\chi \isdef (\top \land \psi) \in \NGCk$ be given according to Theorem~\ref{thm:ngck_is_equivalent}, i.e. it holds that $(\top \land \phi) \equiv (\top \land \psi) = \chi$.

Let $m\deff |\bVI|=|\bV{\I'}|$, and
choose an arbitrary $d\in\NN$ such
that $d\geq |N_{\I}(e)|$ for all $e\in\bVI$ and $d\geq |N_{\I'}(e)|$ for
all $e\in\bV{\I'}$.

Let $Q\deff \quanth{\chi}{m,d}\in\QGLIk$ be the $k$-labeled quantum
incidence graph provided by Lemma~\ref{lemma:FormulaToHom}.
Note that $Q$ has guard function $g_\emptyset$, and
$\dombQ=\domrQ=\emptyset$.
Let $b_\emptyset$ and $r_\emptyset$ the labeling functions with empty
domain, and let
$\kLI\deff (\I,r_\emptyset,b_\emptyset,g_\emptyset)$ and
$\kLI'\deff (\I',r_\emptyset,b_\emptyset,g_\emptyset)$.
They both have real guards w.r.t.\ $g_\emptyset$ (simply because
$g_\emptyset$ has empty domain).
From Lemma~\ref{lemma:FormulaToHom} we know that
$\hom{Q}{\kLI}=1$ (because $\I\models\chi$) and
$\hom{Q}{\kLI'}=0$ (because $\I'\not\models\chi$).
Thus, $\hom{Q}{\kLI}\neq \hom{Q}{\kLI'}$.
Since $Q\in\QGLIk$, there are a degree $\ell\in\NNpos$ and
coefficients
$\alpha_1,\ldots,\alpha_\ell\in\RR$ and components
$\kLI_1,\ldots,\kLI_\ell\in\GLIk$ such that $Q=\sum_{i=1}^\ell
\alpha_i \kLI_i$. We have:
$\sum_{i=1}^\ell \alpha_i \hom{\kLI_i}{\kLI}
 =
 \hom{Q}{\kLI} \neq \hom{Q}{\kLI'}
 = \sum_{i=1}^\ell \alpha_i \hom{\kLI_i}{\kLI'}$.
Hence there must exist an $i\in[\ell]$ with $\alpha_i\neq 0$ and
$\hom{\kLI_i}{\kLI}\neq\hom{\kLI_i}{\kLI'}$.

We know that $\kLI_i\in\GLIk$,
and $\kLI_i=(\I_i,r_i,b_i,g_i)$ has
parameters $\domrQ=\dombQ=\emptyset$ and guard function
$g_i=g_\emptyset$.
I.e., $\Dom{r_i}=\Dom{b_i}=\emptyset$.
From Theorem~\ref{thm:IEHWkIsGLIk}\ref{item:GLIkInIEHWk} we know that
$\I_i\in\IEHWk$. Thus, $\I_i$ witnesses that 
$ \HOM{\IEHW{k}}{\I} \neq \HOM{\IEHW{k}}{\I'}$.
\medskip

\noindent
For ``$\Longleftarrow$'' let $\J$ be an incidence graph in $\IEHWk$
such that $\hom{\J}{\I}\neq\hom{\J}{\I'}$.
Let $m\deff\hom{\J}{\I}$.
From Theorem\ref{thm:IEHWkIsGLIk}\ref{item:IEHWkInGLIk}
we obtain a label-free $\tilde{\kLI}\in\GLIk$ such that
$\I_{\tilde{\kLI}}\isom\J$.

Clearly, for the label-free $k$-labeled incidence graphs $\kLI$ and
$\kLI'$ with $\I_{\kLI}=\I$ and $\I_{\kLI'}=\I'$ we have
$\hom{\tilde{\kLI}}{\kLI}=\hom{\J}{\I}=m$ and
$\hom{\tilde{\kLI}}{\kLI'}=\hom{\J}{\I'}\neq m$.

From Lemma~\ref{lemma:HomToFormula} we obtain a sentence
$\chi\deff \form{\tilde{\kLI}}{m}$ in $\NGCk \subseteq \GCk$ such that
$\myInt{\kLI}\models \chi$ (because $\hom{\tilde{\kLI}}{\kLI}=m$) and
$\myInt{\kLI'}\not\models \chi$ (because
$\hom{\tilde{\kLI}}{\kLI'}\neq m$).
Thus, $\I\models\chi$ and $\I'\not\models\chi$. I.e., $\chi$ witnesses
that $\I\not\equivGCk\I'$.
\\
This completes the proof of Theorem~\ref{thm:GCkVsHomomIndist}.
\end{proof}
 
\end{document}